\documentclass[acmsmall,screen,authorversion,review=false,timestamp=false,natbib=false]{acmart}

\usepackage{anyfontsize}

\newif\ifExtended
\Extendedtrue
\newcommand{\ifextension}[2]{\ifExtended%
#1%
\else%
#2%
\fi}

\newcommand{\new}[2]{#2}
\newcommand{\del}[2]{}
\newcommand{\chg}[3]{#2}
\newcommand{\sidecomment}[2]{}

\usepackage{fancyhdr}

\setcopyright{cc}
\setcctype{by}
\acmDOI{10.1145/3808295}
\acmYear{2026}
\acmJournal{PACMPL}
\acmVolume{10}
\acmNumber{PLDI}
\acmArticle{217}
\acmMonth{6}
\acmSubmissionID{pldi26main-p324-p}
\received{2025-11-13}
\received[accepted]{2026-04-03}

\input{000-DEFINITIONS.tex}

\begin{document}

\title[Heterogeneous Dynamic Logic: Provability Modulo Program Theories]{\texorpdfstring{Heterogeneous Dynamic Logic:\\
Provability Modulo Program Theories}{Heterogeneous Dynamic Logic: Provability Modulo Program Theories}}

\author{Samuel Teuber}
\orcid{0000-0001-7945-9110}
\affiliation{%
  \institution{Karlsruhe Institute of Technology (KIT)}
  \city{Karlsruhe}
  \country{Germany}
}
\email{teuber@kit.edu}

\author{Mattias Ulbrich}
\orcid{0000-0002-2350-1831}
\affiliation{%
  \institution{Karlsruhe Institute of Technology (KIT)}
  \city{Karlsruhe}
  \country{Germany}
}
\email{ulbrich@kit.edu}

\author{André Platzer}
\orcid{0000-0001-7238-5710}
\affiliation{%
  \institution{Karlsruhe Institute of Technology (KIT)}
  \city{Karlsruhe}
  \country{Germany}
}
\email{platzer@kit.edu}

\author{Bernhard Beckert}
\orcid{0000-0002-9672-3291}
\affiliation{%
  \institution{Karlsruhe Institute of Technology (KIT)}
  \city{Karlsruhe}
  \country{Germany}
}
\email{beckert@kit.edu}

\begin{abstract}
\looseness=-1
Formally specifying, let alone verifying, properties of systems involving multiple programming languages is inherently challenging.
We introduce \emph{Heterogeneous Dynamic Logic} (HDL), a framework for
combining reasoning principles from distinct (dynamic) program logics in a modular and compositional way.
HDL mirrors the architecture of satisfiability modulo theories (SMT):
Individual dynamic logics, along with their calculi, are treated as \emph{dynamic theories} that can be combined to reason about heterogeneous systems whose components are verified using different program logics.
HDL provides two key operations:
\emph{Lifting} extends an individual dynamic theory with new program constructs (e.g., the havoc operation or regular programs) and automatically augments its calculus with sound reasoning principles for the new constructs;
and \emph{Combination} enables cross-language reasoning in a \emph{single modality} via \emph{Heterogeneous Dynamic Theories}, facilitating the reuse of existing proof infrastructure.
\new{L5}{By lifting combined theories with regular programs, we obtain \emph{heterogeneous} control structures that allow us to reason about intertwined cross-language behavior.}
We formalize dynamic theories, their lifting and combination, and prove the soundness of all proof rules in Isabelle.
We also introduce a proof rule combining deductive DL-based reasoning with reasoning principles from Kleene Algebras with Tests.
Furthermore, we prove \emph{relative completeness} theorems for lifting and combination:
Under usual assumptions, reasoning about lifted or combined theories is no harder than reasoning about the constituent dynamic theories and their common first-order structure (i.e., the ``data theory'').
We demonstrate HDL's value by verifying an automotive case study where a Java controller (formalized in Java dynamic logic) steers a plant model (formalized in differential dynamic logic).
\end{abstract}



\maketitle

\section{Introduction}
\label{sec:introduction}
\looseness=-1
Since its origins in Hoare logic~\cite{DBLP:journals/cacm/Hoare69} and Dijkstra's \del{CCBY}{weakest precondition }predicate transformers~\cite{Dijkstra_Predicate_Transformers}, formal verification has grown into a mature discipline, supporting a wide range of languages, from bytecode~\cite{DBLP:conf/foveoos/Ulbrich10,DBLP:conf/fm/PaganoniF23}, C~\cite{DBLP:journals/fac/KirchnerKPSY15,DBLP:conf/tacas/ClarkeKL04}, or Java~\cite{KeYBook,DBLP:conf/fm/BlomH14}, to modern languages like Rust~\cite{DBLP:journals/pacmpl/Astrauskas0PS19,DBLP:journals/pacmpl/LattuadaHCBSZHPH23} and even to hybrid programs for cyber-physical systems~\cite{Platzer08}.
While \del{CCBY}{these }individual ``worlds'' enjoy robust verification infrastructures, it remains challenging to \emph{rigorously combine} verification results across languages.
Yet, many modern systems\del{CCBY}{, especially cyber-physical systems,} are inherently heterogeneous and would benefit from \chg{CCBY}{methodologies}{verification methodology} enabling \emph{compositional} reasoning across program logics.
\chg{CCBY}{This}{Crucially, this is not merely a \emph{tooling} challenge, but} raises the fundamental questions of how semantics and proof principles of \del{CCBY}{distinct }program logics can be combined to \emph{specify} and \emph{verify} heterogeneous systems.
In practice, heterogeneous verification today often resorts to encoding one system (e.g., a software controller) in the language of another (e.g., hybrid programs)~\cite{DBLP:conf/iccps/GarciaMP19,DBLP:journals/lites/KamburjanMH22,DBLP:conf/nips/TeuberMP24,DBLP:conf/qsw/KlamrothBSD23}.
However, such encodings are neither seamless nor intuitive:
They obscure structure \del{CCBY}{due to incompatible programming paradigms }and fail to reuse existing \chg{CCBY}{infrastructure}{tools available for the individual languages}.

\paragraph{Dynamic Logic.}
We argue that Dynamic Logic (DL) provides a strong foundation
for reasoning about heterogeneous systems due to its expressiveness and its structural properties.
DL generalizes traditional verification techniques~\cite{DBLP:journals/cacm/Hoare69,Dijkstra_Predicate_Transformers}
by treating programs as first-class citizens of the logic embedded within modalities (rather than residing outside the logic).
This enables DL to specify a large variety of systems including imperative languages~\cite{harel_first-order_1979} or hybrid systems~\cite{Platzer08}.
Crucially, DL's treatment of programs gives rise to a fully compositional logic that
naturally supports properties such as
liveness~\cite{DBLP:journals/fac/TanP21} or incorrectness~\cite{DBLP:conf/tap/RummerS07} as well as hyperproperties (e.g. refinement~\cite{Ulbrich2013,DBLP:conf/lics/LoosP16,prebet_uniform_2024} or secure information flow~\cite{beckert_information_2013}) \emph{without necessitating specialized logics} (as necessary for, e.g., Hoare logic).
However, prior DL frameworks are confined to a \emph{single, fixed programming language}, and are thus ill-suited for verifying heterogeneous systems that combine components from multiple ``worlds'' (e.g. Java code and hybrid systems).
As a lighter-weight counterpart to \del{L12}{(propositional)} DL, Kleene Algebra with Tests~\cite{DBLP:journals/toplas/Kozen97} (KAT) 
is an \del{L12}{established,} effective methodology for (automated) program verification.
\new{L12}{Our approach also integrates with program rewriting via KAT to ease heterogeneous verification.}
\del{L12}{We introduce a proof calculus capable of combining KAT reasoning with deductive DL reasoning.
Section X demonstrates how this combination can be leveraged to ease heterogeneous systems verification.}

\paragraph{Satisfiability Modulo Theories.}
The success of satisfiability modulo theories (SMT) has transformed first-order reasoning across a wide range of applications\chg{S10}{.}{:}
By requiring a \emph{unified set of assumptions} (e.g., stable infiniteness~\cite{DBLP:journals/tcs/Oppen80}) from first-order theories, SMT solvers can delegate theory-specific queries, e.g., involving integers or arrays, to specialized sub-solvers.
Users can freely combine theories within a single formula while theory combination mechanisms ensure soundness and, under suitable assumptions, even completeness.
The seminal results by Shostak, Nelson and Oppen, and their subsequent extensions~\cite{DBLP:journals/tcs/Oppen80,DBLP:journals/toplas/NelsonO79,DBLP:conf/frocos/TinelliH96,TinelliZ04,barrett_satisfiability_2018, Shostak} thus paved the way for the modular combination of first-order theories.
This enables a \emph{separation of concerns} between theory-specific reasoning and the logical structure connecting the theories.
Inspired by this design, we ask: Can a similar, modular structure be applied to program logics?
That is, can we reason about heterogeneous programs while preserving the semantics and proof rules established for individual languages?
This paper answers the question in the affirmative, building on the rigorous foundations of dynamic logic.

\paragraph{Contribution.}
\looseness=-1
This paper introduces \emph{Heterogeneous Dynamic Logic} (HDL): The foundational formal notions that allow
DL to serve as a generalized, SMT-inspired vehicle to combine
different \emph{program logics} (as \emph{dynamic theories}), just like first-order logic serves as a vehicle to combine \del{S1}{of }different \emph{data logics} (as first-order theories).
Heterogeneous Dynamic Logic supports two key operations:
\begin{itemize}
    \item \emph{Lifting} (\Cref{sec:liftedDL}) extends a given \emph{dynamic theory} $\dynamicTheory$ with additional program constructs, such as nondeterministic choice ($\havoced{\dynamicTheory}$, \Cref{sec:liftedDL:havoc}) or the closure over regular programs ($\regular{\dynamicTheory}$, \Cref{sec:liftedDL:regular}).
    Lifting also automatically extends the dynamic theory's calculus with sound proof rules to reason about the extended programming language.
    \item \looseness=-1 \emph{Combination} (\Cref{sec:hdl}) merges
    two \del{ACM}{dynamic }theories $\zero{\dynamicTheory}$ and $\one{\dynamicTheory}$ (defined over \del{ACM}{distinct }programming languages $\zero{\signaturePrograms},\one{\signaturePrograms}$) into a single \emph{heterogeneous dynamic theory}
    $\hero{\dynamicTheory}$.
    This new dynamic theory supports \emph{heterogeneous programs} in which programs from $\zero{\signaturePrograms}$ and $\one{\signaturePrograms}$ can be composed freely using regular programs\chg{ACM}{. Combination also merges and extends the individual proof calculi.}{ without the necessity of embeddings required in prior approaches.
    Combination also merges and extends the calculi of the individual theories into a common, sound calculus for the heterogeneous dynamic theory.}
\end{itemize}
\noindent
\looseness=-1
This enables rigorous, modular reasoning over \emph{heterogeneous programs} with deeply intertwined behavior of programs stemming from individual logics.
Along with a relatively complete proof calculus for lifted/combined dynamic theories,
we \new{L3}{also} present an axiom \chg{S2}{leveraging}{leverages} KAT-style equational rewriting --
coupling KAT's reasoning principles with DL's proof system.
We introduce a formalization of dynamic theories, their lifting and combination and their sound proof calculi in Isabelle in a reusable manner via its Locale infrastructure
(formalized Lemmas/Theorems are in \isaText{cyan}, see \ifextension{\Cref{apx:isabelle}}{Appendix D}).
\new{L3}{%
This formalization directly supports %
interactive theorem proving for heterogeneous systems with proof infrastructure reuse. Simultaneously, the SMT-like architecture of our proof calculus lays the groundwork for %
automated heterogeneous verification in the long term.
}

\paragraph{Overview.}
\looseness=-1
The remainder of this paper is structured as follows:
\Cref{sec:exemplary_case_study_intro} motivates the challenge of heterogeneous verification using a concrete case study.
\Cref{sec:related-work} surveys prior attempts to support general cross-language verification.
\Cref{sec:DL,sec:elementary} introduce our notion of \emph{dynamic theories}, a generalization of dynamic logic with reduced assumptions,
and demonstrate that our assumptions suffice to recover well-known properties of ``DL-like'' logics.
In this context, we also connect DL's proof calculus to equational reasoning via KAT.
\Cref{sec:liftedDL} presents \emph{lifting},
an approach to extend a given dynamic theory with additional program \chg{S16}{features}{featues} such as regular programs along with deriving the necessary proof rules.
To this end, we also prove that our regular program lifting results in a KAT which enables tool reuse.
\Cref{sec:hdl} proves that the \emph{heterogeneous dynamic theory} over two dynamic theories is once again a \emph{dynamic theory}.
Corresponding additional proof rules are derived including additional KAT-style identities which ease verification in practice.
Finally, \Cref{sec:rel-complete} establishes conditions under which \emph{relative completeness} transfers from a dynamic theory to its lifted and/or combined version.
\Cref{case_study_guarantee} showcases the utility of HDL on our automotive case study.
\ifextension{}{Proofs can be found in the extended version on arXiv~\cite{DBLP:journals/corr/abs-2507-08581}.}

\section{Motivation: Case Study}
\label{sec:exemplary_case_study_intro}

\begin{wrapfigure}[17]{r}{0.5\textwidth}
\vspace*{-2em}
\begin{minipage}{\linewidth}
\begin{center}
\includegraphics[width=7cm]{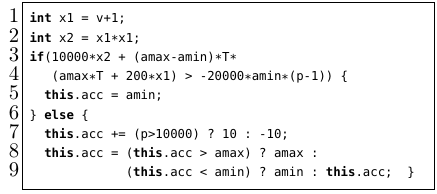}
\end{center}
\vspace*{-1em}%
\subcaption{Java Controller \texttt{ctrl} that can be analyzed using JavaDL.}
\label{lst:background:alpha_ctl}
\end{minipage}

\vspace*{-0.2em}%
\begin{minipage}{\linewidth}
\[
\textit{env} \equiv \Big(\textit{t}:=0; \left(\textit{x}'=\textit{v}, \textit{v}'=\textit{a}, \textit{t}'=1 \ \& \ \textit{t} \leq \textit{T}\right)\Big);
\]

\vspace*{-1.25em}%
\subcaption{Environment model that can be analyzed in \ac{dL}.}
\label{fig:background:dl_model}
\end{minipage}

\vspace*{-0.2em}%
\begin{minipage}{\linewidth}
\begin{center}
\includegraphics[width=7.5cm]{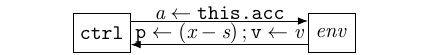}
\end{center}

\vspace*{-0.75em}%
\subcaption{The heterogeneous system ($\leftarrow$ denotes some kind of assignment; semantics see \Cref{case_study_guarantee}).}
\label{fig:heterogeneous_sketch}
\end{minipage}

\vspace*{-1em}
\caption{Heterogeneous system case study}
\end{wrapfigure}%
\looseness=-1
As \chg{S17}{example of a}{an exemplary} heterogeneous system
we consider a car steered by a software controller written in Java while driving towards a stop sign.
To effectively leverage existing proof infrastructures, we wish to model the time-continuous, physical system in Differential Dynamic Logic~\cite{Platzer08} (\acs{dL}) using \ac{keymaerax}~\cite{Platzer08,Platzer2017a,Platzer2020,Fulton2015}\chg{L2}{ %
while verifying the controller using}{.
In contrast, we wish to verify the controller using} Java Dynamic Logic~\cite{DBLP:conf/javacard/Beckert00} (JavaDL) \chg{L2}{using the \KeY{} tool~\cite{KeYBook}}{via the verification tool \KeY{}}.

\paragraph{Java Dynamic Logic.}
\noindent
\looseness=-1
JavaDL allows reasoning about Java programs~\cite{DBLP:conf/javacard/Beckert00,beckert_dynamic_2006,KeYBook} (including exceptions, heap state etc.).
Consider the method in \Cref{lst:background:alpha_ctl} serving as controller for a car approaching a stop sign at distance \texttt{p} with velocity \texttt{v}\chg{S10}{.
The}{:
This} method computes an acceleration via the \emph{stateful} field \texttt{this.acc} and brakes in case the condition in line 3 is satisfied.
The method is not idempotent due to \texttt{this.acc} and its computation also depends on static variables \texttt{amax}, \texttt{amin} and \texttt{T}.
While JavaDL (and \KeY{}) is capable of verifying Java programs,
JavaDL cannot reason about physical dynamics, e.g. given as differential equations.

\paragraph{Differential Dynamic Logic.}
\looseness=-1
Differential Dynamic Logic~\cite{Platzer08} (\acs{dL}) allows reasoning about \emph{hybrid programs} whose states evolve either along discrete state transitions (variable assignments) or continuously along differential equations.
Consider the hybrid program \textit{env} in \Cref{fig:background:dl_model} \chg{S10}{which }{:
It }sets a clock variable \textit{t} to $0$ and evolves \textit{x}, \textit{v} and \textit{t} along a differential equation for a maximum time of \textit{T}.
The program describes the evolution of a car's position under constant acceleration $\textit{a}$.
Using the proof calculus implemented in the theorem prover \ac{keymaerax}~\cite{Platzer08,Platzer2017a,Platzer2020,Fulton2015}, we can\del{L2}{ (only)} prove statements about hybrid programs.
\new{L2}{However, \ac{keymaerax} cannot natively reason over Java programs which, e.g., encompass integer arithmetic with overflows and heap state invariants.} %

\paragraph{Challenge}
Both of the introduced verification methodologies come with strengths and weaknesses\chg{S10}{.}{:}
While they excel\del{S13}{l} at verifying properties for their ``native'' domain of computation, they lack support for their respective counterpart.
However, to verify the safety of the entire heterogeneous system, we require a formalism that allows us to reason about the \emph{looped} heterogeneous system sketched in \Cref{fig:heterogeneous_sketch}.
We propose a \emph{general} framework that gives a formal semantic\new{S19}{s} to this sketch and provides a proof calculus that emphasizes the \emph{\chg{S20}{reuse}{reusage}} of proof infrastructure.
Our framework is not tied to the particular combination of languages discussed in this case study.

\section{Related Work}
\label{sec:related-work}
\looseness=-1
\paragraph{Dynamic Logic.}
As summarized by Ahrendt \emph{et al.}~\cite{Ahrendt2025}, 
DL has been instantiated for a wide range of behavioral languages (e.g., Java~\cite{KeYBook}, hybrid programs~\cite{Platzer08,Platzer2012}, byte code~\cite{DBLP:conf/foveoos/Ulbrich10,Ulbrich2013} or abstract state machines~\cite{DBLP:journals/jucs/SchellhornA97}) allowing proofs of numerous properties (functional correctness~\cite{DBLP:conf/tacas/BeckertSUWW24,DBLP:journals/fac/BoerGKJUW23}, safety~\cite{DBLP:conf/icfem/PlatzerQ09,DBLP:journals/sttt/JeanninGKSGMP17}, liveness~\cite{DBLP:journals/fac/TanP21}, incorrectness~\cite{DBLP:conf/tap/RummerS07}, refinement~\cite{Ulbrich2013,DBLP:conf/lics/LoosP16} or information flow~\cite{beckert_information_2013,DBLP:conf/lics/BohrerP18}).
However, the \emph{combination} of behavioral languages has not previously been studied.

\paragraph{Parameterized Dynamic Logic}
Parameterized Dynamic Logic~\cite{zhang2025parameterizeddynamiclogic}
is a framework designed to build a dynamic logic and a calculus for a
given programming \del{ACM}{or modeling }language based on symbolic execution
rules.
\chg{ACM}{
While their approach also generalizes the concept of dynamic logic, it does not address the challenge
}{
Like HDL, this aims at generalizing the concept of dynamic
logic, but does not address the challenge of integrating logics for
different programming languages }
\new{L13}{to reason about \emph{heterogeneous} programs which is at the core of this paper}.
\del{L13}{However, the ability to reason about \emph{combined, heterogeneous programs} while retaining the proof calculi of the individual logics is at the core of this paper.
Unlike our work, Parameterized DL currently lacks a theorem prover formalization.}

\paragraph{Satisfiability Modulo Theories.}
Research on SMT theory combination~\cite{DBLP:journals/tcs/Oppen80,DBLP:journals/toplas/NelsonO79,DBLP:conf/frocos/TinelliH96,TinelliZ04,barrett_satisfiability_2018} is not readily extensible to modal/dynamic logic: Much of the theory is limited to quantifier-free first-order logic (which is even more restrictive than full first-order (dynamic) logic), and classic SMT theory combination has no way of dealing with a dynamic logic's program primitives.

\paragraph{Intermediate Verification Languages.}
\looseness=-1
To reason about programs or state transition systems across \del{ACM}{programming or modeling }languages, an effective approach is to encode them into a unified intermediate representation\chg{L13}{. }{, transforming them into artifacts within the same formalism. }%
\chg{L13}{This paradigm is used in}{Intermediate representations have long been vital in compiler design and have been extended to}\del{ACM}{ deductive verification via} systems like Boogie~\cite{Boogie}, Why3~\cite{filliatre13esop} and Viper~\cite{MSS16}, which enable language-agnostic reasoning\del{ACM}{ due to their generality}. Similarly, in model checking, common input formats like SMV and  exchange formats such as VMT~\cite{DBLP:journals/fmsd/CimattiGMRT22} and Moxi~\cite{DBLP:conf/cav/JohannsenNDISTVR24} facilitate \chg{ACM}{analyzing}{the analysis of} heterogeneous systems.

In contrast, heterogeneous dynamic logic deliberately avoids assuming a shared \del{L13}{programming }language or logic, aiming to integrate verification systems without a clear common ground. For instance, differential dynamic logic supports\new{L13}{ continuous} value evolution through differential equations -- a capability that cannot be naturally encoded within the languages mentioned earlier.

\paragraph{Fibring.}
In contrast to logic fibring~\cite{gabbay_overview_1996,gabbay1998fibring} (established for modal logics), Heterogeneous Dynamic Logic goes beyond the fibring of ``truth bearing'' entities that one could expect from a DL with multiple modalities. %
\chg{L13}{%
For example, fibring could produce formulas of the form $\phi_1 \implies\left[\alpha_1\right]_1\left[\alpha_2\right]_2\phi_2$ (see \Cref{subsec:hdl:simple}) where $\left[\cdot\right]_1$ and $\left[\cdot\right]_2$ represent modalities for distinct programming languages.
}{%
For example, for formulas $\phi_1 \implies \left[\alpha_1\right]_1\phi_1$ of dynamic logic 1 and $\left[\alpha_2\right]_2\phi_2$ of dynamic logic 2, a suitable definition of fibring could (possibly) produce formulas of the form $\phi_1 \implies\left[\alpha_1\right]_1\left[\alpha_2\right]_2\phi_2$ (see also discussion in Section X).}%
\chg{ACM}{While}{However, while} merging of truth-bearing entities can encode \emph{some} behaviors observed in heterogeneous systems\del{ACM}{ (e.g. sequential composition)}, there is a large class of behaviors that \emph{cannot} be modeled by merely merging truth bearing entities since it ignores the programs.
In particular, we cannot model \del{ACM}{the }looped execution\del{ACM}{ of $\alpha_1$ and $\alpha_2$ which is also essential for our case study (\Cref{sec:exemplary_case_study_intro})}.

\paragraph{UTP and Heterogeneous Bisimulation}
\looseness=-1
The Unified Theories of Programming~\cite{hoare1997unified} (UTP) and its Isabelle formalization\del{ACM}{ Isabelle/UTP}~\cite{DBLP:conf/utp/FosterZW14} provide a common framework for integrating diverse programming languages into a common semantics with \emph{shared state}\del{ACM}{ via relational predicates}.
While this enables proofs on cross-language systems,
it typically requires embedding existing program semantics within UTP.
In contrast, HDL aims for modular\new{L13}{ and sound} \emph{reuse} of existing (external) verification infrastructure via the common interface provided by DL. %
\del{L13}{%
To this end, dynamic theories are granted their individual state spaces with communication explicitly modeled within the heterogeneous dynamic theory.
This approach enables the sound reuse of existing verification infrastructures without the necessity of re-encoding or reimplementing logics:}
For example, our case study\new{L13}{ natively} leverages the existing DL-based proof infrastructures for Java~\cite{KeYBook} and hybrid programs~\cite{Fulton2015} unchanged,
whereas the UTP community has invested considerable effort (notably, Foster's formalisation of hybrid relations~\cite{DBLP:conf/utp/Foster19}) to embed hybrid programs \emph{within} UTP. %
\chg{L13}{%
This highlights HDL's orthogonal philosophy in comparison to UTP. HDL also provides native support for a wide range of properties that can be encoded in DL (see also \cite[Sec. 3]{Ahrendt2025}).}{
This highlights HDL's orthogonal philosophy in supporting
theory combination without sacrificing existing, specialized proof infrastructure. Moreover, by relying on dynamic logic, HDL natively supports a wide range of program properties (see also X). %
}

\looseness=-1
\chg{L13}{The range of properties}{This} also distinguishes HDL from \del{L13}{notions of }heterogeneous (bi)simulation~\cite{DBLP:conf/fossacs/NoraRSW25} 
\chg{L13}{which is}{that are} designed for \del{L13}{relational }\emph{comparison} of \del{L13}{distinct }systems (e.g., a labeled transition system \chg{L13}{vs.\ }{and }a probabilistic system)\del{L13}{, whereas HDL enables verification of a wide range of relational and non-relational first-order properties over arbitrarily composed heterogeneous systems, not merely behavioral comparisons}.

\paragraph{Contract-based\new{L13}{ and Polyglot} Verification.}
\looseness=-1
Our approach provides three key advantages in comparison to contract-based techniques:
First, in our framework all components, all communication, all invariants, and all desired properties can be modeled in a single integral formal system and \chg{L13}{can later be proved}{later proven} in a single formula.
While \del{ACM}{there exist }language agnostic contract languages like Contract-LIB~\cite{DBLP:conf/isola/ErnstPU24} \chg{ACM}{aim to bridge}{\chg{L13}{that bridge}{aiming at bridging}} the specification gap between different PLs, they cannot \chg{L13}{model}{allow modeling} heterogeneous systems in a single \del{L13}{formal }program.
\chg{ACM}{Instead, they require a meta-level argument on combination outside any program logic.}{
Instead, they require a meta-level argument outside the program logic for why a combination is possible.}
Secondly, in its usual formulation\del{ACM}{ (given a precondition, some postcondition holds for all executions)}, contracts \emph{overapproximate} functional behavior and are hence only suitable for correctness/safety verification (box-properties) \chg{L13}{and not liveness/relational properties supported by DL.}{and not reachability/liveness verification (diamond properties, which talk about some executions) or many relational properties supported by our dynamic logic formulation.}
Finally, our relative completeness results in \Cref{sec:rel-complete} ensure that\del{ACM}{, under reasonable assumptions,} \chg{L13}{we can prove all properties of interest using our calculus by reusing proofs for individual components.}{
any correct heterogeneous system can be proven correct via heterogeneous dynamic logic if we can already reason about the system's components. It is unclear under which circumstances contract-based techniques can provide relative completeness guarantees.}

Polyglot verification, specifically PolyVer~\cite{Chen2025PolyVer}, advocates for a contract-based verification approach\del{L13}{ which comes with some of the limitations outlined above}\new{L13}{ for cross-language verification}.
While PolyVer can alleviate the issue of requiring meta-level arguments by providing an automaton-based semantics for the polyglot system, their approach does not provide a relative completeness guarantee.
Hence, it is a priori unclear whether their specification language is sufficiently expressive to prove all safety properties.
In contrast, PolyVer proposes an interesting, \chg{L13}{AI}{Large Language Model}-based \chg{L13}{approach to synthesize auxiliary contracts}{methodology to synthesize intermediate lemmas which are necessary for verification}.
\chg{L13}{Evaluating}{This approach has the potential of increasing the degree of automation in heterogeneous verification -- evaluating} similar techniques for heterogeneous dynamic logic could be an exciting direction of future work.

\paragraph{Kleene Algebras.} %
\looseness=-1%
Kleene Algebra Modulo Theories~\cite{DBLP:conf/pldi/0002BC22} (KAMT) \chg{L13}{is a conceptually related approach that proposes to turn}{proposes an approach to turn} a client theory into a Kleene Algebra with Tests (KAT).
\del{L13}{The objective of this framework is conceptually related to our regular program lifting and product-based combination approach.}
In contrast to KAMT, we derive a \emph{dynamic logic}\new{L10}{ and emphasize \emph{reuse} of proof infrastructure for combined theories}.
\chg{L13}{%
Kozen~\cite{DBLP:journals/toplas/Kozen97} motivated KAT with the observation that}{%
Kozen~\cite{DBLP:journals/toplas/Kozen97} notes as a motivation for the conception of KAT that} %
``many simple program manipulations [...] do not require the full power of [propositional DL]''\chg{L13}{;}{ which underlines that dynamic logics support a much wider range of properties at the cost of lacking decidability (in the first-order case).
This, in particular, also concerns the support for quantifiers, liveness and hyperproperties beyond equality and refinement which are difficult, if not impossible, to encode in KA(M)T.}
\chg{L13}{indeed, dynamic}{As outlined above, dynamic} logics support a \emph{wider} range of properties\new{L13}{, including \emph{quantifiers}, liveness, or properties with \emph{alternating modalities}, which are typically beyond KA(M)T's power.}
While \chg{L13}{DL's}{the} expressivity \chg{L13}{precludes decidability}{of dynamic logics prohibits the derivation of decidability results}, we \chg{L13}{prove}{are able to prove} relative completeness results for our calculus under theory lifting or combination (see \Cref{sec:rel-complete}).
\chg{L13}{Moreover, as}{%
Nonetheless, as we show with} axiom \ref{axiom:E}\chg{L13}{ demonstrates,}{, where applicable,} we can also integrate \chg{L13}{KA(M)T}{KAT/KAMT} results into our reasoning \chg{L13}{where applicable}{while, in general, providing a more powerful logic}.
\new{L10}{In complementary research, Kleene Algebra with Domain~\cite{DBLP:journals/tocl/DesharnaisMS06,DBLP:conf/RelMiCS/Sedlar23,DBLP:journals/corr/abs-2311-06937} extends KAT with (anti-)domain operators to recover propositional DL expressivity, but remains in the propositional, \emph{quantifier-free} setting.}

\del{L13}{
We emphasize differences in comparison to related work here while deferring a more detailed related work discussion to Appendix X.
\emph{Dynamic Logics} have been instantiated for a wide range of behavioral languages and properties~\cite{Ahrendt2025,KeYBook,Platzer08,Platzer2012,DBLP:conf/foveoos/Ulbrich10,Ulbrich2013,DBLP:journals/jucs/SchellhornA97,DBLP:conf/tacas/BeckertSUWW24,DBLP:journals/fac/BoerGKJUW23,DBLP:conf/icfem/PlatzerQ09,DBLP:journals/sttt/JeanninGKSGMP17,DBLP:journals/fac/TanP21,DBLP:conf/tap/RummerS07,Ulbrich2013,DBLP:conf/lics/LoosP16,beckert_information_2013,DBLP:conf/lics/BohrerP18}.
Unlike prior works, we study their common principles and how their semantics can be combined in a uniform fashion.
Unlike \emph{Parameterized Dynamic Logic}~\cite{zhang2025parameterizeddynamiclogic}, we address the challenge of combining programming languages while retaining existing proof calculi and provide a formalization of our results.
In contrast, literature on \emph{Satisfiability Modulo Theories} extensively investigated theory combination~\cite{DBLP:journals/tcs/Oppen80,DBLP:journals/toplas/NelsonO79,DBLP:conf/frocos/TinelliH96,TinelliZ04,barrett_satisfiability_2018}, but is often limited to quantifier-free first-order segments and does not admit program modalities.
While deductive verification systems with \emph{intermediate languages}~\cite{Boogie,filliatre13esop,MSS16} or some \emph{model checking based techniques} with a common exchange format~\cite{DBLP:journals/fmsd/CimattiGMRT22,DBLP:conf/cav/JohannsenNDISTVR24} enable the analysis of heterogeneous systems, we deliberately avoid the need for a common shared programming language. This, e.g., enables the integration of continuous program state transitions as observed in, e.g., differential dynamic logic~\cite{Platzer08,Platzer2012}.
Unlike \emph{Fibring}~\cite{gabbay_overview_1996,gabbay1998fibring}, our combination of logics goes beyond the combination of ``truth bearing'' entities, by equally fibring programming languages and their semantics (see Section X).
While the \emph{Unified Theories of Programming}~\cite{hoare1997unified,DBLP:conf/utp/FosterZW14} provide a common framework for integrating diverse programming languages, they assume shared state and require a reimplementation of program semantics inside the common framework. In contrast, we emphasize reusage of proof infrastructure by reducing verification problems to the individual logics.
In comparison to \emph{contract-based} verification approaches (e.g. CONTRACT-LIB~\cite{DBLP:conf/isola/ErnstPU24} or PolyVer~\cite{Chen2025PolyVer}),
we support a wider range of properties as overapproximating contracts limit application to safety properties.
Additionally, we explicitly formalize the interaction between homogeneous components, i.e. we require no meta-level argument on the possibility or semantics of combination.
Moreover, unlike other approaches, our framework comes with a relative completeness guarantee: Under reasonable assumptions any correct heterogeneous system can be proven correct.
\emph{Kleene Algebra Modulo Theories}~\cite{DBLP:conf/pldi/0002BC22} (KAMT) proposes an approach to turn a client theory into a Kleene Algebra with Tests.
The objective of this framework is conceptually related to our regular program lifting and product-based combination approach.
In contrast to KAMT, we derive a dynamic logic which admits quantifiers and hence provides us with a significantly more expressive logic that supports a wider range of (hyper-)properties due to dynamic logic's compositional nature.
While this expressivity prohibits the derivation of decidability results, we are able to prove relative completeness of our calculus.
As we show with axiom (E), where applicable, we can also integrate KAT/KAMT results into our reasoning while, in general, providing a more powerful logic.
}

\section{Dynamic Theories}
\label{sec:DL}
The foundation of our approach is the idea behind \Ac{FODL}~\cite{harel_first-order_1979,harel_dynamic_2000}, which extends first-order logic to a multi-modal logic where programs parameterize modalities.
For instance, for an integer first-order theory, the formula
\begin{equation}
\label{eq:DL:simple_example}
\underbrace{0 \leq v}_{\text{atom}} \rightarrow
[\underbrace{w \coloneqq v+1}_{\text{program}}]
,\underbrace{1 \leq w}_{\text{atom}}
\end{equation}
states that after all execution of $w \coloneqq v+1$, $w$ is positive if $v$ was non-negative.
Unlike Hoare triples~\cite{DBLP:journals/cacm/Hoare69} or predicate transformers~\cite{Dijkstra_Predicate_Transformers}, Dynamic Logic treats programs as first-class citizens of formulas, enabling compositional reasoning about richer properties such as reachability ($\modDia{w \coloneqq v+1 }{w>0}$) and hyper-properties ($\modBox{\programOne}{\modDia{\programTwo}{\formulaOne}}$).

Our contribution generalizes this principle: not only the first-order theory but also the underlying program language can be chosen freely.
We introduce \emph{dynamic theories}\chg{S10}{ as a}{: A} minimal semantic interface capturing the assumptions on programs and first-order reasoning.
Concrete program logics (existing or new) are instantiations of a dynamic theory; multiple such theories can be composed into a heterogeneous dynamic theory that still satisfies the same interface.
This composability allows ``stacking'' theories while retaining derived proof rules. %
\del{S4}{In Section X we show that this suffices for a relative completeness result: Under reasonable assumptions, two relatively complete calculi for two theories can be merged into one for their heterogeneous combination.}%
By parameterizing over both programs and atoms, our framework unifies diverse program logics and provides a reusable foundation akin to SMT theory combination --- ensuring that compliant dynamic logics can interoperate and inherit proof rules automatically.

\subsection{Syntax}
To focus the presentation on the modal aspects of the logic, we abstract away from the term structure of a Dynamic Logic's first-order fragment and instead begin our investigation directly at the level of first-order atoms. %
\chg{S5}{We formalize the syntactic material of a dynamic theory in the notion of a \emph{dynamic signature}:}{%
Consequently, the syntactic material of a Dynamic Logic is given via the set of variables $\signatureVariables$ the set of its first-order atoms $\signatureAtoms$ and the set of its programs $\signaturePrograms$.
We call this the \emph{dynamic signature} of a dynamic theory:
}
\begin{definition}[Dynamic Signature]
    \label{def:DL:signature}
    A \emph{dynamic signature} $\dynamicSignature$ is given by a set of variables $\signatureVariables$, a set of first-order atoms $\signatureAtoms$, and the set of programs $\signaturePrograms$ with all three sets pairwise disjoint. We denote this as $\dynamicSignature=\dynamicSignatureTuple$
\end{definition}
Given a signature $\dynamicSignature$, we can define the structure of formulas that are part of a dynamic theory in a manner that is similar to classical \ac{FODL} while constraining programs and variables to all syntactical constructs contained in $\dynamicSignature$:
\begin{definition}[$\dynamicSignature$-formulas]
    \label{def:DL:formulas}
Let $\dynamicSignature = \dynamicSignatureTuple$ be a dynamic signature.
Given a variable $\variableOne\in\signatureVariables$, an atom $\atomOne\in\signatureAtoms$, and a program $\programOne\in\signaturePrograms$, $\dynamicSignature$-formulas are defined by the following grammar:
$
    \quad\grammar{\formulaOne,\formulaTwo}{
        \atomOne \grammarOr
        \neg \formulaOne \grammarOr
        \formulaOne \land \formulaTwo \grammarOr
        \fa{\variableOne}{\formulaOne} \grammarOr
        \modBox{\programOne}{\formulaOne}
    }
$
\end{definition}
The set of all $\dynamicSignature$-Formulas is denoted as $\formulaSet{\dynamicSignature}$.
We denote the first-order fragment (i.e. all formulas without any box modality) as $\folFormulaSet{\dynamicSignature} \subset \formulaSet{\dynamicSignature}$.
Further, we introduce syntactic sugar for the diamond modality, which we denote as $\left\langle \programOne \right\rangle \formulaOne \equiv \neg\left[\programOne\right]\neg\formulaOne$.
Similarly, syntactic sugar can be introduced for existential quantifiers and other logical connectives ($\lor$, $\rightarrow$, $\ldots$).

\newcommand{\ring}{\ensuremath{\mathcal{R}}}

\begin{example}[Syntax for Ordered Semirings]
\label{ex:DL:syntax_semiring}
As a running example, consider an ordered semi-ring $\left(\ring,+,\cdot, \leq\right)$ and an arbitrary set of variables $\signatureVariables^{\ring}=\left\{\variableOne_1,\variableOne_2,\dots\right\}$:
We define terms as literals $c\in L \subseteq\ring$, variables $v\in\signatureVariables^{\ring}$ and any composition of terms with $+$ and $\cdot$.
Atoms are any formulas of the form $t_1 \leq t_2$ and programs have the form $\variableOne \coloneqq t_1$ (for $t_1,t_2$ terms and $\variableOne\in\signatureVariables^{\ring}$).
Together, these components form a dynamic signature containing variables, atoms and programs.
\end{example}

\subsection{Semantics}
The previous section outlined the syntactic constructs necessary to define a dynamic theory.
With formulas of our dynamic theory defined,
we can now focus on the semantic constructs that are necessary to define a dynamic theory.
To this end, we require five components:
First, since we are considering a modal logic, we require a state space $\stateSpace$ over which we can evaluate $\dynamicSignature$-formulas.
Secondly, we require functions to evaluate the values of variables ($\variableEval$) w.r.t. some universe ($\universe$), the truth-value of atoms ($\atomEval$), and the state-transition behavior of programs ($\programEval$).
Each of these components must satisfy some conditions.
\chg{}{Once}{As will be shown below, once} these components are defined, we can define the evaluation of formulas.
We now begin by defining state spaces and variable evaluations:
\begin{definition}[Universes, State Spaces and Variable Evaluation]
\label{def:DL:universe_state_space_var_eval}
Given a dynamic signature $\dynamicSignature$ with variables $\signatureVariables$ and a set $\universe \neq \emptyset$ called \emph{universe},
a \emph{state space} $\stateSpace\neq\emptyset$ is defined as a
set of states together with a \emph{variable evaluation function} $\variableEval : \stateSpace \times \signatureVariables \to \universe$ such that $\stateSpace$ and $\variableEval$ satisfy: %
\begin{localeAssmBlock}
\localeAssm{\localeAssmInterpolation}{
For all $\stateOne,\stateTwo \in \stateSpace$ and $\variableSet \subseteq \signatureVariables$,%
there exists an $\stateThree \in \stateSpace$ such that for all $\variableOne \in \signatureVariables$:\\
$\variableEval\mleft(\stateThree,\variableOne\mright) = \variableEval\mleft(\stateOne,\variableOne\mright)$ if $\variableOne \in \variableSet$ and else $\variableEval\mleft(\stateThree,\variableOne\mright)=\variableEval\mleft(\stateTwo,\variableOne\mright)$
}
\end{localeAssmBlock}
\end{definition}
This definition of state spaces slightly deviates from popular definitions of Dynamic Logic that often consider the state space as the set of all mappings from $\signatureVariables$ to $\universe$
which for us corresponds to defining $\variableEval\mleft(\stateOne,\variableOne\mright)=\stateOne\mleft(\variableOne\mright)$.
However, our set-based definition allows for more complex state spaces.
For example, we can globally exclude the possibility of a variable $\variableOne$ taking on certain values $\valueOne\in\universe$, which can be leveraged for typed/sorted state spaces.
Similarly, we can construct state spaces that are formed by composing states from other state spaces (an example of this can be found in \Cref{sec:hdl}).
\new{S69}{Our approach also circumvents formalization hurdles in Isabelle/HOL that would accompany a dependently typed mapping approach for state spaces.}
Importantly, the \newNameForInterpolation{} property retains a notion of well-formedness of the state space\chg{S10}{
 which ensures that if a variable $\variableOne\in\signatureVariables$ is assigned some value $\valueOne\in\universe$, then any other state can also be modified so that $\variableOne$ has that value $\valueOne$.
}{: %
If there exists a state where a variable $\variableOne\in\signatureVariables$ is assigned some value $\valueOne\in\universe$, then any other state can also be modified so that $\variableOne$ has that value $\valueOne$.}
For example, this can ensure that a variable which (in a more fine-grained view of the dynamic theory) is an integer variable can indeed be assigned all integer values independently of the remaining state.
\del{S69}{The }\NewNameForInterpolation{} \del{S69}{property }is important for proving coincidence properties for formulas and programs \chg{S69}{which is crucial for proof rule derivation}{that are crucial for the derivation of the proof rules of dynamic theories}.
\chg{S69}{For}{Going forward, for} states $\stateOne,\stateTwo \in \stateSpace$ we will say $\stateOne$ and $\stateTwo$ are equal on $\variableSet \subseteq \signatureVariables$ (denoted
$\equalOn{\stateOne}{\stateTwo}{\variableSet}$) whenever for all $\variableOne\in\variableSet$ it holds that $\variableEval\mleft(\stateOne,\variableOne\mright)=\variableEval\mleft(\stateTwo,\variableOne\mright)$.

\new{L6}{%
For the semiring dynamic signature from \Cref{ex:DL:syntax_semiring} we can define the universe $\universe=\ring$ and the state space as the set of all mappings $\stateOne:\signatureVariables\to\universe$.
The variable evaluation function then is $\variableEval^{\ring}\left(\stateOne,\variableOne\right)=\stateOne\mleft(\variableOne\mright)$.
\chg{S10}{For example, if $\ring=\mathbb{Z}$, then}{
For example, consider the case of $\ring=\mathbb{Z}$:
In this case,} a variable can take on any integer value.
}

\chg{S6; S65}{A state can also be used for the evaluation of atoms which we define as follows:%
}{%
With state spaces defined, we can now turn to the evaluation of atoms, which requires a function that ensures that the evaluation of atoms depends only on a finite set of variables:}
\begin{definition}[Atom Evaluation]
\label{def:DL:atom_eval}
Given a dynamic signature \dynamicSignature{} with variables \signatureVariables{} and atoms \signatureAtoms{} and a state space $\stateSpace$
we define a \emph{free variable overapproximation} $\freeVarsAtom : \signatureAtoms \to 2^{\signatureVariables}$ and an \emph{atom evaluation function} $\atomEval :\signatureAtoms \to 2^{\stateSpace}$ as any functions satisfying the following properties:
\begin{localeAssmBlock}
\localeAssm{\localeAssmFVAtomOverapprox}{
For any atom $\atomOne\in\signatureAtoms$ and any states $\stateOne,\stateTwo\in\stateSpace$\\
if $\equalOn{\stateOne}{\stateTwo}{\left(\freeVarsAtom\mleft(a\mright)\right)}$ then $\stateOne\in\atomEval\mleft(\atomOne\mright)$ iff $\stateTwo\in\atomEval\mleft(\atomOne\mright)$
}
\localeAssm{\localeAssmFVAtomFinite}{
$\freeVarsAtom\mleft(\atomOne\mright)$ is finite for all atoms $\atomOne \in \signatureAtoms$.
}
\end{localeAssmBlock}
\end{definition}
While the \chg{S22}{coverage}{overapproximation} property ensures that $\freeVarsAtom$ covers all variables that influence the evaluation of a given atom, the finiteness property ensures that atoms are well-behaved in a sense that is classically expected in first-order logic 
(where only variables that syntactically occur in the finite formula can have an impact).
These two properties are once again of particular importance for coincidence properties essential for deriving the dynamic theory proof rules.

\new{L6}{%
Continuing our semiring dynamic signature from \Cref{ex:DL:syntax_semiring}, we define the atom evaluation function that takes a state $\stateOne$ and evaluates an atomic formula (and its resp. terms) w.r.t. $\stateOne$ (i.e. $\stateOne\in\atomEval\mleft(t_1 \leq t_2\mright)$ iff $\stateOne\mleft(t_1\mright) \leq \stateOne\mleft(t_2\mright)$ in $\mathbb{Z}$, where $\stateOne\mleft(t\mright)$ evaluates the variables in term $t$ w.r.t. $\stateOne$).
Concretely, a state $\stateOne$ satisfies the atom $0 \leq v$ iff $\variableEval^{\mathbb{Z}}\mleft(\stateOne,\variableOne\mright)=\stateOne\mleft(\variableOne\mright)\geq\stateOne\mleft(0\mright)= 0$.
The definition of $\freeVarsAtom^{\ring}$ can then be given via straightforward syntactical analysis of terms and formulas.
}

\del{S65}{With the semantics of state spaces and atoms defined, we can now move on to the evaluation of programs.}
Similarly to the semantics of atoms, the semantics of programs also have side conditions on how their behavior is influenced by variables \chg{S47}{which is captured by a program's free variables $\freeVarsProgram{\cdot}$.
The side conditions also provide us with an (implicit) parameterization of programs w.r.t. their variables.
Formally, we require}{. To this end, we define} the program evaluation function to satisfy the following constraints:
\begin{definition}[Program Evaluation]
\label{def:DL:prog_eval}
Given a dynamic signature \dynamicSignature{} with variables \signatureVariables{} and programs \signaturePrograms{} and a state space $\stateSpace$
we define a \emph{free variable overapproximation} $\freeVarsProgram : \signaturePrograms \to 2^{\signatureVariables}$ and a \emph{program evaluation function} $\programEval : \signaturePrograms \to 2^{\stateSpace^2}$ as any functions satisfying the following properties:
\begin{localeAssmBlock}
\localeAssm{\localeAssmFVProgOverapprox}{
For all programs $\programOne\in\signaturePrograms$, all states $\stateOne,\stateTwo,\stateThree\in\stateSpace$, and any set $\variableSet\supseteq \freeVarsProgram\mleft(\programOne\mright)$:\\
if $\equalOn{\stateOne}{\stateTwo}{V}$ and $\left(\stateOne,\stateThree\right) \in \programEval\mleft(\programOne\mright)$, 
then there exists a state $\stateThreeVariant$ such that\\
$\left(\stateTwo,\stateThreeVariant\right)\in\programEval\mleft(\programOne\mright)$
and $\equalOn{\stateThree}{\stateThreeVariant}{V}$
}
\localeAssm{\localeAssmFVProgFinite}{
For all programs $\programOne \in \signaturePrograms$ the result of $\freeVarsProgram\mleft(\programOne\mright)$ is finite.
}
\localeAssm{\localeAssmProgExtensionality}{
For all programs $\programOne\in\signaturePrograms$ and all states $\stateOne,\stateOneVariant,\stateTwo,\stateTwoVariant \in\stateSpace$:\\
if $\equalOn{\stateOne}{\stateOneVariant}{\signatureVariables}$ and $\equalOn{\stateTwo}{\stateTwoVariant}{\signatureVariables}$, then $\left(\stateOne,\stateTwo\right)\in\programEval\mleft(\programOne\mright)$ iff $\left(\stateOneVariant,\stateTwoVariant\right)\in\programEval\mleft(\programOne\mright)$
}
\end{localeAssmBlock}
\end{definition}
The \chg{S22}{coverage}{overapproximation} property establishes that any variable possibly influencing the state transition of a program $\programOne$ must be part of $\freeVarsProgram\mleft(\programOne\mright)$.
This is achieved by ensuring that if two states agree on a superset of $\freeVarsProgram\mleft(\programOne\mright)$, then there must exist ``isomorphic'' state transitions, i.e. state transitions in $\programEval\mleft(\programOne\mright)$ where the resulting states equally agree on this set of variables.
For the remainder of this paper, it is paramount that this property is not just satisfied for $\freeVarsProgram\mleft(\programOne\mright)$, but also for its supersets to ensure equivalent effects on any variables written by $\programOne$.
The second condition ensures that this new set of free variables is finite, too, which is relevant for central coincidence properties.
Finally, the extensionality property ensures that program behavior is only dependent on variable assignment and not intensional differences between states\chg{L6}{ which is also relevant for coincidence.}{.
This is necessary to prove (syntactic and semantic) coincidence properties.}

\new{L6}{
Recall the dynamic signature from \Cref{ex:DL:syntax_semiring},
the program evaluation function for $\variableOne \coloneqq t_1$ induces a transition relation $\left(\stateOne,\stateTwo\right)\in\programEval^{\ring}\mleft(\variableOne\coloneqq t_1\mright)$ iff $\stateTwo\mleft(\variableOne\mright) = \stateOne\mleft(t_1\mright)$
\new{S25}{ and $\equalOn{\stateOne}{\stateTwo}{\left\{\variableOne\right\}^\complement}$ where $\stateOne\mleft(t_1\mright)$ is the evaluation of term $t_1$ in state $\stateOne$} (going forward we denote by $\variableSet^\complement$ the complement set $\signatureVariables\setminus\variableSet$ for $\variableSet\subseteq\signatureVariables$).
$\freeVarsProgram$ is again obtained via syntactic analyses.
}

\chg{S65}{
We summarize the semantics of a dynamic theory as a \emph{domain of computation}, which together with its dynamic signature (syntax) fully specifies the theory:}{
Based on these definitions, we can then define the \emph{domain of computation}.
Together with the underlying syntactic material, this makes up a \emph{dynamic theory}:}
\begin{definition}[Domain of Computation]
\label{def:DL:domain}
Given a dynamic signature \dynamicSignature{}, a \emph{domain of computation} of \dynamicSignature{} (denoted as $\domComp=\domCompTuple$) consists of:
\begin{itemize}
    \item A universe $\universe$ (see \Cref{{def:DL:universe_state_space_var_eval}})
    \item A corresponding state space $\stateSpace$ (see \Cref{{def:DL:universe_state_space_var_eval}})
    \item A variable evaluation function $\variableEval$ (see \Cref{{def:DL:universe_state_space_var_eval}})
    \item An atom evaluation function $\atomEval$ and free variable overapproximation $\freeVarsAtom$ (see \Cref{def:DL:atom_eval})
    \item A program evaluation function $\programEval$ and free variable overapproximation $\freeVarsProgram$ (see \Cref{def:DL:prog_eval})
\end{itemize}
\end{definition}
\begin{definition}[Dynamic Theory]
\label{def:DL:theory}
A \emph{dynamic theory} is the combination of a dynamic signature $\dynamicSignature$ and a corresponding domain of computation $\domComp$.
We denote a dynamic theory as $\dynamicTheory=\dynamicTheoryTuple$.
\end{definition}%
\chg{S65}{%
\noindent
\looseness=-1
The semantics of $\dynamicSignature$-formulas (\del{L6}{see }\Cref{def:DL:formulas}) for a dynamic theory $\dynamicTheory=\dynamicTheoryTuple$ are defined as follows}{
Given a dynamic theory $\dynamicTheory=\dynamicTheoryTuple$, we can then define the semantics of the $\dynamicSignature$-formulas described in Definition X as follows}\del{L6}{ (for $\variableSet\subseteq\signatureVariables$ we denote by $\variableSet^\complement$ the complement set $\signatureVariables\setminus\variableSet$)}:
\begin{definition}[Semantics of Formulas]
\label{def:DL:semantics}
Consider a dynamic theory $\dynamicTheory=\dynamicTheoryTuple$.
We define the semantics of $\dynamicSignature$-formulas as follows:
\begin{align*}
\sem{\atomOne} &= \atomEval\mleft(\atomOne\mright)
&\sem{\neg\formulaOne} &= \stateSpace \setminus \sem{\formulaOne}
&\sem{\formulaOne\land\formulaTwo} &= \sem{\formulaOne} \cap \sem{\formulaTwo}\\
\sem{\fa{\variableOne}{\formulaOne}}&=\rlap{$\left\{\stateOne \in \stateSpace ~\middle|~ \text{for all } \stateTwo \in \stateSpace: \equalOn{\stateOne}{\stateTwo}{\left\{\variableOne\right\}^\complement} \text{ implies } \stateTwo \in \sem{\formulaOne}  \right\}$}\\
\sem{\modBox{\programOne}{\formulaOne}} &=\rlap{$\left\{ \stateOne\in\stateSpace ~\middle|~ \text{for all } \stateTwo \in \stateSpace: \left(\stateOne,\stateTwo\right)\in\programEval\mleft(\programOne\mright) \text{ implies } \stateTwo\in\sem{\formulaOne} \right\}$}
\end{align*}
\end{definition}
We say $\formulaOne$ is \emph{$\dynamicTheory$-valid} (denoted as $\models_{\dynamicTheory} \formulaOne$) iff every state in $\stateSpace$ satisfies $\formulaOne$, i.e. $\stateSpace = \sem{\formulaOne}$.
If a particular state $\stateOne \in \stateSpace$ satisfies a formula $\formulaOne$ (i.e., $\stateOne \in \sem{\formulaOne}$) we denote this as $\stateOne \models_{\dynamicTheory} \formulaOne$ and say $\stateOne$ satisfies $\formulaOne$ or that $\formulaOne$ is $\dynamicTheory$-satisfiable.
For example, the formula $\formulaOne \land \modBox{\programOne}{\formulaTwo}$ is true in a state $\stateOne$ if $\stateOne \models_{\dynamicTheory} \formulaOne$ (i.e. $\stateOne$ satisfies $\formulaOne$) and if for every $\left(\stateOne,\stateTwo\right) \in \sem{\programOne}$ it holds that $\stateTwo \models_{\dynamicTheory} \formulaTwo$ (i.e. every such $\stateTwo$ satisfies $\formulaTwo$).

\new{L6}{%
\looseness=-1
By now we have instantiated the semantics for the semiring dynamic signature $\left(\ring,+,\cdot,\leq\right)$ in \Cref{ex:DL:syntax_semiring}.
\Cref{lem:DL:semiring_dl} shows that our construction yields a dynamic theory $\dynamicTheory^{\ring}$ satisfying all properties of \Cref{def:DL:domain}.
A concrete example for a formula of this theory is \Cref{eq:DL:simple_example}.
Its program $\variableTwo \coloneqq \variableOne + 1$ induces a state transition $\left(\stateOne,\stateTwo\right)\in\programEval^{\mathbb{Z}}\mleft(\variableTwo \coloneqq \variableOne + 1\mright)$ iff $\stateTwo\mleft(\variableTwo\mright) = \stateOne\mleft(\variableOne+1\mright)=\stateOne\mleft(\variableOne\mright)+1$ (and $\equalOn{\stateOne}{\stateTwo}{\left\{\variableTwo\right\}^\complement}$).
It is then easy to see that \Cref{eq:DL:simple_example} is indeed valid in $\dynamicTheory^{\mathbb{Z}}$: For any state with $\stateOne\mleft(\variableOne\mright)\geq 0$ and any state transition to $\stateTwo$ such that $\stateTwo\mleft(\variableTwo\mright) = \stateOne\mleft(\variableOne\mright)+1$ (and otherwise equal) it holds that $\stateTwo\mleft(\variableTwo\mright) \geq 1$.
}

\paragraph{Relation to classical First-Order Dynamic Logic.}
The syntax and semantics of classical first-order dynamic logic~\cite{harel_dynamic_2000}, are more restrictive in the supported structure of atoms and state spaces\chg{S10}{.}{:}
State spaces are assumed to be simple valuations~\cite[p. 293]{harel_dynamic_2000} (and not more complex constructs admitting the evaluation of variables and atoms), and atoms are defined as predicates over terms~\cite[p. 103]{harel_dynamic_2000}.
We demonstrate that axioms and proof rules derivable for classical, interpreted first-order dynamic logic are equally derivable using our smaller set of assumptions;
Moreover, this relaxation, in turn, has two advantages:
\begin{enumerate*}
    \item The designer of a program logic, formalized as dynamic theory, can focus on the particularities of the program logic at hand and obtain extensions like regular programs or havoc operators for free (lifting, see \Cref{sec:liftedDL}),
    \item Program logics formalized as dynamic theories can be combined arbitrarily in a manner that admits compositional reasoning (heterogeneous dynamic theories, see \Cref{sec:hdl}).
\end{enumerate*}

\subsection{Static Semantics}
For deriving proof rules, we need a precise notion of semantically free and bound variables, i.e. those that affect or are affected by a formula or program. \Cref{def:DL:domain} already assumes overapproximations $\freeVarsAtom$\chg{S8}{ and }{/}$\freeVarsProgram$ of these sets that are typically defined syntactically\chg{S10}{. For example, in}{: In} \Cref{ex:DL:syntax_semiring}, variables in expressions are considered free, and those on the left of assignments are bound.
Such approximations are not exact: For instance, $x \coloneqq x + y - y$ changes no state and hence has syntactically, but not semantically, free or bound variables.
Exact computation of these sets is undecidable~\cite{rice1953classes}.
As we have no knowledge on atom and program structure, syntactic approximations are unavailable and we instead adopt the semantic approach from Platzer~\cite{Platzer2017a}:
\begin{definition}[Free and Bound Variables]
\label{def:DL:free_bound_vars}
    Let \programOne{} be some program and $\formulaOne{}$ be some formula over a dynamic theory $\dynamicTheory{}=\dynamicTheoryTuple$.
    Free variables \freeVarsSem{} and bound variables \boundVarsSem{} are defined as follows:
    \begin{align*}
        \freeVarsSem\mleft(\formulaOne\mright) = &\left\{\variableOne \in \signatureVariables~\middle|~
            \text{exists }
            \stateOne,\stateOneVariant\in\stateSpace
            \text{ s.t. }
            \equalOn{\stateOne}{\stateOneVariant}{\left\{\variableOne\right\}^\complement}
            \text{ and }
            \stateOne \in \sem{\formulaOne} \not\ni \stateOneVariant
        \right\}\\
        \freeVarsSem\mleft(\programOne\mright) = &\left\{\variableOne \in \signatureVariables~\middle|~
            \text{exists }
            \stateOne,\stateOneVariant,\stateTwo\in\stateSpace
            \text{ s.t. }
            \equalOn{\stateOne}{\stateOneVariant}{\left\{\variableOne\right\}^\complement}
            \text{ and }
            \left(\stateOne,\stateTwo\right) \in \programEval\mleft(\programOne\mright)
            \right.\\
             &\left.\hphantom{\left\{\variableOne \in \signatureVariables~\middle|\right.}
             \text{but there is no } \stateTwoVariant\text{ s.t. }
             \equalOn{\stateTwo}{\stateTwoVariant}{\left\{\variableOne\right\}^\complement}
             \text{ and }
             \left(\stateOneVariant,\stateTwoVariant\right) \in \programEval\mleft(\programOne\mright)
        \right\}\\
        \boundVarsSem\mleft(\programOne\mright) = &\left\{\variableOne \in \signatureVariables~\middle|~
        \text{exists }
        \stateOne,\stateOneVariant \in \stateSpace
        \text{ s.t. }
        \left(\stateOne,\stateOneVariant\right) \in \programEval\mleft(\programOne\mright)\text{ and }
        \variableEval\mleft(\stateOne,\variableOne\mright) \neq \variableEval\mleft(\stateOneVariant,\variableOne\mright)
        \right\}
    \end{align*}
\end{definition}
We denote $\variablesOfSem{\programOne} = \freeVarsSem\mleft(\programOne\mright) \cup \boundVarsSem\mleft(\programOne\mright)$.
Going forward we will only use $\boundVarsSem\mleft(\programOne\mright)$ and $\variablesOfSem{\programOne}$ which combines both free and bound variables.
The definition of free variables breaks down for atoms
depending on an infinite set of variables (e.g., consider some atom that is true iff an infinite number of variables have value 1).
However, we avoid this issue by assuming that $\freeVarsAtom$ is finite for all atoms.
Note that $\freeVarsAtom\mleft(\atomOne\mright)$ and $\freeVarsProgram\mleft(\programOne\mright)$ resp. overapproximate $\freeVarsSem\mleft(\atomOne\mright)$ and $\freeVarsSem\mleft(\programOne\mright)$ (see \ifextension{\Cref{{isaLem:fv_soundness}}}{{\isaText{Lemma}~C.1}}).
\begin{textAtEnd}
\begin{isabelleLemma}{Soundness of Free Variables}{\texttt{dynLogCore.dyn\_FV\_fml\_sound} and \texttt{dynLogCore.dyn\_FV\_prog\_sound}}
\label{isaLem:fv_soundness}
For any atom $\atomOne\in\signatureAtoms$ and $\programOne\in\signaturePrograms$:
\begin{align*}
     \freeVarsSem\mleft(\atomOne\mright) &\subseteq \freeVarsAtom\mleft(\atomOne\mright)
    &\freeVarsSem\mleft(\programOne\mright) &\subseteq \freeVarsProgram\mleft(\programOne\mright)
\end{align*}
\end{isabelleLemma}
\end{textAtEnd}

\subsection{Instantiations of Dynamic Theories}
\looseness=-1
\chg{S26}{We have now}{Section X} introduced an elaborate definition of dynamic theories.
While the proof-theoretic value of these definitions will become clear in the subsequent sections through the derivability of numerous proof rules from the proposed assumptions,
we should first examine whether these definitions are sensible.
To this end, we prove that three Dynamic Logics are instantiations of our definitions:
\begin{isabelleLemma}{PDL as Dynamic Theory}{pdlDLFull}
\label{isaLem:pdl_dyn_theory}
    Propositional Dynamic Logic over arbitrary atomic programs and propositional atoms is an instantiation of a dynamic theory.    
\end{isabelleLemma}
\begin{isabelleLemma}{Semi-Ring First-Order Dynamic Logic as Dynamic Theory}{semiring\_dl\_full}
\label{lem:DL:semiring_dl}
    Consider first-order atoms (equality and less than) over terms of an ordered semi-ring (e.g. natural numbers, see \Cref{ex:DL:syntax_semiring}). The dynamic logic with programs of the form $\variableOne \coloneqq t$ where $t$ is a term of the semi-ring is an instantiation of a dynamic theory.
\end{isabelleLemma}
\begin{proof}[Proof Sketch]
To prove that Semi-Ring First-Order Dynamic Logic is a dynamic theory, we prove that the state space, variables, atom and program evaluation, and their respective free variable definitions (see \Cref{ex:DL:syntax_semiring}) satisfy the six requirements from \Cref{def:DL:universe_state_space_var_eval,def:DL:atom_eval,def:DL:prog_eval}.
As this logic's state space is the set of all (unique) variable to universe mappings, \textsc{\localeAssmInterpolation{}} and \textsc{\localeAssmProgExtensionality{}} is trivially satisfied. \textsc{\localeAssmFVAtomOverapprox{}} and \textsc{\localeAssmFVProgOverapprox{}} can be achieved by constructing \freeVarsAtom{} and \freeVarsProgram{} as sound syntactical analyses.
\textsc{\localeAssmFVAtomFinite{}} and \textsc{\localeAssmFVProgFinite{}} are achieved through the syntactic construction of atoms/programs which only touch a finite set of variables.
\end{proof}
\begin{isabelleLemma}{Differential Dynamic Logic as Dynamic Theory}{dl\_dyn}
\label{lem:DL:diffDL_theory}
The variables, atoms, and programs of Differential Dynamic Logic \del{S27}{with semantics }as formalized in the \chg{S27}{Archive of Formal Proofs}{AFP}~\cite{Differential_Dynamic_Logic-AFP} form a dynamic theory.
\end{isabelleLemma}

Consequently, all proof rules and results described below hold for these three dynamic theories.
While the logic from \Cref{lem:DL:semiring_dl} might at first seem overly simplistic, we show in \Cref{sec:liftedDL} that we can lift this theory to support a wider range of programs.
As discussed in \Cref{sec:introduction}, this paper is not only about existing dynamic theories, but also about the ability to construct \emph{further} dynamic theories which are, by design, interoperable with the collection of existing dynamic theories.

\subsection{Proof Calculi}
From a program verification perspective, the objective of dynamic logics, and dynamic theories, is to \chg{S9}{prove}{check} that a program $\programOne$ satisfies \new{S9}{formalized (functional) }properties.
For example, we might want to verify that for all inputs satisfying $\formulaOne$, the program $\programOne$ satisfies the postcondition $\formulaTwo$.
From a proof-theoretic perspective this corresponds to proving that dynamic theory formula $\formulaOne \implies \modBox{\programOne}{\formulaTwo}$ is $\dynamicTheory$-valid, i.e.\ that $\models_{\dynamicTheory} \formulaOne \implies \modBox{\programOne}{\formulaTwo}$.
This is usually achieved by \emph{proof calculi} that, from axioms and proof rules, derive that a certain formula is valid.

For the purpose of this work, we consider a calculus to be a relation $\calculus[\dynamicTheory] \in 2^{\formulaSet{\dynamicSignature}} \times \formulaSet{\dynamicSignature}$ where $\left(\Gamma,\formulaOne\right)\in\;\calculus[\dynamicTheory]$ (denoted $\Gamma \calculus[\dynamicTheory] \formulaOne$ or $\calculus[\dynamicTheory] \formulaOne$ iff $\Gamma=\emptyset$) iff $\formulaOne$ can be derived from $\Gamma$ using the proof rules and axioms of the calculus $\calculus[\dynamicTheory]$.
We are only interested in \emph{sound} proof calculi:
\begin{definition}[Sound Calculus]
For a dynamic theory $\dynamicTheory$ a proof calculus $\calculus[\dynamicTheory]$ is called \emph{sound} iff for all formulas $\formulaOne\in\formulaSet{\dynamicSignature}$ it holds that $\calculus[\dynamicTheory] \formulaOne$ implies $\models_{\dynamicTheory} \formulaOne$.
\end{definition}%
\looseness=-1%
\chg{S66}{Soundness}{In the present context, soundness} is a compositional property\chg{S10}{, i.e.\ a}{:
A} proof calculus is sound if all proof rules are sound and all instantiations of axioms are valid.
Soundness is thus an inductive consequence of the soundness of individual calculus steps, because the calculus iteratively generates (valid) axiom instantiations and applies (sound) proof rules.
\chg{S62}{As common in dynamic logics~\cite{beckert_dynamic_2006,Platzer08,DBLP:conf/foveoos/Ulbrich10} and their practical implementations~\cite{Fulton2015,KeYBook},}{
For} the remainder of this work\chg{S62}{ assumes}{, we assume} that all (elementary) dynamic theories \dynamicTheory{} are equipped with a sound proof calculus $\calculus[\dynamicTheory]$ that is complete w.r.t. propositional and uninterpreted first-order reasoning (e.g., by extending the Hilbert or Sequent calculus).
\Cref{sec:elementary} extends this calculus with additional proof rules that hold independently of a particular dynamic theory.
\Cref{sec:liftedDL,sec:hdl} extend and combine dynamic theories and provide sound calculi building on the given calculus for $\dynamicTheory$.
This approach guarantees compositionality\chg{S10}{.}{:}
Although our theories become more expressive as we lift and merge them with the proposed constructs, we can nonetheless reuse the calculus $\calculus[\dynamicTheory]$ of the original theory.

\paragraph{Equational Reasoning.}
In addition to proof rule based reasoning, our approach also supports reasoning via a refinement~\cite{DBLP:conf/lics/LoosP16} and an equivalence relations between programs which can be proven via syntactic rewriting using techniques from Kleene Algebra with Tests~\cite{DBLP:journals/toplas/Kozen97} (KAT):
\begin{definition}[\del{L12}{Program }Refinement and Equivalence]
\label{def:DL:refinement_equivalence}
A program $\programOne$ \emph{refines} $\programTwo$ (denoted as $\katRef{\programOne}{\programTwo}$) iff for all $\left(\stateOne,\stateTwo\right)\in\programEval\mleft(\programOne\mright)$ it holds that $\left(\stateOne,\stateTwo\right)\in\programEval\mleft(\programTwo\mright)$.
\chg{L12}{%
Programs are \emph{equivalent} (\katEq{\programOne}{\programTwo}) iff $\katRef{\programOne}{\programTwo}$ and $\katRef{\programTwo}{\programOne}$.
}{
We say two programs are equal (denoted \katEq{\programOne}{\programTwo}) iff $\katRef{\programOne}{\programTwo}$ and $\katRef{\programTwo}{\programOne}$.}
\end{definition}

\del{L12}{
Throughout this work we will introduce known, proven identities $\katEq{\cdot}{\cdot}$ which can then be used to justify the application of axiom \ref{axiom:E} introduced in the next section.
Additionally, we show that our regular program lifting (Section X) forms a Kleene Algebra with Tests~\cite{DBLP:journals/toplas/Kozen97} (KAT) over \katRef{\cdot}{\cdot},\katEq{\cdot}{\cdot} --- making it ammenable to rewriting and normalization based tooling for KAT.
}

\section{Elementary Results and Proof Rules for Dynamic Theories}
\label{sec:elementary}
\begin{figure}[t]
    \centering
    \fbox{\begin{minipage}[t]{0.97\linewidth}\vspace{0pt}
    \begin{minipage}{0.36\linewidth}
    \begin{axiomBlock}
        \refAxiom{G}{G}{
        $\dfrac{\calculus[\dynamicTheory]\schemaVarOne}{\calculus[\dynamicTheory]\modBox{\schemaProgramOne}{\schemaVarOne}}$
        }
        \refAxiom{E}{E}{
        $\modBox{\schemaProgramOne}{\schemaVarOne}
        \iff
        \modBox{\schemaProgramTwo}{\schemaVarOne}$
        \hfill $(\katEq{\schemaProgramOne}{\schemaProgramTwo})$
        }
    \end{axiomBlock}
    \end{minipage}\begin{minipage}{0.61\linewidth}
    \begin{axiomBlock}
        \refAxiom{V}{V}{
        $
        \schemaVarOne \implies \modBox{\schemaProgramOne}{\schemaVarOne}
        $
        \hfill ($\freeVarsSem\mleft(\schemaVarOne\mright) \cap \boundVarsSem\mleft(\schemaProgramOne\mright) = \emptyset$)
        }
        \refAxiom{B}{B}{
        $
        \left(\fa{\variableOne}{\modBox{\schemaProgramOne}{\schemaVarOne}}\right)
        \iff
        \left(\modBox{\schemaProgramOne}{\fa{\variableOne}{\schemaVarOne}}\right)
        $
        \hfill
        ($\variableOne \notin \variablesOfSem{\schemaProgramOne}$)
        }
        \refAxiom{K}{K}{
        $
        \left(\modBox{\schemaProgramOne}{\left(\schemaVarOne \implies \schemaVarTwo\right)}\right)
        \implies
        \left(
        \modBox{\schemaProgramOne}{\schemaVarOne}
        \implies
        \modBox{\schemaProgramOne}{\schemaVarTwo}
        \right)
        $
        }
    \end{axiomBlock}
    \end{minipage}
    \end{minipage}}\vspace*{-0.5em}
    \caption{Elementary axioms and proof rules for dynamic theories}
    \label{fig:elementary-axioms}
\end{figure}
Concerning the semantic definitions of free and bound variables,
we observe that the definitions from \Cref{def:DL:free_bound_vars} exactly characterize the variables determining the valuation of a formula as well as the \chg{S51}{variables changed by the}{footprint of a} program.
\chg{S51}{All}{Importantly, all} three definitions come with a minimality guarantee that ensures we are not overly conservative in our estimation of which variables impact valuations:
\vspace*{-0.4cm} %
\begin{isabelleLemma}{Coincidence Lemma for Formulas \& Minimality}{dyn\_fv\_sem\_coincidence,dyn\_fv\_sem\_coincidence\_smallest}
\label{isaLem:coincidence_formulas}
\chg{S52}{Given a dynamic theory $\dynamicTheory$ 
and formula $\formulaOne$ define $\variableSet = \freeVarsSem\mleft(\formulaOne\mright)$.
Then $\formulaOne$ has the same truth value in any two states agreeing on $\variableSet$ (i.e., if $\stateSpace \ni \equalOn{\stateOne}{\stateOneVariant}{\variableSet} \in \stateSpace$ then $\stateOne \in \sem{\formulaOne}$ iff $\stateOneVariant \in \sem{\formulaOne}$).
$\variableSet$ is the smallest set with this \emph{coincidence property}.
}{
For any dynamic theory $\dynamicTheory$ 
a formula $\formulaOne$ has the same truth value in any two states that agree on $\freeVarsSem\mleft(\formulaOne\mright)$ (i.e., for $\variableSet = \freeVarsSem\mleft(\formulaOne\mright)$, if $\stateSpace \ni \equalOn{\stateOne}{\stateOneVariant}{\variableSet} \in \stateSpace$ then $\stateOne \in \sem{\formulaOne}$ iff $\stateOneVariant \in \sem{\formulaOne}$).
The set of semantically free variables ($\freeVarsSem$) is the smallest set with this \emph{coincidence property}.}
\end{isabelleLemma}
\begin{isabelleLemma}{Coincidence Lemma for Programs \& Minimality}{dyn\_prog\_fv\_sem\_coincidence,dyn\_prog\_fv\_sem\_coincidence\_smallest}
\label{isaLem:coincidence_programs}
\chg{S52}{%
Given a dynamic theory $\dynamicTheory$ with program $\programOne$ consider some $\variableSet \supseteq \freeVarsSem\mleft(\programOne\mright)$.
If for states $\stateOne,\stateTwo\in\stateSpace$ only differing on a finite set of variables (i.e., $\left\{v \in \signatureVariables~\middle|~\variableEval\mleft(\stateOne,\variableOne\mright) \neq \variableEval\mleft(\stateTwo,\variableOne\mright)\right\}$ is finite) it holds that $\equalOn{\stateOne}{\stateTwo}{\variableSet}$
and $\left(\stateOne,\stateOneVariant\right)\in\programEval\mleft(\programOne\mright)$,
then there exists a $\stateTwoVariant\in\stateSpace$ such that $\left(\stateTwo,\stateTwoVariant\right)\in\programEval\mleft(\programOne\mright)$ and $\equalOn{\stateOneVariant}{\stateTwoVariant}{\variableSet}$.
$\freeVarsSem\mleft(\programOne\mright)$ is the smallest set with this property.
}{%
For any dynamic theory $\dynamicTheory$, consider states $\stateOne,\stateTwo,\stateOneVariant\in\stateSpace$ such that $\stateOne$ and $\stateTwo$ only differ on a finite set (i.e., $\left\{v \in \signatureVariables~\middle|~\variableEval\mleft(\stateOne,\variableOne\mright) \neq \variableEval\mleft(\stateTwo,\variableOne\mright)\right\}$ is finite).
If $\equalOn{\stateOne}{\stateTwo}{\variableSet}$ for $\variableSet \supseteq \freeVarsSem\mleft(\programOne\mright)$ and $\left(\stateOne,\stateOneVariant\right)\in\programEval\mleft(\programOne\mright)$,
then there exists a $\stateTwoVariant\in\stateSpace$ such that $\left(\stateTwo,\stateTwoVariant\right)\in\programEval\mleft(\programOne\mright)$ and $\equalOn{\stateOneVariant}{\stateTwoVariant}{\variableSet}$.
The semantically free variables of a program ($\freeVarsSem$) are the smallest set with this property.}
\end{isabelleLemma}
\begin{isabelleLemma}{Bounded Effect \& Minimality}{dyn\_prog\_bv\_sem\_bounded,dyn\_prog\_bv\_sem\_bounded\_smallest1,dyn\_prog\_bv\_sem\_bounded\_smallest2}
\label{isaLem:bounded_effect_progams}
\chg{S52}{%
Given a dynamic theory $\dynamicTheory$ and a program $\programOne$ define $\variableSet = \boundVarsSem\mleft(\programOne\mright)$.
For all $\left(\stateOne,\stateTwo\right) \in \programEval\mleft(\programOne\mright)$ it holds that $\equalOn{\stateOne}{\stateTwo}{\variableSet^\complement}$.
$\variableSet$ is the smallest set with this property.
}{%
For any dynamic theory $\dynamicTheory$, the set $\boundVarsSem\mleft(\programOne\mright)$ is the minimal set with the \emph{bounded effect} property:
For all $\left(\stateOne,\stateTwo\right) \in \programEval\mleft(\programOne\mright)$ it holds that $\equalOn{\stateOne}{\stateTwo}{\left(\boundVarsSem\mleft(\programOne\mright)\right)^\complement}$.%
}
\end{isabelleLemma}
\looseness=-1
\chg{L5}{%
We will now derive Hilbert-style proof calculus rules (extending the usual first-order axioms) which hold \emph{independently} of any particular dynamic theory.
In particular, these rules also hold for \emph{combined}, heterogeneous dynamic theories (see \Cref{sec:hdl}).
}{%
Based on these results, we can now begin the derivation of a proof calculus for the validity of $\dynamicSignature$-formulas.
To this end, we propose a Hilbert-style calculus with the usual first-order axioms.}
Since we do not know the inner workings of programs $\programOne \in \signaturePrograms$ it is the obligation of the designer of a dynamic theory to provide a suitable calculus to decompose and prove statements about such atomic programs.
However, we can nonetheless provide calculus rules that hold independently of the concrete programs.
A first set of these axioms and rules can be found in \Cref{fig:elementary-axioms}.
These rules will be extended in the subsequent section by calculus rules 
for reasoning about the closure of regular programs over the programs in $\signaturePrograms$.
By providing our calculus rules, we propose a \emph{compositional} calculus\chg{S10}{.}{:}
The theory designer can focus on proof procedures for statements about (from our perspective) atomic programs, while our rules provide the general, compositional principles that are independent of the particular programming language at hand.
To this end, we formalize the soundness of the described axioms:
\begin{isabelleTheorem}{Soundness of Elementary Proof Rules}{dyn_axiom_G,dyn_axiom_K,dyn_axiom_V,dyn_axiom_B}
\label{thm:elementary:soundness}
The proof rules and axioms in \Cref{fig:elementary-axioms} are sound w.r.t. to any dynamic theory \dynamicTheory.
\end{isabelleTheorem}

\looseness=-1
The proof rule \ref{axiom:G} and the axiom \ref{axiom:V} are akin to first-order Hilbert-calculus generalization rules by allowing us to wrap formulas into modalities.
Barcan's axiom \ref{axiom:B} encodes a constant domain property for the dynamic theory's state space that emerges from the state space properties outlined in \Cref{sec:DL}.
\ref{axiom:V} allows us to eliminate programs that do not affect the postcondition while axiom \ref{axiom:K} allows distributing modalities across implications.
Finally, \ref{axiom:E} allows us to swap programs in modalities for equivalent programs.
As described in \Cref{thm:elementary:soundness}, we can prove the soundness of these axioms independently of the considered state space, atomic formulas, and programs.
\new{L4}{%
While \Cref{fig:elementary-axioms} presents the axioms as requiring \emph{semantic} side conditions, these can be discharged via syntactic algorithms.
For example, the sets $\freeVarsSem$ and $\boundVarsSem$ can be over-approximated using syntactic analyses (e.g., \Cref{def:DL:prog_eval} establishes an overapproximation of $\freeVarsSem$ for programs)
Equivalence results ($\katEq{\alpha}{\beta}$) can be obtained using sound KAT~\cite{DBLP:journals/toplas/Kozen97} rewriting rules or natively within dynamic logic~\cite[Sec. 7]{DBLP:conf/lics/LoosP16}.
}

\section{Lifting Dynamic Theories}
\label{sec:liftedDL}
In this section, we demonstrate, for the first time, the power of the concise formalization of necessary properties for the construction of a dynamic theory\chg{S10}{.}{:}
Given a dynamic theory $\dynamicTheory$, we can \emph{extend} its functionality, e.g., by extending the structure of state spaces, atoms, or programs, with \emph{common} constructs.
By proving that this extension \emph{again} satisfies the properties outlined in \Cref{sec:DL}, all properties, proof rules, and axioms are lifted to the extended, more expressive, dynamic theory.
Similarly to conservative extensions in first-order theory, our lifting ensures that previously valid axioms and formulas remain sound.
Unlike conservative extensions, some liftings of dynamic theories meaningfully extend the expressiveness of the logic's programming language beyond the original theory's expressiveness by allowing new state transitions that were previously impossible.
The most prominent example of this is the introduction of loops in \Cref{sec:liftedDL:regular}\new{L11}{ which lifts a theory to its \emph{regular closure} over programs}.
\chg{L11}{%
Before discussing regular closures, 
we first propose a lifting which does \emph{not} increase expressivity but serves as a gentle introduction of all concepts reused later.
To this end we extend a dynamic theory by the havoc operator $\variableOne \coloneqq *$ which assigns the variable $\variableOne$ to an arbitrary value (\Cref{sec:liftedDL:havoc})
}{%
In Section X, we begin by illustrating this principle for the havoc operator $\variableOne \coloneqq *$ which assigns the variable $\variableOne$ to an arbitrary value.
Subsequently, in Section X, we apply this approach to a more ambitious extension: %
The \emph{closure} of regular programs over $\signaturePrograms$:
In this instance, we provide axioms for the decomposition of all regular program operators, including loop invariants and variants (the latter under minimal assumptions about the availability of natural number constraints).}

\subsection{Havoc Operator Lifting}
\label{sec:liftedDL:havoc}
On a syntactical level, lifting a dynamic theory $\dynamicTheory$ to a new dynamic theory $\havoced{\dynamicTheory}$ amounts to enhancing the original set of programs $\signaturePrograms$ with the havoc operator on all variables in $\signatureVariables$, i.e. $\havoced{\signaturePrograms} = \signaturePrograms \cupdot \left\{ \variableOne \coloneqq * ~\middle|~ \variableOne\in\signatureVariables\right\}$.
We can similarly extend $\dynamicTheory$'s domain of computation $\domComp$.
Overall, this amounts to the following lifting operation $\havoced{\left(\cdot\right)}$\new{S29}{(where $\cupdot$ is the disjoint union)}:
\begin{definition}[Havoc Lift]
\label{def:liftedDL:havoced}
Given a dynamic theory $\dynamicTheory$ we define $\havoced{\dynamicTheory} = \left(\havoced{\dynamicSignature},\havoced{\domComp}\right)$ such that all components of $\dynamicTheory$ and $\havoced{\dynamicTheory}$ are equal except for:\\
\begin{minipage}[t]{0.6\textwidth}\vspace{0pt}
\begin{itemize}
    \item $\havoced{\signaturePrograms} = \signaturePrograms\enspace\cupdot\enspace\left\{ \variableOne \coloneqq * ~\middle|~ \variableOne\in\signatureVariables\right\}$
    \item $\havoced{\programEval}\mleft(\programOne\mright) = \programEval\mleft(\programOne\mright)$ (for $\programOne \in \signaturePrograms$)
    \item $\havoced{\programEval}\mleft(\variableOne\coloneqq *\mright) = \left\{ \left(\stateOne,\stateTwo\right) \in \stateSpace^2 ~\middle|~ \equalOn{\stateOne}{\stateTwo}{\left\{\variableOne\right\}^\complement} \right\}$\hfill ($\variableOne \in \signatureVariables$)
\end{itemize}
\end{minipage}%
\begin{minipage}[t]{0.4\textwidth}\vspace{0pt}
\begin{itemize}
    \item $\havoced{\freeVarsProgram}\mleft(\programOne\mright) = \freeVarsProgram\mleft(\programOne\mright)$ (for $\programOne \in \signaturePrograms$)
    \item $\havoced{\freeVarsProgram}\mleft(\variableOne \coloneqq * \mright) = \emptyset$ (for $\variableOne \in \signatureVariables$)
\end{itemize}
\end{minipage}
\end{definition}%
\chg{S53}{%
\noindent
Note that $\variableOne$ is \emph{not} free in $\variableOne \coloneqq *$ as its value does not influence the program's reachable states. 
We prove that the havoc lift $\havoced{\dynamicTheory}$ of a dynamic theory $\dynamicTheory$ satisfies all properties of a dynamic theory:
}{
For any dynamic theory $\dynamicTheory$ we can prove that its havoc lifting $\havoced{\dynamicTheory}$ again satisfies all properties of a dynamic theory and consequently all properties, axioms and proof rules carry over:}
\begin{isabelleTheorem}{Havoc Lifted Dynamic Theories}{HavocDynLog.havoc\_dl}
\label{thm:liftedDL:havoced}
The havoc lift $\havoced{\dynamicTheory}$ of any dynamic theory $\dynamicTheory$ is a dynamic theory.
\end{isabelleTheorem}
For first-order atoms (\signatureAtoms) and the state space (\stateSpace), the required properties for a dynamic theory are still satisfied, because we did not modify these components. For programs, we ensured that our semantic definition of havoc satisfies the properties asserted in \Cref{def:DL:prog_eval}.
While it is advantageous that the properties proven for arbitrary dynamic theories carry over from $\dynamicTheory$ to $\havoced{\dynamicTheory}$, the havoc lifting would be useless if properties of the \emph{specific} theory $\dynamicTheory$ under consideration would not carry over as well.
If this were the case, all properties proven for a $\dynamicTheory$ under consideration would have to be reproved for $\havoced{\dynamicTheory}$ and the lifting operation would lose its utility.
Fortunately,\del{S50}{ it turns out that, syntactically,} our new set of formulas is a superset of the original formulas, i.e. $\formulaSet{\dynamicSignature} \subseteq \formulaSet{\havoced{\dynamicSignature}}$.
This raises the question of whether their semantics also carry over.
Indeed, we can prove a calculus rule which translates proofs about $\dynamicTheory$ into proofs about $\havoced{\dynamicTheory}$:
\begin{isabelleLemma}{Havoc Reduction}{dyn_axiom_HR}
\label{isaLem:liftedDL:havoc:reduction}
For any dynamic theory $\dynamicTheory$, the following proof rule is sound:
\begin{prooftree}
    \AxiomC{$\calculus[\dynamicTheory] \schemaVarOne$}
    \LeftLabel{\axiomLabel{HR}{HR}\normalfont (HR)}
    \RightLabel{ (assuming $\schemaVarOne\in\formulaSet{\dynamicSignature}$)}
    \UnaryInfC{$\calculus[\havoced{\dynamicTheory}] \schemaVarOne$}
\end{prooftree}
\end{isabelleLemma}
Hence, we can reuse \emph{any} decision procedure developed for determining the validity of formulas in $\formulaSet{\dynamicSignature}$ to determine the validity of formulas in the havoc-free fragment of $\formulaSet{\havoced{\dynamicSignature}}$ and carry these results over using \ref{axiom:HR}.
Additionally, we prove the soundness of the well-known \ref{axiom:havoc} axiom for the decomposition of the havoc operator:
\begin{isabelleLemma}{Havoc Axiom}{dyn_axiom_havoc}
\label{isalem:liftedDL:havoc:axiom}
For any dynamic theory $\dynamicTheory$, the following axiom is sound in $\havoced{\dynamicTheory}$:
\begin{axiomBlock}
    \refAxiom{havoc}{$\coloneqq\!\!*$}{
        $\left(\modBox{\variableOne\coloneqq *}{\schemaVarOne}\right) \iff \left(\fa{\variableOne}{\schemaVarOne}\right)$
    }
\end{axiomBlock}
\end{isabelleLemma}

Before we move to regular programs, let us recap what this section demonstrated\chg{S10}{.}{:}
Given any dynamic theory $\dynamicTheory$, we can create a new dynamic theory $\havoced{\dynamicTheory}$ which:
\begin{enumerate*}
    \item Has the additional program primitive $\variableOne \coloneqq *$;
    \item Inherits the dynamic theory axioms \ref{axiom:G}, \ref{axiom:V}, \ref{axiom:B} and \ref{axiom:K};
    \item Allows reuse of existing proof infrastructure for $\dynamicTheory$ through \ref{axiom:HR} (as a valid formula of $\dynamicTheory$ is also valid in $\havoced{\dynamicTheory}$).
\end{enumerate*}
\new{S63}{%
Rule \ref{axiom:HR} enables this reuse even when $\dynamicTheory$'s axiom schemata may not hold for arbitrary $\havoced{\dynamicTheory}$-instantiations, by restricting reductions to formulas from $\formulaSet{\dynamicSignature}$.
}%
The remainder of this section, as well as \Cref{sec:hdl}, will demonstrate significantly more expressive lifting procedures with the same properties\chg{S10}{,
i.e.\ transfer}{:
Transfer} of proof rules and validity preservation.

\begin{example}%
\label{ex:liftedDL:havoc:semiring}
Continuing the example of a dynamic theory over a semi-ring $\dynamicTheory^{\ring}$ defined in \Cref{ex:DL:syntax_semiring}\chg{S64}{%
, we can lift $\dynamicTheory^{\ring}$ with havoc resulting in $\havoced{\left(\dynamicTheory^{\ring}\right)}$.
}{:
Using the lifting procedure above, we can lift this dynamic theory to support the havoc operator
resulting in the theory $\havoced{\left(\dynamicTheory^{\ring}\right)}$.}
\chg{S64}{%
Beyond programs $x \coloneqq t$ (for semiring terms $t$) the theory $\havoced{\left(\dynamicTheory^{\ring}\right)}$ then also supports programs $x \coloneqq *$.
While this does not increase expressivity (see axiom \ref{axiom:havoc}), it can increase expressivity when combined with the regular closure defined below.
$\havoced{\left(\dynamicTheory^{\ring}\right)}$ also admits the application of all proof principles from $\dynamicTheory^{\ring}$ via \ref{axiom:HR}.
}{In addition to the proof principles from $\dynamicTheory^{\ring}$, this theory also supports the proof rule \ref{axiom:HR} and the axiom \ref{axiom:havoc}.}
\end{example}

\subsection{Regular Program Lifting}
\label{sec:liftedDL:regular}
Given a set of (atomic) programs, it is often useful to consider their \emph{regular programs closure}, i.e., any construction of programs obtained by sequential composition, nondeterministic choice, (first-order logic) checks, and nondeterministic loops.
This is an interesting set of programs, because it allows constructing well-known imperative constructs such as \texttt{if} or \texttt{while}, from \chg{S54}{a small}{minimal} set of operations.
It is also isomorphic to the program constructs considered in Kleene Algebra with Tests~\cite{DBLP:journals/toplas/Kozen97}.
We now show that any dynamic theory $\dynamicTheory$ can be lifted to this more expressive set of programs while retaining the properties established above.
To this end, we begin by defining the regular closure of programs along with their semantics:
\begin{definition}[Regular Closure]
\label{def:liftedDL:regular_closure}
Given a dynamic theory $\dynamicTheory$ over programs $\signaturePrograms$ with program evaluation function $\programEval$ we define the regular closure of programs, denoted $\regular{\signaturePrograms}$ via the following grammar (where $\programThree \in \signaturePrograms$ and $\formulaOneFOL\in\folFormulaSet{\dynamicSignature}$):
$
\grammar{\programOne,\programTwo}{
\programThree \grammarOr
\programOne;\programTwo \grammarOr
\programOne \cup \programTwo \grammarOr
?\left(\formulaOneFOL\right) \grammarOr
\left(\programOne\right)^*
}.
$\\
We then lift the program evaluation function to $\regular{\programEval}$ as follows:
\begin{itemize}
    \item $\regular{\programEval}\mleft(\programThree\mright) = \programEval\mleft(\programThree\mright)$ (for $\programThree\in\signaturePrograms$)
    \item $\regular{\programEval}\mleft(\programOne;\programTwo\mright) = \left\{
    \left(\stateOne,\stateThree\right) \in \stateSpace^2~\middle|~
    \text{there exists } \stateTwo\in\stateSpace \text{ s.t. }
    \left(\stateOne,\stateTwo\right) \in \regular{\programEval}\mleft(\programOne\mright) \text{ and }
    \left(\stateTwo,\stateThree\right) \in \regular{\programEval}\mleft(\programTwo\mright) 
    \right\}$
    \item $\regular{\programEval}\mleft(\programOne\cup\programTwo\mright) = 
    \regular{\programEval}\mleft(\programOne\mright)\cup\regular{\programEval}\mleft(\programTwo\mright)$
    \item $\regular{\programEval}\mleft(?\left(\formulaOneFOL\right)\mright) = \left\{
    \left(\stateOne,\stateTwo\right) \in \stateSpace ~\middle|~
    \stateOne\in\sem{\formulaOneFOL} \text{ and } \equalOn{\stateOne}{\stateTwo}{\signatureVariables}
    \right\}$
    \item $\regular{\programEval}\mleft(\left(\programOne\right)^*\mright) = \bigcup_{n\in\mathbb{N}} \regular{\programEval}\mleft(\programOne^n\mright)$
    where $\programOne^0 \enspace\equiv\enspace ?\left(\top\right)$ and $\programOne^{n+1}\enspace\equiv\enspace \left(\programOne;\programOne^n\right)$
\end{itemize}
\end{definition}
To create a new dynamic theory, it remains to define the set of free variables $\regular{\freeVarsProgram}$.
To this end, the free variables for $\programThree\in\signaturePrograms$ are defined as $\freeVarsProgram\mleft(\programThree\mright)$, and the free variables under sequential composition and nondeterministic choice are defined as the union of their components.
The free variables of a looped program $\left(\programOne\right)^*$ are the free variables of $\programOne$.
For the check $?\left(\formulaOneFOL\right)$ we leverage our knowledge about the free variables of atoms (encoded in \freeVarsAtom) to construct the free variables of a first-order formula (see \ifextension{\Cref{def:liftedDL:freeVarsFOL}}{Definition~C.1}).
Note that popular definitions of $?\left(\formulaOneFOL\right)$ often enforce semantics where it must be the case that $\stateOne=\stateTwo$ while we allow transitions to equivalent states as the usual definition breaks our extensionality assumption in \Cref{def:DL:prog_eval}
(since we allow distinct states $\stateOne,\stateTwo\in\stateSpace$ such that $\equalOn{\stateOne}{\stateTwo}{\signatureVariables}$).
This particularity disappears if we assume \stateSpace{} to be a ``classical'' mapping from variables to values.
We can then formally define the regular closure lift and prove that it constructs once again a dynamic theory:
\begin{definition}[Regular Closure Lift]
\label{def:liftedDL:regular_lift}%
\looseness=-1
For a dynamic theory $\dynamicTheory$ and its regular closure over programs $\regular{\signaturePrograms},\regular{\programEval}$ (see \Cref{def:liftedDL:regular_closure}) and $\regular{\freeVarsProgram}$ as defined in \ifextension{\Cref{def:liftedDL:programFV}}{{Definition~C.2}},
the \emph{regular closure lift} is defined as:
\[
\regular{\dynamicTheory} = \left(\left(\signatureVariables,\signatureAtoms,\regular{\signaturePrograms}\right),\left(
\universe,\stateSpace,\variableEval,\atomEval,\freeVarsAtom,\regular{\programEval},\regular{\freeVarsProgram}
\right)\right)
\]
\end{definition}
\begin{textAtEnd}
\begin{definition}[Free Variables of First-Order Formulas]
\label{def:liftedDL:freeVarsFOL}
Given a dynamic theory $\dynamicTheory$, we define the syntactically free variables $\freeVarsSyn$ over
$\formulaSet{\dynamicSignature}$ as follows:
\begin{itemize}
    \item $\freeVarsSyn\mleft(\atomOne\mright) = \freeVarsAtom\mleft(\atomOne\mright)$ (for $\atomOne \in \signatureAtoms$)
    \item $\freeVarsSyn\mleft(\neg\formulaOne\mright) = \freeVarsSyn\mleft(\formulaOne\mright)$
    \item $\freeVarsSyn\mleft(\formulaOne\land\formulaTwo\mright) = \freeVarsSyn\mleft(\formulaOne\mright)\cup\freeVarsSyn\mleft(\formulaOne\mright)$
    \item $\freeVarsSyn\mleft(\fa{\variableOne}{\formulaOne}\mright) = \freeVarsSyn\mleft(\formulaOne\mright)\setminus\left\{\variableOne\right\}$
    \item $\freeVarsSyn\mleft(\modBox{\programOne}{\formulaOne}\mright) = \freeVarsProgram\mleft(\programOne\mright) \cup \freeVarsSyn\mleft(\formulaOne\mright)$
\end{itemize}
\end{definition}
\begin{isabelleLemma}{$\freeVarsSyn$ is an Overapproximation}{dynLogCore.dyn\_FV\_fml\_sound}
\label{isaLem:fv_is_overapprox}
For any dynamic theory $\dynamicTheory$ and any formula $\formulaOne\in\formulaSet{\dynamicSignature}$ it holds that:
\[
\freeVarsSem\mleft(\formulaOne\mright) \subseteq \freeVarsSyn\mleft(\formulaOne\mright)
\]
\end{isabelleLemma}
\begin{definition}[Free Variables over Regular Programs]
\label{def:liftedDL:programFV}
For any dynamic theory $\dynamicTheory$ and its regular closure over programs $\regular{\signaturePrograms}$, we define $\regular{\freeVarsProgram}$ as follows:
\begin{itemize}
    \item $\regular{\freeVarsProgram}\mleft(\programThree\mright) = \freeVarsProgram\mleft(\programThree\mright)$ (for $\programThree \in \signaturePrograms$)
    \item $\regular{\freeVarsProgram}\mleft(\programOne;\programTwo\mright) = \regular{\freeVarsProgram}\mleft(\programOne\cup\programTwo\mright) = \regular{\freeVarsProgram}\mleft(\programOne\mright) \cup \regular{\freeVarsProgram}\mleft(\programTwo\mright)$
    \item $\regular{\freeVarsProgram}\mleft(?\left(\formulaOneFOL\right)\mright) = \freeVarsSyn\mleft(\formulaOneFOL\mright)$
    \item $\regular{\freeVarsProgram}\mleft(\left(\programOne\right)^*\mright) = \regular{\freeVarsProgram}\mleft(\programOne\mright)$
\end{itemize}
\end{definition}
\end{textAtEnd}
\begin{isabelleTheorem}{Regular Closure over Dynamic Theory}{kat_dl}
\label{thm:liftedDL:regular}
The regular closure lift $\regular{\dynamicTheory}$ of any dynamic theory $\dynamicTheory$ is a dynamic theory.
\end{isabelleTheorem}
Just as before, we can define a reduction rule which allows us to reuse proof results in $\regular{\dynamicTheory}$ that have already been shown for $\dynamicTheory$.
Additionally, our lifting comes with numerous axioms on the decomposition of regular programs which we summarize in \Cref{fig:liftedDL:axioms} and are provably sound:
\begin{isabelleLemma}{Regular Reduction}{dyn_axiom_RR}
\label{isalem:liftedDL:rr_axiom}
For any dynamic theory $\dynamicTheory$, the following proof rule is sound:
\begin{prooftree}
    \AxiomC{$\calculus[\dynamicTheory] \schemaVarOne$}
    \LeftLabel{\axiomLabel{RR}{RR}\normalfont (RR)}
    \RightLabel{ (assuming $\schemaVarOne\in\formulaSet{\dynamicSignature}$)}
    \UnaryInfC{$\calculus[\regular{\dynamicTheory}] \schemaVarOne$}
\end{prooftree}
\end{isabelleLemma}
\begin{isabelleTheorem}{Soundness of Proof Rules over Regular Closure}{}
\label{isathm:soundess_regular_axioms}
For any dynamic theory $\dynamicTheory$ the axioms in \Cref{fig:liftedDL:axioms} are sound in $\regular{\dynamicTheory}$.
\end{isabelleTheorem}
\begin{figure}[t]
\fbox{\begin{minipage}[t]{0.97\linewidth}\vspace{0pt}
\def\fCenter{~\vdash~}
\begin{minipage}[t]{0.48\textwidth}
\begin{axiomBlock}
    \refAxiom{?}{?}{
    $
    \modBox{?\left(\schemaVarOne\right)}{ \schemaVarTwo}
    \iff
    \left(\schemaVarOne \implies \schemaVarTwo\right)
    $}
    \refAxiom{seq}{;}{
    $
    \modBox{\schemaProgramOne;\schemaProgramTwo}{\schemaVarOne}
    \iff
    \modBox{\schemaProgramOne}{\modBox{\schemaProgramTwo}}{\schemaVarOne}
    $}
    \refAxiom{star}{\ensuremath{^*}}{
    $
    \modBox{\schemaProgramOne^*}{\schemaVarOne}
    \iff
    \schemaVarOne \land 
    \modBox{\schemaProgramOne}{\modBox{\schemaProgramOne^*}}{\schemaVarOne}
    $}
\end{axiomBlock}
\end{minipage}%
\begin{minipage}[t]{0.49\textwidth}
\begin{itemize}[label={ABC},align=left, leftmargin=*,itemsep=5pt]
    \refAxiom{cup}{\ensuremath{\cup}}{
    $
    \modBox{\schemaProgramOne \cup \schemaProgramTwo}{\schemaVarOne}
    \iff
    \modBox{\schemaProgramOne}{\schemaVarOne}
    \land
    \modBox{\schemaProgramTwo}{\schemaVarOne}
    $}
    \refAxiom{I}{I}{
    $
    \modBox{\schemaProgramOne^*}{\left(\schemaVarOne \implies \modBox{\schemaProgramOne}\schemaVarOne\right)}
    \implies
    \left(\schemaVarOne \implies \modBox{\schemaProgramOne^*} \schemaVarOne\right)
    $}
\end{itemize}
\end{minipage}
\end{minipage}}
\caption{Axioms for regular programs \new{S55}{originally due to Segerberg~\cite{Segerberg1982ACT}}}
\label{fig:liftedDL:axioms}
\end{figure}%
The proof rules provided in \Cref{fig:liftedDL:axioms} allow us to decompose checks, sequential compositions and nondeterminstic choice programs (see axioms \ref{axiom:?}, \ref{axiom:seq}, \ref{axiom:cup}).
Moreover, they allow the iterative decomposition of loops as well as inductive invariant reasoning about loops (see axioms \ref{axiom:star}, \ref{axiom:I}).
However, the axioms are,  yet, insufficient for loop termination.

\paragraph{Loop Convergence}
A significant strength of dynamic logics, and hence dynamic theories, is their ability to reason about partial \emph{and} total correctness via the dual diamond and box modalities.
To this end, it is sometimes necessary to not only reason about loop invariants (via \ref{axiom:I}), but also about loop \emph{variants} to prove progress.
To this end, consider a formula $\modDia{\programOne^*}{\formulaOne}$\chg{S10}{.
In this instance,}{:
Here,} we must prove that running program $\programOne$ in a loop \emph{eventually} leads to a state satisfying $\formulaOne$.
This is usually achieved via well-founded induction proofs\chg{S10}{.}{:}
Given a formula $\phi\mleft(N\mright)$ that is satisfied for a sufficiently large $N\in\mathbb{N}$, one iteration of \programOne{} ensures that $\phi\mleft(N-1\mright)$ is satisfied in the post state.
If $\phi\mleft(0\mright)$ implies $\formulaOne$,
this proves that $\formulaOne$ becomes satisfied after iteratively running $\programOne$ for a sufficiently large number of steps.
This requires \emph{counting}.
However, due to the generality of our assumed state space, we currently have no way of counting or performing well-founded induction.
Hence, we now propose a set of additional assumptions that allow us to introduce a loop termination rule:
\begin{definition}[Inductive Expressivity]
\label{def:liftedDL:inductive_expressive}
Given some dynamic theory $\dynamicTheory$, assume some function $\universeToNat : \universe \to \mathbb{N}$ mapping elements from the universe to natural numbers.
We say a variable $\variableOne \in \signatureVariables$ is \emph{\integerExpressive} w.r.t. $\universeToNat$ iff for all $n \in \mathbb{N}$ there exists a state $\stateOne\in\stateSpace$ such that $\universeToNat\left(\variableEval\mleft(\stateOne,\variableOne\mright)\right)=n$.
Let $\integerVariables$ be a set of such variables.
We say $\dynamicTheory$ has \emph{inductive expressivity} iff for some given $\universeToNat$ there exist two functions $\natPlusOne,\natEq : \integerVariables^2 \to \folFormulaSet{\dynamicSignature}$ and one function $\natPositive : \integerVariables \to \folFormulaSet{\dynamicSignature}$ such that:
\begin{localeAssmBlock}
\localeAssm{Positive Sound}{
For any variable $\variableOne\in\integerVariables$ and any state $\stateOne\in\stateSpace$:\\
$\stateOne\in\sem{\natPositive\mleft(\variableOne\mright)}$ iff $\universeToNat\mleft(\variableEval\mleft(\stateOne,\variableOne\mright)\mright)>0$}
\localeAssm{Equal Sound}{
For variables $\variableOne,\variableTwo\in\integerVariables$  with $\variableOne \neq \variableTwo$ and any state $\stateOne\in\stateSpace$:\\
$\stateOne\in\sem{\natEq\mleft(\variableOne,\variableTwo\mright)}$ iff $
\universeToNat\mleft(\variableEval\mleft(\stateOne,\variableOne\mright)\mright) = 
\universeToNat\mleft(\variableEval\mleft(\stateOne,\variableTwo\mright)\mright)
$
}
\localeAssm{Plus One Sound}{
For variables $\variableOne,\variableTwo\in\integerVariables$  with $\variableOne \neq \variableTwo$ and any state $\stateOne\in\stateSpace$:\\
$\stateOne\in\sem{\natPlusOne\mleft(\variableOne,\variableTwo\mright)}$ iff $
\universeToNat\mleft(\variableEval\mleft(\stateOne,\variableOne\mright)\mright) + 1 = 
\universeToNat\mleft(\variableEval\mleft(\stateOne,\variableTwo\mright)\mright)
$
}
\end{localeAssmBlock}
\end{definition}

This inductive expressivity yields all that is required to prove loop convergence, which enables us to prove loop termination and reachability of certain states via loop iteration.
To this end, dynamic logics usually are equipped with %
\chg{S56}{%
a proof rule~\cite{harel_first-order_1979}
or axiom~\cite{Platzer2012}
}{
an axiom} that looks something like this:
\begin{align*}
\modBox{\programOne^*}{
\left(
\fa{\variableOne}{\left(
\left(v>0 \land \schemaVarOne\left(\variableOne\right)\right) \implies
\modDia{\programOne}{\schemaVarOne\left(\variableOne-1\right)}
\right)
}
\right)
} &\implies
\fa{\variableOne}{\left(
\schemaVarOne\mleft(\variableOne\mright) \implies
\modDia{\programOne^*}{
\schemaVarOne\mleft(0\mright)}
\right)
}
&&(\variableOne\notin\variablesOf{\programOne})
\end{align*}
This \chg{S11}{axiom}{axioms} formalizes the well-founded induction argument outlined above:
If $\schemaVarOne\mleft(\variableOne\mright)$ is satisfied and one iteration of $\programOne$ allows us to reach a state satisfying $\schemaVarOne\mleft(\variableOne-1\mright)$, then we can reach a state with $\schemaVarOne\mleft(0\mright)$ by iterating the execution of $\programOne$ (assuming $\variableOne$ is a natural number).
To prove loop reachability properties, we also want an axiom like this one in the regular closure over our dynamic theory; however, we lack any means to parameterize formulas with specific variables.
After all, our theory so far makes no assumptions about how atomic formulas in $\signatureAtoms$ evaluate w.r.t. variables beyond the definition of $\freeVarsAtom$.
Hence, we derive the following axiom which can be proven independently of the dynamic theory at hand so long as the properties from \Cref{def:liftedDL:inductive_expressive} are given:
\begin{isabelleLemma}{Loop Convergence Rule}{dyn\_axiom\_C}
\label{isalem:loop_convergence}
For any inductively expressive dynamic theory $\dynamicTheory$, the following axiom is sound for its regular closure $\regular{\dynamicTheory}$ for any \integerExpressive{} distinct variables $\variableOne,\variableTwo\in\integerVariables$ with $\variableTwo \notin \freeVarsSem\mleft(\schemaVarOne\mright)$ and $\variableOne,\variableTwo \notin \variablesOf{\programOne}$:
\begin{axiomBlock}
    \refAxiom{C}{C}{
    $\begin{array}{l}\left(
    \modBox{\programOne^*}{\left(
    \fa{\variableOne}{\left(
    \left(
    \natPositive\left(\variableOne\right) \land
    \phi
    \right)
    \implies
    \modDia{\programOne}{\left(
    \fa{\variableTwo}{\left(
    \natPlusOne\left(\variableTwo,\variableOne\right) \implies
    \fa{\variableOne}{\left(
    \natEq\left(\variableOne,\variableTwo\right) \implies \schemaVarOne\right)
    }
    \right)
    }
    \right)
    }
    \right)
    }
    \right)
    }
    \right)\\\implies
    \left(
    \fa{\variableOne}{\left(
    \schemaVarOne \implies
    \modDia{\programOne^*}{\left(
    \ex{\variableOne}{\left(\neg\natPositive\left(\variableOne\right) \land \schemaVarOne\right)}
    \right)}
    \right)
    }
    \right)
    \end{array}
    $
    }
\end{axiomBlock}
\end{isabelleLemma}
Note that once we apply $\natPlusOne,\natEq$ or $\natPositive$ to concrete variables, they reduce to a concrete $\dynamicTheory$-formula yielding a concrete instantiation of axiom $\ref{axiom:C}$.
Importantly, we cannot subtract from  $\variableOne$ within a single predicate since any predicate (including $\natPlusOne$) requires a reference point from which to subtract 1. We mitigate this by employing the helper variable $\variableTwo$.
To demonstrate the utility of the additional assumptions, we have also proven that a concrete Dynamic Logic instantiates the properties defined in \Cref{def:liftedDL:inductive_expressive}:
\begin{isabelleLemma}{Semi-Ring First-Order Dynamic Logic over Natural Numbers}{semiring_dl_full_nat,semiring_dl_full_nat_int}
\label{lem:semiring_dl_inductive}
The semi-ring dynamic theory $\dynamicTheory^{\ring}$\del{}{ (from Lemma X)}, instantiated for natural numbers or integers (resp. \del{}{denoted }$\dynamicTheory^{\mathbb{N}}$ or $\dynamicTheory^{\mathbb{Z}}$), is inductively expressive.
\chg{}{Hence}{Consequently}, axiom \ref{axiom:C} is a sound axiom for its havoc lifted, regular closure $\regular{\left(\havoced{\left(\dynamicTheory^{\mathbb{N}}\right)}\right)}$ (resp. $\regular{\left(\havoced{\left(\dynamicTheory^{\mathbb{Z}}\right)}\right)}$).
\end{isabelleLemma}

\begin{example}[Regular Program Lifting for Ordered Semiring]
\label{ex:liftedDL:regular:semiring}
We reconsider the example of a dynamic logic over an ordered semiring $\left(\ring,+,\cdot,\leq\right)$ for $\ring=\mathbb{Z}$:
As seen in \Cref{lem:semiring_dl_inductive}, this dynamic theory is inductively expressive.
Hence, all axioms derived up to this point (\Cref{fig:elementary-axioms,fig:liftedDL:axioms} as well as \ref{axiom:HR}, \ref{axiom:RR}, \ref{axiom:havoc}, \ref{axiom:C}) apply to its havoc lifted, regular closure $\regular{\left(\havoced{\left(\dynamicTheory^{\mathbb{Z}}\right)}\right)}$.
From a very simple definition of assignment, we have thus derived a dynamic theory with full support for regular programs and nondeterministic assignment.
For example, the following program is part of this dynamic theory:
\begin{equation}
\label{eq:gauss_sum}
1 \leq n \implies
\modBox{
x \coloneqq 0;
i \coloneqq 0;
\big(
?\left(i \leq n\right);
x \coloneqq x + i;
i \coloneqq i + 1
\big)^*;
?\left(\neg\left(i \leq n\right)\right)
}{
2*x \leq n*(n+1)
}
\end{equation}
This formula asserts that x is no larger than the Gauss formula predicts after summing up the first $n$ integers.
Based on our simple definitions in \Cref{ex:DL:syntax_semiring} we can now perform reasoning about programs like the one in \Cref{eq:gauss_sum}.
In fact, $\regular{\left(\havoced{\left(\dynamicTheory^{\mathbb{Z}}\right)}\right)}$ provides the equivalent of a guarded command language over integers and, as will be shown in \Cref{sec:rel-complete}, a relatively complete proof calculus for it.
In fact, our initial definition of $\dynamicTheory^{\mathbb{Z}}$
even could have omitted the assignment programs (by initializing $\signaturePrograms^{\mathbb{Z}}=\emptyset$) by defining assignment as syntactic sugar for havoc and check~\cite{DBLP:journals/corr/abs-2504-03495}:
\begin{align*}
\variableOne \coloneqq t &\equiv\left(
\variableTwo\coloneqq*;?\left(t\leq\variableTwo\land\variableTwo\leq t\right);\variableOne\coloneqq*;?\left(\variableOne\leq\variableTwo\land\variableTwo\leq\variableOne\right)\right)
&\text{for $w\in\signatureVariables^{\mathbb{Z}}$ fresh}
\end{align*}
\end{example}

\paragraph{Kleene Algebra with Tests}%
\chg{L12}{%
We prove that our regular closure lifting forms a Kleene Algebra with Test~\cite{DBLP:journals/toplas/Kozen97,harel_dynamic_2000} over the theory's programs (see \ifextension{\Cref{isathm:reg_closure_kat}}{{\isaText{Theorem}~C.1}}).
In combination with axiom \ref{axiom:E}, this enables KAT-style rewriting within dynamic logics.
}{
In contrast to the formula-based approach presented in this section so far,
Kleene Algebra with Tests, as an orthogonal direction of research, proposes to prove properties over regular programs of the kind discussed in this section using equational rewriting and normalization.
Indeed, our definition of regular programs cannot only be used to derive standard axioms of dynamic logic, but it can equally be used to prove the axioms of a Kleene Algebra with Tests w.r.t. our refinement and equality relations $\katRef{\cdot}{\cdot}$ and $\katEq{\cdot}{\cdot}$:}
\begin{textAtEnd}
\begin{isabelleTheorem}{Regular Closure forms KAT}{KATRewriting}
\label{isathm:reg_closure_kat}
The program constructs of the regular closure lifting satisfy all axioms of a Kleene Algebra with Tests (see \Cref{def:KAT}) where $+$ is $\cup$, $\cdot$ is $;$, $1$ is $?\left(\top\right)$, $( \cdot )^*$ is the nondeterministic loop and the boolean algebra is constructed over checks $?\left(\cdot\right)$.\footnotemark
\footnotetext{For negation of check constructs we introduce an additional operator which aggregates and negates a given program consisting purely of checks (see \Cref{def:KAT}).}
\end{isabelleTheorem}%
\begin{definition}[Kleene Algebra with Tests]
\label{def:KAT}
We base our definition of Kleene and Boolean Algebra on prior literature~\cite{DBLP:journals/toplas/Kozen97,harel_dynamic_2000}.
In the notation of this paper, our regular program lifting satisfies the Kleene Algebra with Tests axioms if the following identitites hold:
\begin{align*}
    \schemaProgramOne \cup \left(\schemaProgramTwo \cup \schemaProgramThree\right) &\katEqSym \left(\schemaProgramOne \cup \schemaProgramTwo\right) \cup \schemaProgramThree
&   \schemaProgramOne \cup \schemaProgramTwo &\katEqSym \schemaProgramTwo \cup \schemaProgramOne
&   \schemaProgramOne \cup ?\left(\bot\right) &\katEqSym \schemaProgramOne
&   \schemaProgramOne \cup \schemaProgramOne &\katEqSym \schemaProgramOne\\
\schemaProgramOne ; \left(\schemaProgramTwo ; \schemaProgramThree\right)
    &\katEqSym \left(\schemaProgramOne ; \schemaProgramTwo\right) ; \schemaProgramThree
&   ?\left(\top\right) ; \schemaProgramOne
    &\katEqSym \schemaProgramOne
&   \schemaProgramOne ; ?\left(\top\right)
    &\katEqSym \schemaProgramOne
&   \schemaProgramOne ; \left(\schemaProgramTwo \cup \schemaProgramThree\right)
    &\katEqSym \left(\schemaProgramOne ; \schemaProgramTwo\right) \cup \left(\schemaProgramOne ; \schemaProgramThree\right)\\
   ?\left(\top\right) \cup (\schemaProgramOne ; \schemaProgramOne^{*})
    &\katEqSym \schemaProgramOne^{*}
&    ?\left(\bot\right) ; \schemaProgramOne
    &\katEqSym ?\left(\bot\right)
&   \schemaProgramOne ; ?\left(\bot\right)
    &\katEqSym ?\left(\bot\right)
&   \left(\schemaProgramOne \cup \schemaProgramTwo\right) ; \schemaProgramThree
    &\katEqSym \left(\schemaProgramOne ; \schemaProgramThree\right) \cup \left(\schemaProgramTwo ; \schemaProgramThree\right)\\
    ?\left(\top\right) \cup (\schemaProgramOne^{*} ; \schemaProgramOne)
    &\katEqSym \schemaProgramOne^{*}
    &\rlap{$%
    \schemaProgramOne \cup (\schemaProgramTwo;\schemaProgramThree) \katRefSym \schemaProgramThree \implies \schemaProgramTwo^*;\schemaProgramOne \katRefSym \schemaProgramThree%
    $}\\
    &&\rlap{$%
    \schemaProgramOne \cup (\schemaProgramThree;\schemaProgramTwo) \katRefSym \schemaProgramThree \implies \schemaProgramOne;\schemaProgramTwo^* \katRefSym \schemaProgramThree%
    $}
\end{align*}
and if furthermore for any programs $\schemaProgramOne,\schemaProgramTwo,\schemaProgramThree$ \emph{only} containing checks as atomic programs the following identitites hold:
\begin{align*}
    \schemaProgramOne \cup \left(\schemaProgramTwo \cup \schemaProgramThree\right) &\katEqSym \left(\schemaProgramOne \cup \schemaProgramTwo\right) \cup \schemaProgramThree
&   \schemaProgramOne ; \left(\schemaProgramTwo ; \schemaProgramThree\right) &\katEqSym \left(\schemaProgramOne ; \schemaProgramTwo\right) ; \schemaProgramThree
&   \schemaProgramOne \cup \schemaProgramTwo &\katEqSym \schemaProgramTwo \cup \schemaProgramOne
&   \schemaProgramOne ; \schemaProgramTwo &\katEqSym \schemaProgramTwo ; \schemaProgramOne\\
\schemaProgramOne \cup \left(\schemaProgramTwo ; \schemaProgramThree\right)
    &\katEqSym \left(\schemaProgramOne \cup \schemaProgramTwo\right) ; \left(\schemaProgramOne \cup \schemaProgramThree\right)
&   \schemaProgramOne ; \left(\schemaProgramTwo \cup \schemaProgramThree\right)
    &\katEqSym \left(\schemaProgramOne ; \schemaProgramTwo\right) \cup \left(\schemaProgramOne ; \schemaProgramThree\right)
&   \schemaProgramOne \cup ?\!\left(\bot\right)
    &\katEqSym \schemaProgramOne
&   \schemaProgramOne \cup ?\!\left(\top\right)
    &\katEqSym ?\!\left(\top\right)\\
    \schemaProgramOne \cup \overline{\schemaProgramOne}
    &\katEqSym ?\!\left(\top\right)
&   \schemaProgramOne \cup \schemaProgramOne
    &\katEqSym \schemaProgramOne
&   \schemaProgramOne ; ?\!\left(\top\right)
    &\katEqSym \schemaProgramOne
&   \schemaProgramOne ; ?\!\left(\bot\right)
    &\katEqSym ?\!\left(\bot\right)\\
    \schemaProgramOne ; \overline{\schemaProgramOne}
    &\katEqSym ?\!\left(\bot\right)
&   \schemaProgramOne ; \schemaProgramOne
    &\katEqSym \schemaProgramOne
&
&
&
&
\end{align*}
Where $\overline{\schemaProgramOne} \equiv ?\left(\neg\progToFml{\schemaProgramOne}\right)$ and $\progToFml{\schemaProgramOne}$ is defined as follows:
\begin{align*}
    \progToFml{?\left(\schemaVarOne\right)} &\equiv \schemaVarOne
&   \progToFml{\schemaVarOne;\schemaVarTwo} &\equiv \schemaVarOne \land \schemaVarTwo
&   \progToFml{\schemaVarOne\cup\schemaVarTwo} &\equiv \schemaVarOne \lor \schemaVarTwo
&   \progToFml{\left(\schemaVarOne\right)^*} &\equiv \top
\end{align*}
Note that $\progToFml{\schemaProgramOne}$ for $\schemaProgramOne\in\signaturePrograms$ need not be defined as we assume all atomic programs to be checks whenever $\progToFml{\cdot}$ is applied.
\end{definition}
\begin{isabelleLemma}{Correctness of $\progToFml{\cdot}$}{KATRewriting.eq_if_pure}
\label{isalem:correctness_check_to_fml}
If $\schemaProgramOne$ only contains checks as atomic programs then the following identity holds:
\[
\katEq{\schemaProgramOne}{?\left(\progToFml{\schemaProgramOne}\right)}
\]    
\end{isabelleLemma}
\end{textAtEnd}
\del{L12}{%
In combination with axiom \ref{axiom:E}, this insight allows us to mix-and-match proof tactics from Kleene Algebra with Tests with proof tactics from dynamic logic.}

\section{Heterogeneous Dynamic Theories}
\label{sec:hdl}%
\del{S65}{%
We have shown how any dynamic theory~$\dynamicTheory$ can be lifted for more complex behavior—either via havoc~($\havoced{\dynamicTheory}$) or via the regular closure~($\regular{\dynamicTheory}$).
In both cases, the developer of a logic can focus on axiomatizing their concrete program constructs while reusing existing components to increase expressiveness.}%
We \del{S65}{now }enhance dynamic theories further by \chg{S65}{allowing \emph{combination of}}{\emph{combining}} multiple theories to reason about \emph{heterogeneous systems}, whose components operate on different domains of computation, but may interact.
Heterogeneous theories serve two purposes:
\begin{enumerate*}
    \item they provide a unified framework that explicitly models both systems and their interactions, and
    \item they enable rigorous reasoning about the overall system by \emph{reusing} results and proof calculi from the \emph{homogeneous} component theories.
\end{enumerate*}
The heterogeneous proof calculus decomposes obligations to the individual theories, analogous to how reasoning over regular programs reduces to reasoning over atomic programs.

Given two \del{}{dynamic }theories $\zero{\dynamicTheory},\one{\dynamicTheory}$, we first construct a \emph{simple heterogeneous theory}~$\preHero{\dynamicTheory}$ over $\zero{\signaturePrograms}\cup\one{\signaturePrograms}$, then lift it with havoc and regular closure.
The resulting \emph{fully heterogeneous dynamic theory}~$\hero{\dynamicTheory}$ provides an expressive logic combining $\zero{\dynamicTheory}$ and $\one{\dynamicTheory}$ with all axioms from \Cref{sec:elementary,sec:liftedDL} included.

\subsection{Simple Heterogeneous Dynamic Theories}
\label{subsec:hdl:simple}
We assume we are given two dynamic theories $\zero{\dynamicTheory},\one{\dynamicTheory}$ and begin by defining the state space of our new dynamic theory:
\begin{definition}[Heterogeneous State Space]
\label{def:hdl:hero_state_space}
Given two dynamic theories $\zero{\dynamicTheory},\one{\dynamicTheory}$ and their resp.
state spaces $\zero{\stateSpace},\one{\stateSpace}$ we define the \emph{heterogeneous state space} as $\hero{\stateSpace}=\zero{\stateSpace}\times\one{\stateSpace}$.
For a given state $\left(\zero{\stateOne},\one{\stateOne}\right)\in\hero{\stateSpace}$ variable evaluation is defined as follows:
\[
\hero{\variableEval}\left(\left(\zero{\stateOne},\one{\stateOne}\right),\variableOne\right)=\begin{cases}
    \zero{\variableEval}\left(\zero{\stateOne},\variableOne\right) & \variableOne \in \zero{\signatureVariables}\\
    \one{\variableEval}\left(\one{\stateOne},\variableOne\right) & \variableOne \in \one{\signatureVariables}\\
\end{cases}
\]    
\end{definition}
A priori, the behavior modeled in $\zero{\dynamicTheory}$ and $\one{\dynamicTheory}$ hence happens in isolation.
However, similarly to reasoning about concurrent systems, the appeal of reasoning about heterogeneous systems lies precisely in the \emph{interaction} of these a priori independent behaviors.
Hence, to relate the behavior of these two independent state space components,
we introduce a new set of atoms (denoted $\common{\signatureAtoms}$) \chg{L7}{that allows us to constrain \emph{permissible combinations} of individual states $\zero{\stateOne}$ and $\one{\stateOne}$}{to reason about and relate $\zero{\stateOne}$ and $\one{\stateOne}$}:
\begin{definition}[Heterogeneous Atoms]
\label{def:hero_atoms}
Given two dynamic theories $\zero{\dynamicTheory},\one{\dynamicTheory}$ and their heterogeneous state space $\hero{\stateSpace}$ we define a new set of \emph{heterogeneous atoms} $\common{\signatureAtoms}$ over variables $\common{\signatureVariables}=\zero{\signatureVariables}\cup\one{\signatureVariables}$ with an evaluation function $\common{\atomEval}:\common{\signatureAtoms} \to 2^{\hero{\stateSpace}}$ and a free variable function $\common{\freeVarsAtom}:\common{\signatureAtoms} \to 2^{\common{\signatureVariables}}$ such that the constructs satisfy the atom evaluation requirements from \Cref{def:DL:atom_eval}.
\end{definition}
Heterogeneous \chg{S36}{atoms}{Atoms} can be considered the dynamic theory counterpart to mixed terms in first-order theory combination.
Just like in the case of \new{L7}{first-order }theory combination, these terms allow us to constrain the relation between values in both worlds.
\new{L7}{%
As a concrete example, \Cref{case_study_guarantee} defines the heterogeneous atom $\mathrm{round}\mleft(\texttt{v},\textit{w}\mright)$ for JavaDL integer variables $\texttt{v}$ and real-valued \dL{} variables \textit{w}.
Its semantics will allow us to \emph{constrain} the possible values of the Java variable $\texttt{v}$ (in state $\zero{\stateOne}$) based on the values of the \dL{} variable $\textit{w}$ (in $\one{\stateOne}$).
In combination with havoc ($x \coloneqq *$) and checks ($?\left(\formulaOneFOL\right)$), heterogeneous atoms serve as communication primitives between languages.
}

\chg{S65}{A simple heterogeneous dynamic theory is then defined as follows:}{Based on these definitions, we can now proceed to define the syntactic materials of our simple heterogeneous dynamic theory:}
\begin{definition}[Simple Heterogeneous Dynamic Signature]
\label{def:hdl:simple_signature}
Given two dynamic theories $\zero{\dynamicTheory},\one{\dynamicTheory}$ with resp. signatures $\zero{\dynamicSignature},\one{\dynamicSignature}$ and a set of heterogeneous atoms $\common{\signatureAtoms}$ over $\zero{\dynamicTheory},\one{\dynamicTheory}$ (see \Cref{def:hero_atoms}) we define the \emph{simple heterogeneous dynamic signature} as $\preHero{\dynamicSignature} = \left(\hero{\signatureVariables},\hero{\signatureAtoms},\preHero{\signaturePrograms}\right)$ with:

\begin{align*}
    \hero{\signatureVariables} &=
    \zero{\signatureVariables}\Dot{\cup}\one{\signatureVariables}
    &\hero{\signatureAtoms} &=
    \zero{\signatureAtoms}\Dot{\cup}\one{\signatureAtoms}\Dot{\cup}\common{\signatureAtoms}
    &\preHero{\signaturePrograms} &=
    \zero{\signaturePrograms}\Dot{\cup}\one{\signaturePrograms}
\end{align*}   
\end{definition}
The domain of computation, \chg{}{i.e.\ the heterogeneous theory's semantics, naturally follows:}{i.e. the semantical definition of our new dynamic theory then naturally follows:}
\begin{definition}[Simple Heterogeneous Domain of Computation]
\label{def:hdl:simple_dom_comp}
Given two dynamic theories $\zero{\dynamicTheory},\one{\dynamicTheory}$ and notation as in \Cref{def:hdl:simple_signature} we define the simple heterogeneous domain of computation $\preHero{\domComp}$ as follows:
\begin{itemize}
    \item The universe is the combination of individual universes: $\hero{\universe}=\zero{\universe}\cup\one{\universe}$
    \item The state space is the heterogeneous state space $\hero{\stateSpace}$ over $\zero{\dynamicTheory},\one{\dynamicTheory}$ with the corresponding variable evaluation function $\hero{\variableEval}$ (see \Cref{def:hdl:hero_state_space})
    \item The atom evaluation function $\hero{\atomEval}$ and program evaluation function $\preHero{\programEval}$ is given as below (with corresponding free variable functions $\hero{\freeVarsAtom},\preHero{\freeVarsProgram}$ as defined in \ifextension{\Cref{def:free_vars_simple_hero}}{{Definition~C.4}}).
\end{itemize}
\begin{align*}
    \hero{\atomEval}\mleft(\atomOne\mright) &= 
    \begin{cases}
        \zero{\atomEval}\mleft(\atomOne\mright)\times\one{\stateSpace} & \atomOne\in\zero{\signatureAtoms}\\
        \zero{\stateSpace}\times\one{\atomEval}\mleft(\atomOne\mright) & \atomOne\in\one{\signatureAtoms}\\
        \common{\atomEval}\mleft(\atomOne\mright) & \atomOne\in\common{\signatureAtoms}
    \end{cases}
    &\preHero{\programEval}\mleft(\programOne\mright) &= 
    \begin{cases}
        \zero{\programEval}\mleft(\programOne\mright) & \programOne\in\zero{\signaturePrograms}\\
        \one{\programEval}\mleft(\programOne\mright) & \programOne\in\one{\signaturePrograms}
    \end{cases}
\end{align*}
\end{definition}
\begin{textAtEnd}
\begin{definition}[Free Variables for Simple Heterogeneous Domain of Computation]
\label{def:free_vars_simple_hero}
Given two dynamic theories $\zero{\dynamicTheory},\one{\dynamicTheory}$ and notation as in \Cref{def:hdl:simple_signature}, we define free variables of atoms and programs as follows:
\begin{align*}
\hero{\freeVarsAtom}\mleft(\atomOne\mright) &= \begin{cases}
    \zero{\freeVarsAtom}\mleft(\atomOne\mright) & \atomOne\in\zero{\signatureAtoms}\\
    \one{\freeVarsAtom}\mleft(\atomOne\mright) & \atomOne\in\one{\signatureAtoms}\\
    \common{\freeVarsAtom}\mleft(\atomOne\mright) & \atomOne\in\common{\signatureAtoms}\\
\end{cases}
&
\hero{\freeVarsProgram}\mleft(\programOne\mright) &=\begin{cases}
    \zero{\freeVarsProgram}\mleft(\programOne\mright) & \atomOne\in\zero{\signaturePrograms}\\
    \one{\freeVarsProgram}\mleft(\programOne\mright) & \atomOne\in\one{\signaturePrograms}
\end{cases}
\end{align*}
\end{definition}
\end{textAtEnd}
The notation of our simple heterogeneous signature $\preHero{\dynamicSignature}$ and the simple heterogeneous dynamic theory $\preHero{\domComp}$ is already giving a hint at the components of the dynamic theory which are here to stay for the fully heterogeneous version (denoted with $\hero{\left(\cdot\right)}$) vs. components of the simple heterogeneous theory that still require further refinement (denoted with $\preHero{\left(\cdot\right)}$):
While the universe, state space, variables and atoms remain as-is, the range of admissible programs will still be extended.
Before we lift $\preHero{\dynamicTheory}$ to a more expressive set of programs, we begin by showing that our simple heterogeneous dynamic signature and domain of computation indeed form a dynamic theory:
\begin{isabelleTheorem}{Simple Heterogeneous Dynamic Theories}{het\_dl}
\label{isathm:simple_hdl}
Assume two dynamic theories $\zero{\dynamicTheory},\one{\dynamicTheory}$ and a corresponding set of heterogeneous atoms $\common{\signatureAtoms}$, then $\preHero{\dynamicSignature}$ and $\preHero{\domComp}$ as defined in \Cref{def:hdl:simple_signature,def:hdl:simple_dom_comp} form a dynamic theory $\preHero{\dynamicTheory}$ as defined in \Cref{def:DL:theory}.
\chg{S57}{The soundness of all}{%
All} axioms and proof rules from \Cref{fig:elementary-axioms} \chg{S57}{carries}{carry} over.
\end{isabelleTheorem}
\noindent
\looseness=-1
This result unlocks the first four proof rules for $\preHero{\dynamicTheory}$.
However, the range of properties we will be able to formalize and prove with $\preHero{\dynamicTheory}$ \chg{S37}{is}{are} still underwhelming\chg{S10}{ %
as modalities %
}{:
Each modality} in $\preHero{\dynamicTheory}$ can only contain a program from $\zero{\dynamicTheory}$ \emph{or} a program from $\one{\dynamicTheory}$\chg{S10}{ --- }{, but }\emph{not both}.
Before we \emph{lift} this limitation in the next section, we \del{S10}{first }consider how this logic relates to another technique for combining logics, namely fibring.

\paragraph{Relation to Fibring.}%
By explicitly labeling the modalities w.r.t. the contained program (i.e. $\modBox{\programOne}{\formulaOne}$ becomes $\zero{\left[\programOne\right]}\formulaOne$ for $\programOne\in\zero{\signaturePrograms}$ and $\one{\left[\programOne\right]}\formulaOne$ otherwise), we can interpret this logic as a ``multi-dynamic logic'' with two distinct modal operators.
This approximately corresponds to what could be achieved via fibring~\cite{gabbay1998fibring} of dynamic logics where we could combine the truth bearing entities of $\zero{\dynamicTheory}$ and $\one{\dynamicTheory}$.
\chg{S38}{%
As fibring does not recognize programs as first-class  citizens,
}{While fibring allows for the combination of truth-bearing entities from $\zero{\dynamicTheory}$ and $\one{\dynamicTheory}$, 
it does not recognize programs as first-class citizens.
Hence,} it cannot integrate their respective program constructs
which restricts the expressive power of analyzable systems.
In particular, it cannot capture systems involving, e.g., intricate control structures such as loops.
Instead, fibring is confined to modeling systems that can be decomposed into a (nondeterministic) finite, sequential compositions of homogeneous subsystems.
Importantly, our dynamic theory does \emph{not} have two distinct modalities.
This key difference allows us to apply the lifting procedures devised in \Cref{sec:liftedDL} which naturally yields a significantly more expressive program logic.

\subsection{Fully Heterogeneous Dynamic Theories}
Using the lifting procedures from \Cref{sec:liftedDL} we can define the fully heterogeneous dynamic theory, supporting the composition of homogeneous programs using all regular program constructs, in a straightforward manner:
\begin{definition}[Fully Heterogeneous Dynamic Theories]
\label{def:hdl:fully_hero}%
Given two dynamic theories $\zero{\dynamicTheory},\one{\dynamicTheory}$ and a corresponding set of heterogeneous atoms $\common{\signatureAtoms}$ forming the simple heterogeneous dynamic theory $\preHero{\dynamicTheory}$, we define the \emph{fully heterogeneous dynamic theory} over $\zero{\dynamicTheory},\one{\dynamicTheory}$ as
$
\hero{\dynamicTheory} = \regular{\left(\havoced{\left(\preHero{\dynamicTheory}\right)}\right)}
$
with $\regular{\left(\cdot\right)}$ and $\havoced{\left(\cdot\right)}$ as given in \Cref{thm:liftedDL:havoced} and \Cref{thm:liftedDL:regular}.
\end{definition}

As planned, we can recover all axioms established above (with the exception of \ref{axiom:C} for which the requirements will be discussed below):
\begin{isabelleTheorem}{Fully Heterogeneous Dynamic Theories}{het\_dl\_havoc,het\_dl\_kat}
\label{isathm:full_hdl}
Given two dynamic theories $\zero{\dynamicTheory},\one{\dynamicTheory}$ and a corresponding set of heterogeneous atoms $\common{\signatureAtoms}$ the fully heterogeneous dynamic theory $\hero{\dynamicTheory}$ is a dynamic theory.
Proof rules and axioms from \Cref{fig:liftedDL:axioms} as well as axioms \ref{axiom:HR}, \ref{axiom:havoc}, and \ref{axiom:RR} carry over to $\hero{\dynamicTheory}$.
\end{isabelleTheorem}

Just as for the regular closure over programs of \emph{one} dynamic theory, the construction of a heterogeneous dynamic theory over \emph{two} theories would be entirely devoid of purpose if validity wouldn't carry over from the individual logics to the combined one.
Fortunately, we can show that for any \emph{pure} $\zero{\dynamicSignature}/\one{\dynamicSignature}$-formula the validity w.r.t. $\zero{\dynamicTheory}/\one{\dynamicTheory}$ entails the validity w.r.t. $\hero{\dynamicTheory}$:
\begin{isabelleLemma}{Heterogeneous Reduction}{dyn\_axiom\_R0,dyn\_axiom\_R1}
\label{isalem:hdl:reduction}
For any fully heterogeneous dynamic theory $\hero{\dynamicTheory}$ over arbitrary $\zero{\dynamicTheory},\one{\dynamicTheory}$ the following proof rules are sound:\\
\begin{minipage}{0.5\linewidth}
\begin{prooftree}
    \AxiomC{$\calculus[\zero{\dynamicTheory}] \schemaVarOne$}
    \LeftLabel{\axiomLabel{HR0}{HR0}\normalfont (HR0)}
    \RightLabel{ (assuming $\schemaVarOne\in\formulaSet{\zero{\dynamicSignature}}$)}
    \UnaryInfC{$\calculus[\hero{\dynamicTheory}] \schemaVarOne$}
\end{prooftree}
\end{minipage}
\begin{minipage}{0.5\linewidth}
\begin{prooftree}
    \AxiomC{$\vdash_{\one{\dynamicTheory}} \schemaVarOne$}
    \LeftLabel{\axiomLabel{HR1}{HR1}\normalfont (HR1)}
    \RightLabel{ (assuming $\schemaVarOne\in\formulaSet{\one{\dynamicSignature}}$)}
    \UnaryInfC{$\vdash_{\hero{\dynamicTheory}} \schemaVarOne$}
\end{prooftree}
\end{minipage}
\end{isabelleLemma}

\paragraph{Loop Convergence}
In order for the loop convergence rule to carry over as well, we require that our heterogeneous dynamic theory $\hero{\dynamicTheory}$ is inductively expressive.
To this end, we note that whenever one of the homogeneous dynamic theories $\zero{\dynamicTheory},\one{\dynamicTheory}$ is inductive, then the heterogeneous dynamic theory is inductive as well:
\begin{isabelleLemma}{Inductive Expressivity of $\hero{\dynamicTheory}$}{inductively\_expressive\_1,inductively\_expressive\_2}
\label{lem:hdl:inductive}
Let $\zero{\dynamicTheory}$ be inductively expressive, then there exists a (constructive) adjusted function $\hero{\universeToNat}$ such that $\hero{\dynamicTheory}$ is inductively expressive w.r.t. the variables $\zero{\integerVariables}$.
The same holds in case $\one{\dynamicTheory}$ is inductively expressive.
\end{isabelleLemma}
\new{S58}{Using this result, any dynamic theory can be made inductively expressive by combining it with an already inductively expressive theory.}
Depending on whether we use a loop convergence rule inherited from $\zero{\dynamicTheory}$ or from $\one{\dynamicTheory}$, we denote it as (C$^{(0)}$) or (C$^{(1)}$), respectively.

\paragraph{Program Rewriting.}
\chg{L12}{%
While not necessary for relative completeness,
axiom \ref{axiom:E} can also be used to rewrite \emph{homogeneous regular program constructs} (here in {\color{blue}blue} and {\color{orange}orange}) into \emph{fully heterogeneous regular program constructs} (here in {\color{magenta}magenta}) via the following identities (see \ifextension{\Cref{isathm:heterogeneous_rewriting}}{{%
\isaText{Theorem}~C.2%
}}%
):
}{
While not strictly necessary for relative completeness, it turns out that axiom \ref{axiom:E} is of particular appeal for the verification of heterogeneous systems in cases where our heterogeneous dynamic theory contains (at least) one homogeneous dynamic theory which is lifted with regular programs.
To demonstrate the utility of equational program rewriting, consider the program 
$
\alpha_{\text{mixed}}
~\equiv~
{\color{magenta} \big(}
{\color{orange} \programOne_1 ; \programOne_2}
{\color{magenta} ; }
{\color{blue} \left(\programTwo_1 \cup \programTwo_2\right)^*}
{\color{magenta} \big)^*}.
$
where we have marked \emph{homogeneous program constructs} in {\color{blue} blue} and {\color{orange} orange} while \emph{fully heterogeneous program constructs} are marked in {\color{magenta} magenta}.
In this instance, we would now have to take account of program compositions in three different proof calculi:
The proof calculus of {\color{orange} homogeneous dynamic theory 1} for the sequential composition ${\color{orange} ; }$ and
the proof calculus of {\color{blue} homogeneous dynamic theory 2} for ${\color{blue} \cup}$ and ${\color{blue} \left(\cdot\right)^*}$.
Additionally. we need the proof calculus of the {\color{magenta} heterogeneous dynamic theory} for ${\color{magenta} ; }$ and the outer loop ${\color{magenta} \left(\cdot\right)^*}$.
Likely, there are cases where a more elegant proof could be achieved by reasoning about all regular programs at the \emph{heterogeneous} level.
To this end, we prove an additional set of identities:}
\new{L13}{
\begin{align*}
    \katEq{{\color{blue} ? \left(\schemaVarOne\right)}}{
        ~\color{magenta} ? \left({\color{blue} \schemaVarOne}\right)
    }
    &&
    \katEq{{\color{blue} \schemaProgramOne \cup \schemaProgramTwo}}{
        \color{magenta}
        {\color{blue} \schemaProgramOne} \cup {\color{blue} \schemaProgramTwo}
    }
    &&
    \katEq{{\color{blue} \schemaProgramOne ; \schemaProgramTwo}}{
        \color{magenta}
        {\color{blue} \schemaProgramOne} ; {\color{blue} \schemaProgramTwo}
    }
    &&
    \katEq{{\color{blue} \left(\schemaProgramOne\right)^*}}{
        \color{magenta}\left(
        {\color{blue} \schemaProgramOne} \right)^*
    }
\end{align*}
}
\begin{textAtEnd}
\begin{isabelleTheorem}{Heterogeneous Rewriting}{HeterogeneousRewriting}
\label{isathm:heterogeneous_rewriting}
Let $\zero{\dynamicTheory},\one{\dynamicTheory}$ be two dynamic theories and let $\hero{\dynamicTheory}$ be the fully heterogeneous dynamic theory over $\regular{\left(\zero{\dynamicTheory}\right)}$ and $\one{\dynamicTheory}$.
Then the following identities hold where programs of $\hero{\dynamicTheory}$ are in {\color{magenta} magenta} and programs of $\regular{\left(\zero{\dynamicTheory}\right)}$ are in {\color{blue} blue}:
\begin{align*}
    \katEq{{\color{blue} ? \left(\schemaVarOne\right)}}{
        ~\color{magenta} ? \left({\color{blue} \schemaVarOne}\right)
    }
    &&
    \katEq{{\color{blue} \schemaProgramOne \cup \schemaProgramTwo}}{
        \color{magenta}
        {\color{blue} \schemaProgramOne} \cup {\color{blue} \schemaProgramTwo}
    }
    &&
    \katEq{{\color{blue} \schemaProgramOne ; \schemaProgramTwo}}{
        \color{magenta}
        {\color{blue} \schemaProgramOne} ; {\color{blue} \schemaProgramTwo}
    }
    &&
    \katEq{{\color{blue} \left(\schemaProgramOne\right)^*}}{
        \color{magenta}\left(
        {\color{blue} \schemaProgramOne} \right)^*
    }
\end{align*}
\end{isabelleTheorem}
\noindent
\new{S12}{A corresponding result can be derived if some dynamic theory $\zero{\dynamicTheory}$ is combined with the regular closure $\regular{\left(\one{\dynamicTheory}\right)}$.}
\end{textAtEnd}
Using these identities, \chg{L13}{%
we can turn a program
$
\alpha_{\text{mixed}}
~\equiv~
{\color{magenta} \big(}
{\color{orange} \programOne_1}
{\color{magenta} ; }
{\color{blue} \left(\programTwo_1 \cup \programTwo_2\right)^*}
{\color{magenta} \big)^*}.
$
into a program 
$%
\alpha_{\text{unmixed}}
~\equiv~
{\color{magenta} \big(
{\color{orange} \programOne_1}
;
\left({\color{blue} \programTwo_1} \cup {\color{blue} \programTwo_2}\right)^*
\big)^*}
$
by syntactically deriving that $\katEq{\alpha_{\text{mixed}}}{\alpha_{\text{unmixed}}}$.
}{we can derive that $\katEq{\alpha_{\text{mixed}}}{\alpha_{\text{unmixed}}
~\equiv~
{\color{magenta} \big(
{\color{orange} \programOne_1} ; {\color{orange}\programOne_2}
;
\left({\color{blue} \programTwo_1} \cup {\color{blue} \programTwo_2}\right)^*
\big)^*}
}$. %
}
\chg{L13}{
This result, in combination with axiom \ref{axiom:E}, formalizes the natural correspondence between homogeneous and heterogeneous regular programs and enables us to decompose a homogeneous programs.
}{
This allows us to \emph{atomatize} the appearances of homogeneous programs in the case where logics with regular program support are turned into a heterogeneous theory\chg{S10}{,
i.e.\ any}{:
Any} appearance of $\alpha_{\text{mixed}}$ in some formula can then be replaced by an appearance of $\alpha_{\text{unmixed}}$ using axiom (E).
While the idea that a homogeneous loop may be replaced by a heterogeneous loop (and vice-versa) may seem obvious,
Theorem X along with axiom (E) turns this knowledge into a set of proof rules that may be applied in practice.}

\section{Relative Completeness}
\label{sec:rel-complete}
The expressivity of heterogeneous dynamic logic, enabling us to reason about intertwined heterogeneous dynamics, would equally be its downfall if the expressivity were to keep us from verifying properties in practice.
To this end, we usually desire to formalize dynamic logics in a way that guarantees \emph{relative completeness}, i.e. the property that any valid formula can also \chg{S39}{be proved}{been proven} valid if we assume the availability of an oracle for first-order validity\chg{L14}{:}{.}
\begin{definition}[Relative Completeness]
\label{def:rel-complete:rel-complete}
\chg{L14}{A dynamic theory $\dynamicTheory$ is called \emph{relatively complete}}{%
Given a dynamic theory $\dynamicTheory$, we say it is \emph{relatively complete}} if there exists a calculus $\calculus$ such that
if for some formula $\formulaOne\in\formulaSet{\dynamicTheory}$ it holds that $\models_{\dynamicTheory} \formulaOne$ then also
$\Gamma_{\dynamicTheory} \calculus \formulaOne$ where
$\Gamma_{\dynamicTheory} = \left\{ \formulaTwoFOL \in \folFormulaSet{\dynamicTheory}~\middle|~ \models_{\dynamicTheory} \formulaTwoFOL\right\}$.
\end{definition}
\noindent
\looseness=-1
This property is interesting, because it guarantees that, using the given calculus, proving properties about the dynamic logic is ``no harder'' than proving properties about its underlying (first-order) data theory.
\del{L14}{In }\ifextension{\Cref{apx:rel-complete}}{Appendix B} \chg{L14}{introduces}{we present} a calculus of metaproperties \del{L14}{which allows us} to reason about relative completeness and related properties for lifted and combined dynamic logics.
\chg{L14}{%
Due to space constraints, we focus the exposition in the paper to central assumptions necessary to obtain relative completeness and then present \Cref{thm:rel-complete:hero-dl} which proves that, under common, reasonable assumptions, the heterogeneous dynamic theory over two relatively complete dynamic theories is once again relatively complete.
}{%
Due to space constraints, we focus the exposition of the paper to one central result, namely
the preservation of relative completeness under combination in heterogeneous dynamic theories\chg{S10}{.
We show that under}{:
Under} common, reasonable assumptions, for two dynamic theories with relatively complete proof calculi their combination with a fully heterogeneous dynamic theory's proof calculus is once again relatively complete:}

\subsection{Assumptions for Relative Completeness.}
\new{L14}{%
We begin by providing a succinct overview on
properties of a dynamic theory $\dynamicTheory=\dynamicTheoryTuple$ required to prove 
relative completeness
while deferring formal definitions to the appendix.}

\paragraph{Finite Support.}
\new{L14}{%
We require that $\dynamicTheory$ has \emph{finite support} (see \ifextension{\Cref{def:rel-complete:finite-support}}{Definition B.2}), i.e. that there exists a syntactic over-approximation $\boundVarsProgram$ of semantically bound variables (i.e. $\boundVarsSem[\signaturePrograms]\mleft(\programOne\mright) \subseteq \boundVarsProgram\mleft(\programOne\mright)$ for all programs $\programOne$) and that $\boundVarsProgram$ is always \emph{finite}, i.e. all programs only modify a finite set of variables.}

\paragraph{Gödel Expressivity.}
\new{L14}{%
\looseness=-1
Like many relative completeness proofs in dynamic logic~\cite{Platzer2012,beckert_dynamic_2006,harel_first-order_1979}, we require a first-order theory which is capable of encoding \emph{finite sequences} of domain values into a \emph{single value}.
The first example of this mechanism is due to Gödel~\cite{godel1931formal} who showed that any finite sequence of integer values can be injectively mapped to a single integer value.
A more general description of this encoding mechanism (here predicate $R$) was established by Parikh (first published by Harel~\cite[p. 29]{harel_first-order_1979}):
\[
\left(\forall x_1\dots x_n\right)\left(\exists y\right)\left(\forall w\right)\left(\forall i\right) \left(\left(i\in\mathbb{N} \land 1 \leq i \leq n\right) \Longrightarrow \left(R\mleft(w,i,y\mright) \equiv \left(w=x_i\right)\right)\right)
\]
We generalize this notion to our definition of state spaces and atoms here and call an inductively expressive theory \emph{Gödel expressive} (\ifextension{\Cref{def:rel-complete:goedel-expressive}}{Definition B.5}), if the domains of all its variables support this type of sequence encoding.
Assuming suitable first order theories, Gödel expressivity has previously been shown for multiple dynamic logics including interpreted integer first-order dynamic logic~\cite{harel_first-order_1979}, differential dynamic logic over reals~\cite{Platzer2012} and dynamic logic with non-rigid functions~\cite{beckert_dynamic_2006}.}

\paragraph{FOL Expressivity.}
\new{L14}{%
An additional assumption serving as bread and butter of many meaningful relative completeness proofs~\cite{Platzer2012,harel_first-order_1979,beckert_dynamic_2006} is the notion of \emph{FOL expressivity} (\ifextension{\Cref{def:rel-complete:fol-expressive}}{{Definition~B.3}}).
This property asserts that for any program $\programOne$ there exists a first-order \emph{rendition} formula $\programRendition{\programOne}$ that \emph{exactly} relates the pre and post state of $\programOne$'s variables.
In essence, this ensures that we can express the behavior of programs in the theory's first-order fragment.
To prove FOL expressivity, one often relies on Gödel expressiveness.
For example, \ifextension{\Cref{lem:rel-completeness:fol-expressiveness-regular}}{{Lemma~C.14}} proves that for any FOL expressive, Gödel expressive dynamic theory with finite support its regular closure is also FOL expressive.}

\paragraph{Communicating Theories.}
\new{L14}{%
\looseness=-1
While previous assumptions constrained individual theories,
we also require an assumption about the expressivity of heterogeneous atoms in $\common{\signatureAtoms}$.
To this end, \ifextension{\Cref{def:hdl:communicating}}{Definition B.6} ensures that heterogeneous theories adhering to it have a common first-order structure that is expressive enough to exchange integer values via atoms formalizing equality between integers of both theories.
Without communication, our renditions of heterogeneous programs cannot \emph{exactly} represent the progress of a heterogeneous loop as the individual theories would not be able to synchronously count the number of loop iterations.}

\subsection{Relative Completeness under Combination.}
\new{L14}{We prove \Cref{thm:rel-complete:hero-dl} which states that the fully heterogeneous dynamic theory over two relatively complete dynamic theories is again relatively complete (proof see \ifextension{\cpageref{proof:rel-complete:hero-dl}}{{extended version}}).}
\begin{theorem}[Relative Completeness of Fully Heterogeneous Dynamic Theories]
\label{thm:rel-complete:hero-dl}
Let $\zero{\dynamicTheory},\one{\dynamicTheory}$ be two relatively complete (\Cref{def:rel-complete:rel-complete}), FOL expressive (\ifextension{\Cref{def:rel-complete:fol-expressive}}{{Definition~B.3}}), Gödel expressive (\ifextension{\Cref{def:rel-complete:goedel-expressive}}{{Definition~B.5}}) dynamic theories with finite support (\ifextension{\Cref{def:rel-complete:finite-support}}{{Definition~B.2}})
that are communicating in $\preHero{\dynamicTheory}$ (\ifextension{\Cref{def:hdl:communicating}}{{Definition~B.6}}).
Then the fully heterogeneous dynamic theory $\hero{\dynamicTheory}$ with its proof calculus $\calculus[\hero{\dynamicTheory}]$ is relatively complete.
\end{theorem}

\section{Case Study}
\label{case_study_guarantee}
\chg{L9}{To demonstrate our approach, we apply HDL to the case study from \Cref{sec:exemplary_case_study_intro}. We }{%
To address the case study presented in Section X, we}%
first derive results about the individual two components.
Subsequently, we merge these results into a system level guarantee.
\new{L9}{We will use the homogeneous (atomic) equality predicate $t_1 \doteq t_2$ which states that the values of two terms $t_1,t_2$ are equal
as well as the homogeneous inequality predicates ($\leq$,$>$,$<$).}

\paragraph{Java.}
\looseness=-1
For the Java component \texttt{ctrl} we wish to verify that 
once the relative distance to the stop sign (given as \texttt{p}) becomes too small it must brake by choosing an appropriate acceleration.
Using JavaDL and the proof calculus implemented in the interactive theorem prover \KeY~\cite{DBLP:conf/javacard/Beckert00,KeYBook}, we can then derive the result in \Cref{lem:java_ctrl_correctness}.
We do not only prove properties on the\new{L2}{ stateful} result of \texttt{this.acc}, but also on the preservation of the heap structure\new{L2}{, the correctness of (possibly overflowing) integer arithmetic,} and the field values of \texttt{A}, \texttt{B}, and \texttt{T} (mechanized proof in \KeY{})\chg{L2}{. %
Note that these aspects of control software would be hard to model in \dL{} which allows only real-valued variables.
}{:}

\begin{lemma}[Correctness of \texttt{ctrl}]
\label{lem:java_ctrl_correctness}
The following \texttt{JavaDL} formula is valid:\footnote{For the actual proof we require further assumptions which we summarize as $\texttt{heap\_assumptions}$. This concerns the correct initialization of the heap, the correct instantiation of Java objects, etc. While this paper denotes object fields as variables (e.g. \texttt{A}), these values are technically obtained via a rigid function which we omit for clarity~\cite{KeYBook}.}
\begin{align*}
&\underbrace{
\left(
\begin{array}{l}
    \texttt{heap\_assumptions}~\land\\
    \texttt{A} > 0 \land \texttt{A}^-\doteq\texttt{A}~\land\\
    \texttt{B} < 0 \land \texttt{B}^-\doteq\texttt{B}~\land\\
    \texttt{T} > 0 \land \texttt{T}^-\doteq\texttt{T}
\end{array}
\right)
}_{\texttt{ctrlPre}}
 \implies
\modBox{\texttt{ctrl}}
\underbrace{
\left(
\begin{array}{l}
    \texttt{heap\_assumptions}~\land\\
    \texttt{B} \leq \texttt{this.acc} \leq \texttt{A}~\land\\
    \left(\texttt{brake\_cond} \implies \texttt{this.acc}\doteq\texttt{B}\right)~\land\\
    \texttt{A}^-\doteq\texttt{A} \land \texttt{B}^-\doteq\texttt{B} \land \texttt{T}^-\doteq\texttt{T}
\end{array}
\right)
}_{\texttt{ctrlPost}}
\end{align*}
with \texttt{brake\_cond} abbreviating $
10^4(\texttt{v}+1)^2 + (\texttt{A}-\texttt{B})\texttt{T}(\texttt{A}\texttt{T} + 200(\texttt{v}+1)) > -2 \cdot 10^4\texttt{B}(\texttt{p}-1)$.
\end{lemma}

\paragraph{Differential Dynamic Logic.}
Using the \del{}{proof }calculus implemented in the theorem prover \ac{keymaerax},
we can prove statements about hybrid programs.
For example, we can prove the \chg{L9}{validity of the }{following }formula\chg{L9}{ in \Cref{lem:dl_env_correctness}}{ valid} which shows (under some assumptions) that  \textit{x} will always be at most \textit{s} (the stop sign's position) and an additional invariant after execution of \textit{env}
(proof mechanized in \ac{keymaerax}).
\new{L9}{While \dL{} provides us with an effective infrastructure to reason about the environment's differential equations, it is unfit to reason about Java's heap and integer arithmetic.}
\begin{lemma}[Evolution of \textit{env}]
\label{lem:dl_env_correctness}
The following formula is valid in \dL{}:
\begin{align*}
&\underbrace{\left(
\begin{array}{l} 
\textit{A} > 0~\land~
\textit{B} > 0~\land~
\textit{T} > 0~\land\\
\mathrm{acc\_assumptions}~\land~
\textit{x} + \textit{v}^2/(2\textit{B}) \leq \textit{s} 
\end{array}
\right)}_{\textit{envPre}}
\implies
\modBox{\textit{env}}
\underbrace{\big(
    \textit{x}\leq\textit{s} \land
    \overbrace{\textit{x} + \textit{v}^2/(2\textit{B}) \leq \textit{s}}^{\text{invariant}}
\big)}_{\textit{envPost}}
\end{align*}
with the following abbreviation for $\mathrm{acc\_assumptions}$:
\begin{equation*}
\left(
    -\textit{B} \leq \textit{a} \land  \textit{a} \leq \textit{A} \land
    \left(
        \left(\textit{x} + {\left(\textit{v}^2\right)}/{2\textit{B}} + \left({\textit{A}}/{\textit{B}} + 1\right)\left({\textit{A}\textit{T}^2}/{2} +\textit{T}\textit{v}\right)  > \textit{s}\right)
        \implies
        \textit{a} \doteq -\textit{B}
    \right)
\right)\lor
\left(
    \textit{v} \doteq 0 \land \textit{a} \doteq 0
\right)
\end{equation*}
\end{lemma}

\paragraph{Heterogeneous Safety.}
Both \dL{} and JavaDL are instantiations of dynamic theories and we can hence construct a fully heterogeneous dynamic theory $\dynamicTheory^{(d\texttt{Java}\mathcal{L})}$, here called \emph{differential JavaDL}, over $\dynamicTheory^{(\dL{})}$ and $\dynamicTheory^{(\texttt{JavaDL})}$.
To this end, we add heterogeneous atoms to $\common{\signatureAtoms}$ with respective evaluation $\common{\atomEval}$ for any integer variable \texttt{v} in JavaDL and any real arithmetic variable \textit{w} in \dL{}:
\begin{itemize}
    \item $\intToReal\left(\texttt{v},\textit{w}\right)$ 
    with
    $\stateOne\in\common{\atomEval}\mleft(\intToReal\left(\texttt{v},\textit{w}\right)\mright)$ iff 
    $\mathbb{Z} \ni \hero{\variableEval}\left(\stateOne,\texttt{v}\right) =\hero{\variableEval}\left(\stateOne,\textit{w}\right)$
    \item $\mathrm{round}\left(\texttt{v},\textit{w}\right)$ 
    with
    $\stateOne\in\common{\atomEval}\mleft(\mathrm{round}\left(\texttt{v},\textit{w}\right)\mright)$ iff
    $\hero{\variableEval}\left(\stateOne,\texttt{v}\right) \in \mathbb{Z}$ and
    $\hero{\variableEval}\left(\stateOne,\texttt{v}\right) - 0.5 <
    \hero{\variableEval}\left(\stateOne,\textit{w}\right)$ \new{S44}{as well as}
    $ \hero{\variableEval}\left(\stateOne,\textit{w}\right)
    \leq \hero{\variableEval}\left(\stateOne,\texttt{v}\right) + 0.5$
\end{itemize}
\chg{L9}{The}{These} atoms allow us to explicitly model the communication between $\dynamicTheory^{(\dL{})}$ and $\dynamicTheory^{(\texttt{JavaDL})}$
\chg{L9}{and are}{which is} all we need to model the full heterogeneous system in $\dynamicTheory^{(d\texttt{Java}\mathcal{L})}$.
After executing the controller we \chg{L9}{set}{anonymize} the \dL{} variable $a$\new{L9}{ to an arbitrary value} and \chg{L9}{then assert it has the value}{set it to the value} of the Java variable in \texttt{this.acc} via a check (using the atomic formula $\intToReal$).
Based on the \chg{L9}{new}{updated} acceleration we execute the environment \textit{env} and \del{L9}{similarly }retrieve the \del{L9}{new }position and velocity\new{L9}{ via $\mathrm{round}$}.
This can be formalized as
\begin{align*}
&\texttt{ctrl};~\hero{\alpha} \equiv~
\texttt{ctrl};~
\left(
\textit{a} \coloneqq *;~
?\left(
\intToReal\left(\texttt{this.acc},a*100\right)
\right);~
\textit{env};~
\texttt{p}\coloneqq *;~
\texttt{v} \coloneqq *;~
?\left(
\mathrm{coupling}
\right)
\right)
\end{align*}
\looseness=-1
where we abbreviate $\mathrm{coupling}\equiv\left(\mathrm{round}\left(\texttt{p}, 100\left(\textit{s}-\textit{x}\right)\right) \land
\mathrm{round}\left(\texttt{v}, 100\textit{v}\right)\right)$.
\new{L9}{%
For our safety proof, we define the following coupled precondition which defines admissible initial states of our heterogeneous system:
\begin{align*}
    \mathrm{coupledPre} ~\equiv~&
    \left(\texttt{ctrlPre} ~\land~ \phi_2 \land \phi_3\right) &
    \phi_2~\equiv~
& \left( \textit{A}>0 ~\land~ \textit{B}>0 ~\land~ \textit{T}>0 ~\land~ x + v^2/(2*B) \leq s\right)\\
    \phi_3~\equiv~
&
    \rlap{$\big(\intToReal\left(\texttt{A}^-, 100\textit{A}\right)  ~\land~
     \intToReal\left(\texttt{B}^-, -100\textit{B}\right)~\land~
     \intToReal\left(\texttt{T}^-, 100*\textit{T}\right)~\land~\mathrm{coupling}\big)$}%
\end{align*}
Note that 
$\texttt{ctrlPre}\in\formulaSet{\dynamicTheory^{(\texttt{JavaDL})}}$,
$\phi_2\in\formulaSet{\dynamicTheory^{(\dL{})}}$, and
$\phi_3\in\formulaSet{\dynamicTheory^{(d\texttt{Java}\mathcal{L})}}$,
i.e.\ $\texttt{ctrlPre}$ and $\phi_2$ constrain the homogeneous states, while $\phi_3$ synchronizes them.
}%
Using our heterogeneous calculus we can then prove that our heterogeneous system never overshoots the stop sign:
\begin{lemma}[Safety of the Heterogeneous System]
\label{lem:case-study-result}
The following formula is valid in $\dynamicTheory^{(d\texttt{Java}\mathcal{L})}$:
\begin{align}
    \label{eq:valid_formula_example}
    \mathrm{coupledPre}
    \implies
    \modBox{\left(\texttt{ctrl};~\hero{\alpha}\right)^*}
    x \leq s
\end{align}
\end{lemma}
\begin{proof}[Proof Sketch]
We provide a sketch for the proof of \Cref{lem:case-study-result} below (Proof A and B). %
Beyond the heterogeneous calculus rules presented above, we use the rules from \Cref{fig:additional_axioms}
which are derivable in our calculus (see \ifextension{\Cref{sec:additional-proof-rules}}{Appendix A}).

\begin{minipage}{0.6\textwidth}
\textbf{\footnotesize Proof A}\\
\resizebox{\linewidth}{!}{%
\begin{minipage}{10cm}
\begin{prooftree}
\footnotesize
\AxiomC{\textit{see }\textbf{Proof B}}
\UnaryInfC{%
$
\mathrm{coupledPre}
\implies
\modBox{
\texttt{ctrl}}\modBox{
\hero{\alpha}
}
\mathrm{coupledPre}
$}
\LeftLabel{\ref{axiom:seq}}
\UnaryInfC{$
\mathrm{coupledPre}
\implies
\modBox{
\texttt{ctrl};~
\hero{\alpha}
}
\mathrm{coupledPre}
$}
\LeftLabel{\ref{axiom:indMain}}
\UnaryInfC{$
\mathrm{coupledPre}
\implies
\modBox{\left(\texttt{ctrl};~\hero{\alpha}\right)^*}
\mathrm{coupledPre}
$}
\AxiomC{$*$ (FOL)}
\UnaryInfC{$B>0 \land x+v^2/2B \leq s \rightarrow x \leq s$}
\UnaryInfC{$\cdots$}
\UnaryInfC{$
\mathrm{coupledPre}
\implies
x \leq s
$}
\LeftLabel{\ref{axiom:MRMain}}
\BinaryInfC{$
\mathrm{coupledPre}
\implies
\modBox{\left(\texttt{ctrl};~\hero{\alpha}\right)^*}
x \leq s
$}
\end{prooftree}
\end{minipage}
}
\end{minipage}

\textbf{\footnotesize Proof B}\\
\vspace*{-1cm}\begin{prooftree}
\footnotesize
\AxiomC{$*$ (\KeY: \Cref{lem:java_ctrl_correctness})}
\LeftLabel{\ref{axiom:HR0}}
\UnaryInfC{$\texttt{ctrlPre}
\implies
\modBox{\texttt{ctrl}} \texttt{ctrlPost}$}
\LeftLabel{\ref{axiom:FiMain}}
\UnaryInfC{$\texttt{ctrlPre} \land \phi_2
\implies
\modBox{\texttt{ctrl}} \left(\texttt{ctrlPost} \land \phi_2\right)$}
\AxiomC{$*$ (FOL)}
\UnaryInfC{%
$\mathrm{coupledPre}
\implies
\phi_3$}
\LeftLabel{\ref{axiom:V}}
\UnaryInfC{%
$\mathrm{coupledPre}
\implies
\modBox{\texttt{ctrl}} \phi_3$}
\LeftLabel{\ref{axiom:boxAndMain}}
\BinaryInfC{%
\parbox{6cm}{
$
\mathrm{coupledPre}
\implies
\modBox{\text{\texttt{ctrl}}} \left((\text{\texttt{ctrlPost}} \land \phi_2) \land \phi_3\right)
$
\vspace*{-0.3cm}
}
}
\def\extraVskip{2pt}
\AxiomC{\parbox{2cm}{
\textit{using \Cref{lem:dl_env_correctness}}\\
\textit{see \ifextension{\Cref{apx:proof_example}}{Appendix C.3}}
}}
\UnaryInfC{\parbox{3cm}{$
\left((\texttt{ctrlPost} \land \phi_2) \land \phi_3\right)$\\
$
\implies
\modBox{
\hero{\alpha}
}
\mathrm{coupledPre}
$}}
\LeftLabel{\ref{axiom:MRMain}}
\BinaryInfC{%
$
\mathrm{coupledPre}
\implies
\modBox{
\texttt{ctrl}}\modBox{
\hero{\alpha}
}
\mathrm{coupledPre}
$}
\end{prooftree}
\vspace*{-0.5cm}
\end{proof}
\begin{figure}[t]
\fbox{\begin{minipage}[t]{0.97\textwidth}\vspace{0pt}
\def\fCenter{~\vdash~}
\hspace*{-0.5em}%
\begin{minipage}[t]{0.53\textwidth}
\begin{axiomBlock}
\refAxiom{boxAndMain}{\ensuremath{\left[\right]\land}}{
    $
    \left(
    \modBox{\schemaProgramOne}{\schemaVarOne} \land \modBox{\schemaProgramOne}{ \schemaVarTwo}
    \right)
    \iff
    \modBox{\schemaProgramOne}{\left(\schemaVarOne \land \schemaVarTwo\right)}
    $}
\refAxiom{FiMain}{\ensuremath{F^{(i)}}}{
    $
    \left(\schemaVarTwo \implies \modBox{\schemaProgramOne}{\schemaVarThree}\right)
    \implies
    \left(
    \schemaVarOne \land \schemaVarTwo
    \implies
    \modBox{\schemaProgramOne}{\left(\schemaVarThree \land \schemaVarOne\right)}
    \right)
    $\\
    \hspace*{-0.5cm}\textit{where $\schemaProgramOne\in\signaturePrograms^{(1-i)},\schemaVarOne \in \formulaSet{\dynamicSignature^{(i)}}$ 
    and $i \in \left\{0,1\right\}$}
    }
\end{axiomBlock}
\end{minipage}%
\hspace*{-1em}%
\begin{minipage}[t]{0.27\textwidth}
\begin{prooftree}
    \AxiomC{$\modBox{\schemaProgramOne}{\schemaVarOne}$}
    \AxiomC{\hspace*{-1em}$\left(\schemaVarOne \implies \schemaVarTwo\right)$}
    \LeftLabel{\axiomLabel{MRMain}{MR} \normalfont (MR)}
    \BinaryInfC{$\modBox{\schemaProgramOne}{\schemaVarTwo}$}
\end{prooftree}
\end{minipage}%
\hspace*{0.5em}%
\begin{minipage}[t]{0.2\textwidth}
\begin{prooftree}
    \AxiomC{$\schemaVarOne \implies \modBox{\schemaProgramOne}{\schemaVarOne}$}
    \LeftLabel{\axiomLabel{indMain}{ind} \normalfont (ind)}
    \UnaryInfC{$\schemaVarOne \implies \modBox{\schemaProgramOne^*}{\schemaVarOne}$}
\end{prooftree}
\end{minipage}
\end{minipage}}
\caption{Additional proof rules derived from previously defined axioms (see \ifextension{\Cref{sec:additional-proof-rules}}{Appendix A}).}
\label{fig:additional_axioms}
\end{figure}%
\noindent
\looseness=-1
The validity of this formula guarantees safety ($x\leq s$) for any number of control-environment loop iterations assuming $\mathrm{coupledPre}$\del{L9}{ (see Appendix X)}.
First, note that there is no straight forward way to even specify this desired behavior without HDL as it represents a system level guarantee only emerging from the looped interaction of Java and the hybrid program.
Secondly, note that our calculus proves this result by \emph{decomposing} proof obligations until the \texttt{JavaDL} and \dL{} lemmas above are applicable.
This underscores the power of our logic to state and prove statements about systems that do not neatly live in the world of a single dynamic theory, but require their interaction.

\paragraph{Limitations}
\new{L15}{
Our case study demonstrates that practical proofs reach a scale where pen-and-paper formalization alone becomes infeasible.
In future work, we plan to extend our Isabelle formalization into a practical tool for HDL proofs.
Our approach is also fundamentally limited by its definition of a dynamic theory.
For example, the validity of Barcan's axiom \ref{axiom:B} implies that a variable's domain does not grow or shrink. However, this limitation can typically be circumvented by introducing predicates stating whether an object actually exists~\cite{DBLP:journals/corr/abs-1206-3357}.
}

\section{Conclusion}
\looseness=-1
This work introduces \emph{Heterogeneous Dynamic Logic}, a compositional framework for the verification of heterogeneous programs.
Similarly to how Satisfiability Modulo Theories modularly combines data logics to construct more expressive, combined first-order theories, Heterogeneous Dynamic Logic modularly combines program logics to construct more expressive, combined heterogeneous program logics.
The resulting heterogeneous dynamic theories enable not only the \emph{specification} of otherwise hard to formalize properties (see \Cref{sec:exemplary_case_study_intro,case_study_guarantee}),
but also their \emph{verification} by providing a proof calculus that decomposes problems in a manner that is compatible with existing homogeneous verification infrastructure -- including, both, classical dynamic logic calculi and KAT-based equational reasoning.
\Cref{sec:rel-complete} shows, under common assumptions, that heterogeneous proof calculi inherit not just rules and axioms but also relative completeness.
This proves that verifying heterogeneous programs is no harder than verifying properties about the homogeneous programs and their shared first-order structure (i.e., their data logic).
Our proof theory of heterogeneous systems will not only enable the combination of existing \chg{S45}{verification methodologies}{verificaiton methodolgies}, but also serves as a vehicle for future dynamic logic or KAT-based verification techniques, which can seamlessly integrate with the existing ecosystem through HDL.

\section*{Data Availability}
We provide our extensible Isabelle formalization on Zenodo~\cite{samuel_teuber_2026_19075933}.

\section*{Acknowledgments}
This work was supported by funding from the pilot program Core-Informatics of the Helmholtz Association (HGF) and by an Alexander von Humboldt Professorship.

\printbibliography

\ifextension{
\clearpage
\appendix

\section{Additional Proof Rules}
\label{sec:additional-proof-rules}

\subsection{Additional Proof Rules for Dynamic Theories}
\FloatBarrier
\begin{figure}[t]
\begin{minipage}{0.65\textwidth}
\begin{axiomBlock}
    \refAxiom{boxAnd}{\ensuremath{\left[\right]\land}}{
    $
    \left(
    \modBox{\schemaProgramOne}{\schemaVarOne} \land \modBox{\schemaProgramOne}{ \schemaVarTwo}
    \right)
    \iff
    \modBox{\schemaProgramOne}{\left(\schemaVarOne \land \schemaVarTwo\right)}
    $}
    \refAxiom{KDiamond}{\ensuremath{K_{\langle\rangle}}}{
    $
    \modBox{\schemaProgramOne}\left(\schemaVarOne \implies \schemaVarTwo\right) \implies
    \modDia{\schemaProgramOne}{\schemaVarOne} \implies \modDia{\schemaProgramOne}{\schemaVarTwo}
    $}
    \refAxiom{MPDiamond}{\ensuremath{MP_{\langle\rangle}}}{
    $
    \modDia{\schemaProgramOne}{\schemaVarTwo} \land
    \fa{\variableVecOne}{
    \left(\schemaVarTwo \implies \schemaVarOne\right)}
    \implies
    \modDia{\schemaProgramOne}{\schemaVarOne}
    $\\
    where $\variableVecOne \supseteq \boundVarsSem\mleft(\schemaProgramOne\mright)$
    }
\end{axiomBlock}
\end{minipage}%
\begin{minipage}{0.35\textwidth}
\begin{prooftree}
    \AxiomC{$\schemaVarOne \implies \schemaVarTwo$}
    \LeftLabel{\axiomLabel{M}{M}\normalfont (M)}
    \UnaryInfC{$\modBox{\schemaProgramOne}{\schemaVarOne}
    \implies
    \modBox{\schemaProgramOne}{\schemaVarTwo}$}
\end{prooftree}
\begin{prooftree}
    \AxiomC{$\modBox{\schemaProgramOne}{\schemaVarOne}$}
    \AxiomC{$\left(\schemaVarOne \implies \schemaVarTwo\right)$}
    \LeftLabel{\axiomLabel{MR}{MR} \normalfont (MR)}
    \BinaryInfC{$\modBox{\schemaProgramOne}{\schemaVarTwo}$}
\end{prooftree}
\end{minipage}
    \caption{Additional proof rules for Dynamic Theories. All rules pictured here are derivable from the axioms of Dynamic Logic established in \Cref{sec:DL}.}
    \label{fig:DL_rules_helper}
\end{figure}

\begin{lemma}[Helper Rules]
    \label{lem:DL_rules_helper}
    The axioms \ref{axiom:boxAnd}, \ref{axiom:KDiamond}, and \ref{axiom:MPDiamond} as well as the proof rules \ref{axiom:M} and \ref{axiom:MR} in \Cref{fig:DL_rules_helper} are sound and can be derived from the proof rules of Dynamic Theories established in \Cref{sec:DL}.
\end{lemma}
\begin{proof}
The proofs presented here are based on proofs for Modal Logic from the literature~\cite{cresswell2012new}.\\
The rule \ref{axiom:boxAnd} is sound and can be derived:\\
$\Rightarrow$\\
    \begin{prooftree}

        \AxiomC{$*$}
        \UnaryInfC{$\schemaVarOne\implies\schemaVarTwo\implies\left(\schemaVarOne\land\schemaVarTwo\right)$}
        \LeftLabel{\ref{axiom:G}}
        \UnaryInfC{$\modBox{\schemaProgramOne}{\left(\schemaVarOne\implies\schemaVarTwo\implies\left(\schemaVarOne\land\schemaVarTwo\right)\right)}$}
        \LeftLabel{\ref{axiom:K}}
        \UnaryInfC{$\modBox{\schemaProgramOne}{\schemaVarOne}\implies\modBox{\schemaProgramOne}{\left(\schemaVarTwo\implies\left(\schemaVarOne\land\schemaVarTwo\right)\right)}$}
        \LeftLabel{\ref{axiom:K}}
        \UnaryInfC{$\modBox{\schemaProgramOne}{\schemaVarOne} \implies \left(\modBox{\schemaProgramOne}{\schemaVarTwo}\implies\modBox{\schemaProgramOne}{\left(\schemaVarOne\land\schemaVarTwo\right)}\right)$}
    \end{prooftree}
    Elementary logic transformations yield that this is equivalent to
    $\left(\modBox{\schemaProgramOne}{\schemaVarOne} \land \modBox{\schemaProgramOne}{\schemaVarTwo}\right)\implies\modBox{\schemaProgramOne}{\left(\schemaVarOne\land\schemaVarTwo\right)}$.\\
$\Leftarrow$\\
First, we show the property for $\phi$; subsequently the same proof can be repeated for $\psi$ and subsequently be joined:
    \begin{prooftree}
        \AxiomC{$*$}
        \UnaryInfC{$\schemaVarOne \land \schemaVarTwo \implies \schemaVarOne$}
        \LeftLabel{\ref{axiom:G}}
        \UnaryInfC{$\modBox{\schemaProgramOne}{\left(\schemaVarOne \land \schemaVarTwo \implies \schemaVarOne\right)}$}
        \LeftLabel{\ref{axiom:K}}
        \UnaryInfC{$\modBox{\schemaProgramOne}{\left(\schemaVarOne \land \schemaVarTwo\right)} \implies \modBox{\schemaProgramOne}{\schemaVarOne}$}
    \end{prooftree}
The rule \ref{axiom:M} is sound and can be derived:
\begin{prooftree}
    \AxiomC{$ \left(\schemaVarOne \implies \schemaVarTwo\right)$}
    \LeftLabel{\ref{axiom:G}}
    \UnaryInfC{$ \modBox{\schemaProgramOne}{\left(\schemaVarOne \implies \schemaVarTwo\right)}$}
    \LeftLabel{\ref{axiom:K}}
    \UnaryInfC{$ \modBox{\schemaProgramOne}{\schemaVarOne} \implies \modBox{\schemaProgramOne}{\schemaVarTwo}$}
\end{prooftree}
The proof for \ref{axiom:KDiamond} is based on work by Platzer~\cite[Appendix B.4]{Platzer2012}:
We first obtain $\modBox{\schemaProgramOne} \left(\neg\schemaVarTwo \implies \neg\schemaVarOne\right)\implies \left(\modBox{\schemaProgramOne}\neg\schemaVarTwo \implies \modBox{\schemaProgramOne}\neg\schemaVarOne\right)$ via \ref{axiom:K}.
Via propositional reasoning this can be rewritten as:
$\modBox{\schemaProgramOne} \left(\neg\schemaVarTwo \implies \neg\schemaVarOne\right)\implies \left(\neg\modBox{\schemaProgramOne}\neg\schemaVarOne \implies \neg\modBox{\schemaProgramOne}\neg\schemaVarTwo\right)$.
Via duality this is the same as 
$\modBox{\schemaProgramOne} \left(\neg\schemaVarTwo \implies \neg\schemaVarOne\right)\implies \left(\modDia{\schemaProgramOne}{\schemaVarOne} \implies \modDia{\schemaProgramOne}{\schemaVarTwo}\right)$.
Via $\left(\schemaVarOne \implies \schemaVarTwo\right) \implies \left(\neg\schemaVarTwo \implies \neg\schemaVarOne\right)$ and \ref{axiom:G} with \ref{axiom:K} we derive that $\modBox{\schemaProgramOne}\left(\schemaVarOne\implies\schemaVarTwo\right) \implies \left(\modDia{\schemaProgramOne}{\schemaVarOne} \implies \modDia{\schemaProgramOne}{\schemaVarTwo}\right)$.

The rule \ref{axiom:MR} is sound and can be derived by cutting $\modBox{\schemaProgramOne}{\schemaVarOne}$ into the proof and subsequently applying \ref{axiom:M} (proven sound above).\\

For the proof of \ref{axiom:MPDiamond} we start with the formulas $\modDia{\schemaProgramOne}{\schemaVarTwo}$ and $\fa{\variableVecOne}{\left(\schemaVarTwo \implies \schemaVarOne\right)}$.
Via \ref{axiom:V} (all bound variables of $\schemaProgramOne$ are not free in the formula) we get that the latter formula implies 
$\modBox{\schemaProgramOne} \fa{\variableVecOne}{\left(\schemaVarTwo \implies \schemaVarOne\right)}$.
We can then instantiate the variables in $\variableVecOne$ with their respective non-bound versions which yields 
$\modBox{\schemaProgramOne} \left(\schemaVarTwo \implies \schemaVarOne\right)$.
Via \ref{axiom:KDiamond} (proven above) we then get $\modDia{\schemaProgramOne}{\schemaVarTwo} \implies \modDia{\schemaProgramOne}{\schemaVarOne}$.
Therefore, the following formula is valid: $\modDia{\schemaProgramOne}{\schemaVarTwo} \land \left(\modDia{\schemaProgramOne}{\schemaVarTwo} \implies \modDia{\schemaProgramOne}{\schemaVarOne}\right) \implies \modDia{\schemaProgramOne}{\schemaVarOne}$.
This in turn implies $\modDia{\schemaProgramOne}{\schemaVarOne}$.

\end{proof}

\begin{lemma}[Pullback Axiom]
The following axiom is sound for any dynamic theory $\dynamicTheory$ and derivable from the axioms and proof rules in \Cref{sec:DL}:
    \begin{axiomBlock}
    \refAxiom{PB}{PB}{
        $
        \fa{\variableVecOneNext}{
        \left(\schemaVarTwo \implies \modDia{\schemaProgramOne}{\schemaVarThree}\right)}
        \land
        \ex{\variableVecOneNext}{
        \left(
        \schemaVarTwo \land
        \fa{\variableVecOne}{
        \left(\schemaVarThree \implies \schemaVarOne\right)}
        \right)}
        \implies
        \left(
        \modDia{\schemaProgramOne}{\schemaVarOne}
        \right)
        $
    \hfill
    $
    \begin{array}{ll}
         \variableVecOne &\supseteq \left(\boundVarsSem\mleft(\schemaProgramOne\mright) \cup \freeVarsSem\left(\modDia{\schemaProgramOne}{\schemaVarOne}\right)\right)\\
    \variableVecOneNext &= \freeVarsSem\left(\schemaVarTwo\land\schemaVarThree\right) \setminus \variableVecOne
    \end{array}
    $}
    \end{axiomBlock}
\end{lemma}
\begin{proof}
Starting from 
$\fa{\variableVecOneNext}{
\left(\schemaVarTwo \implies \modDia{\schemaProgramOne}{\schemaVarThree}\right)}$
and
$
\ex{\variableVecOneNext}{
\left(
\schemaVarTwo \land
\fa{\variableVecOne}{
\left(\schemaVarThree \implies \schemaVarOne\right)}
\right)}$
we can derive via first-order reasoning that this implies
$
\ex{\variableVecOneNext}{
\left(
\schemaVarTwo \land
\fa{\variableVecOne}{
\left(\schemaVarThree \implies \schemaVarOne\right)}
\land
\left(\schemaVarTwo \implies \modDia{\schemaProgramOne}{\schemaVarThree}\right)
\right)}
$
via propositional recombination this can be turned into
$
\ex{\variableVecOneNext}{
\left(
\schemaVarTwo 
\land
\left(\schemaVarTwo \implies \left(
\modDia{\schemaProgramOne}{\schemaVarThree} \land
\fa{\variableVecOne}{
\left(\schemaVarThree \implies \schemaVarOne\right)}
\right)
\right)
\right)}
$.
Using \ref{axiom:MPDiamond} this implies
$
\ex{\variableVecOneNext}{
\left(
\schemaVarTwo 
\land
\left(\schemaVarTwo \implies
\modDia{\schemaProgramOne}{\schemaVarOne}
\right)
\right)}
$.
Via (local) Modus Ponens this implies
$
\ex{\variableVecOneNext}{
\left(
\modDia{\schemaProgramOne}{\schemaVarOne }
\right)}
$.
Since $\variableVecOneNext \notin \freeVarsSem\left(\modDia{\schemaProgramOne}{\schemaVarOne} \right)$ this in turn implies $\modDia{\schemaProgramOne}{\schemaVarOne}$
\end{proof}

\subsection{Additional Proof Rules for Regular Programs}
\begin{lemma}[Loop Induction]
For any dynamic theory \dynamicTheory, the following axiom is sound in $\regular{\dynamicTheory}$ and derivable from the axioms in \Cref{sec:liftedDL:regular}:
\begin{prooftree}
    \AxiomC{$\schemaVarOne \implies \modBox{\schemaProgramOne}{\schemaVarOne}$}
    \LeftLabel{\axiomLabel{ind}{ind} \normalfont (ind)}
    \UnaryInfC{$\schemaVarOne \implies \modBox{\schemaProgramOne^*}{\schemaVarOne}$}
\end{prooftree}    
\end{lemma}
\begin{proof}
The proof can be performed as follows:
\begin{prooftree}
    \AxiomC{$\schemaVarOne \implies \modBox{\schemaProgramOne}{ \schemaVarOne}$}
    \LeftLabel{\ref{axiom:G}}
    \UnaryInfC{$\modBox{\schemaProgramOne^*}{ \left(\schemaVarOne \implies \modBox{\schemaProgramOne}{\schemaVarOne}\right)}$}
    \LeftLabel{\ref{axiom:I}}
    \UnaryInfC{$\schemaVarOne \implies \modBox{\schemaProgramOne^*}{ \schemaVarOne}$}
\end{prooftree}%
\end{proof}
\subsection{Additional Proof Rules for Heterogeneous Dynamic Theories}
\begin{lemma}[Soundness of Frame Rule]
\label{lem:derived_rules}
For any heterogeneous dynamic theory $\hero{\dynamicTheory}$ over dynamic theories $\zero{\dynamicTheory},\one{\dynamicTheory}$ the following axiom is sound and can be derived from the rules from \Cref{sec:hdl}:
\begin{axiomBlock}
    \refAxiom{Fi}{\ensuremath{F^{(i)}}}{
    $
    \left(\schemaVarTwo \implies \modBox{\schemaProgramOne}{\schemaVarThree}\right)
    \implies
    \left(
    \schemaVarOne \land \schemaVarTwo
    \implies
    \modBox{\schemaProgramOne}{\left(\schemaVarThree \land \schemaVarOne\right)}
    \right)
    $\\
    where $\schemaProgramOne\in\signaturePrograms^{(1-i)},\schemaVarOne \in \formulaSet{\dynamicSignature^{(i)}}$ 
    and $i \in \left\{0,1\right\}$
    }
\end{axiomBlock}
\end{lemma}
\begin{proof}
The proofs assume the previous derivation of the rules from \Cref{lem:DL_rules_helper}\\
After application of \ref{axiom:boxAnd}, the left side reads $\schemaVarOne \land \schemaVarTwo \implies \modBox{\schemaProgramOne}\schemaVarThree \land \modBox{\schemaProgramOne}\schemaVarOne$ which we can decompose into two implications.
Observe that $\schemaVarTwo \implies \modBox{\schemaProgramOne}\schemaVarThree$ is our assumption on the right.
For the case of $\schemaVarOne \implies \modBox{\schemaProgramOne}\schemaVarOne$ we can apply the \ref{axiom:V} axiom since $\ith{\signatureVariables} \cap \ith[1-i]{\signatureVariables} = \emptyset$ and $\boundVarsSem\mleft(\schemaProgramOne\mright) \subseteq \ith[1-i]{\signatureVariables}$ while $\freeVarsSem\left(\schemaVarOne\right) \subseteq \ith[i]{\signatureVariables}$.\\
\end{proof}

\begin{lemma}[Ghost Axiom given Equality]
For any heterogeneous dynamic theory $\hero{\dynamicTheory}$ over dynamic theories $\zero{\dynamicTheory},\one{\dynamicTheory}$, the following axiom is sound and can be derived from the rules from \Cref{sec:hdl} under the additional assumption of the axiom $\variableOne \doteq \variableOne$ (for all $\variableOne\in\signatureVariables$):\\
\begin{axiomBlock}
    \refAxiom{ghost}{ghost}{
    $
    \modBox{\variableOne \coloneqq *; ?\left(\variableOne \doteq \variableTwo\right)} \schemaVarOne \implies \schemaVarOne
    $
    (for $\variableOne \notin \freeVarsSem\left(\schemaVarOne\right)$)}
\end{axiomBlock}
\end{lemma}
\begin{proof}
The rule Ghost is sound and can be derived:
Resolving the program on the left using \ref{axiom:havoc} and \ref{axiom:?}, we get an equivalent representation
\[
    \fa{\variableOne}{
    \left(\variableOne \doteq \variableTwo
    \implies
    \schemaVarOne
    \right)}
    \text{ for }\variableOne \notin \freeVarsSem\left(\schemaVarOne\right)
\]
By instantiating $\variableOne$ with $\variableTwo$ we get $\variableTwo \doteq \variableTwo \implies \schemaVarOne$ (since $\variableOne \notin \variablesOf{\schemaVarOne}$). Using Modus Ponens with the axiom $\variableTwo \doteq \variableTwo$ this derives $\schemaVarOne$.%
\end{proof}

\section{A Calculus for Reasoning about Meta-Properties of Dynamic Theories}
\label{apx:rel-complete}
A desirable meta property for dynamic logics is the relative completeness of their axiomatizations~\cite{harel_dynamic_2000,beckert_dynamic_2006,Platzer2012}:
Due to Gödel~\cite{godel1931formal}, for most dynamic logics, one cannot hope for a \emph{complete} axiomatization, as their expressiveness allows the axiomatization of integers~\cite{harel_dynamic_2000} that constitute an undecidable first-order theory.
However, it is often possible to create a finite set of axiom schemata and proof rules that allow us to reduce any validity question over dynamic logic formulas to a validity question over first-order theory formulas for which we then assume the availability of an oracle.
This limited notion of completeness, called \emph{relative completeness}, implies that our calculus ensures that
programs are handled completely, as proving properties about the considered programs is exactly as hard as proving properties about the underlying first-order theory (and no harder).
We formalize this notion in \Cref{def:rel-complete:rel-complete}.

The preceding sections of this paper have proposed a highly compositional approach to the construction of dynamic theories:
Starting from a base theory $\dynamicTheory$, we can generate lifted dynamic theories with more program functionality ($\havoced{\dynamicTheory}$ and $\regular{\dynamicTheory}$) or combine theories ($\preHero{\dynamicTheory}$ and then $\hero{\dynamicTheory}$).
The key advantage of theory lifting and combination was the ability to \emph{reuse} proof calculi and obtain standardized additional proof rules which are valid irrespective of the lifted/combined theory.
The approach outlined above raises the question to what degree meta properties of the calculus, in particular its relative completeness, are equally compositional:
For example, given a relatively complete dynamic theory $\zero{\dynamicTheory}$, can we derive a notion of relative completeness for $\regular{\left(\zero{\dynamicTheory}\right)}$? Or, similarly, can we derive a relative completeness result for its theory combination w.r.t. some secondary dynamic theory $\hero{\dynamicTheory}$?
To this end, we observe that relative completeness proofs in the literature, implicitly or explicitly, follow a highly-compositional structure~\cite{beckert_dynamic_2006,Platzer2012,harel_first-order_1979} that can be leveraged to compositionalize results on the relative completeness of dynamic theory calculi.
Hence, in this section, we describe a number of results that allow us to lift and merge relative completeness results when lifting and combining dynamic theories in a structured manner.

\Cref{sec:DL} outlined properties/assumptions that defined what a dynamic theory is.
In \Cref{sec:liftedDL:regular} we saw, for the first time, a case where we leveraged further assumptions (inductive expressiveness, \Cref{def:liftedDL:inductive_expressive}) to derive stronger calculus rules (axiom \ref{axiom:C}).
We will begin this section by describing a set of additional meta properties, like relative completeness (\Cref{def:rel-complete:rel-complete}), that certain dynamic theories may satisfy (\Cref{subsec:rel-complete:assumptions}).
For these additional meta properties/assumptions, we then derive sufficient conditions under which the meta properties are preserved under theory lifting and/or theory combination (\Cref{subsec:rel-complete:calculus}).
Finally, we leverage the defined meta properties and proven preservation rules to derive relative completeness results for havoc lifting, regular program lifting, and theory combination in (fully) heterogeneous dynamic theories.

\paragraph{Notation and Admissible Variable Values.}
Going forward, we will use underlining of variables to denote (syntactically!) finite vectors of variables or values (e.g., $\variableVecOne$ or $\valueVecOne$).
By slight abuse of notation, we will also use these vectors in quantifiers (e.g., $\fa{\variableVecOne}{\dots}$), which can be reduced to a (syntactically) finite sequence of quantifiers.
By an even worse abuse of notation, we will sometimes use the vectors in set expressions (e.g. $\variableVecOne \subseteq \freeVarsSem\mleft(\schemaVarOne\mright)$), in which case we consider the set of all variables in the vector.

\begin{definition}[Reasoning about admissible variable values]
Given a dynamic theory $\dynamicTheory$ we define $\valueSet: \signatureVariables \to 2^\universe$ as the mapping from variables to all admissible values of a variable (i.e. $\valueSet\mleft(\variableOne\mright) = \left\{ \valueOne\in\universe ~\middle|~ \text{it exists }\stateOne\in\stateSpace \text{ such that }\variableEval\mleft(\stateOne,\variableOne\mright)=\valueOne\right\}$).
We call two variables $\variableOne,\variableOneVariant$ \emph{twin variables} iff $\valueSet\mleft(\variableOne\mright)=\valueSet\mleft(\variableOneVariant\mright)$.
We call two variable vectors $\variableVecOne,\variableVecOneVariant$ \emph{twins} iff the vectors are of equal length and for every component it holds that $\variableVecOne_i$ and $\variableVecOneVariant_i$ are twin variables.
\end{definition}

\subsection{Assumptions for Relative Completeness}
\label{subsec:rel-complete:assumptions}
To transfer relative completeness across liftings and into heterogeneous theories, we must impose one further meaningful restriction about the behavior of programs in our underlying dynamic theory.
We now require that programs not only have a finite read footprint (\Cref{def:DL:prog_eval}), but also a finite write footprint to be able to characterise their behavior:

\begin{definition}[Finite Support]
\label{def:rel-complete:finite-support}
Given a dynamic theory $\dynamicTheory$.
We say $\dynamicTheory$ has \emph{finite support} iff for some function $\boundVarsProgram : \signaturePrograms \to 2^{\signatureVariables}$ it holds that:
\begin{localeAssmBlock}
\localeAssm{$\boundVarsProgram$ Coverage}{
For all programs $\programOne\in\signaturePrograms$ it holds that $\boundVarsSem[\signaturePrograms]\mleft(\programOne\mright) \subseteq \boundVarsProgram\mleft(\programOne\mright)$
}
\localeAssm{$\boundVarsProgram$ Finite}{
For all programs $\programOne\in\signaturePrograms$ 
the result of $\boundVarsProgram\mleft(\programOne\mright)$ is finite
}
\end{localeAssmBlock}
\end{definition}

\paragraph{First Order Expressiveness}
Beyond this assumption on program behavior and the assumption on the availability of a relatively complete proof calculus (\Cref{def:rel-complete:rel-complete}), we also require further tools that are necessary to prove relative completeness.
These assumptions do not always constrain the range of allowed programs or atomic formulas, but for relative completeness results to carry over, the designer of a given dynamic theory must ensure these additional meta properties hold.
The first requirement of this kind is termed \textit{FOL Expressiveness} and ensures that a given dynamic theory is capable of representing the behavior of its programs in the theory's first-order fragment.
While this requirement might seem surprising and unintuitive for someone unfamiliar with dynamic logic relative completeness results, it is, in fact, the bread and butter of meaningful relative completeness proofs~\cite{harel_first-order_1979,beckert_dynamic_2006,Platzer2012}:
\begin{definition}[FOL Expressive]
\label{def:rel-complete:fol-expressive}
Consider a dynamic theory $\dynamicTheory$ with finite support.
We say $\dynamicTheory$ is \emph{FOL Expressive} iff it satisfies the following properties:
\begin{localeAssmBlock}
\localeAssm{Infinite Variables}{For any variable $\variableOne\in\signatureVariables$ there exists an infinite number of twin variables.}
\localeAssm{Eq Predicate}{
There exists a function $\doteq : \signatureVariables^2 \to \folFormulaSet{\dynamicTheory}$ such that for twin variables $\variableOne,\variableTwo\in\integerVariables$ and any state $\stateOne\in\stateSpace$:\\
$\stateOne\in\sem{\doteq\left(\variableOne,\variableTwo\right)}$ iff $
\variableEval\mleft(\stateOne,\variableOne\mright) = 
\variableEval\mleft(\stateOne,\variableTwo\mright)
$
}
\localeAssm{Rendition}{
There exists a function $\programRenditionSym$ such that for all programs $\programOne\in\signaturePrograms$ and all twin variable vectors $\variableVecOne,\variableVecOneNext$ with $\variableVecOne \supseteq \variablesOf{\programOne}$ and $\variableVecOne \cap \variableVecOneNext=\emptyset$ the function $\programRendition{\programOne}\mleft(\variableVecOne,\variableVecOneNext\mright)$ returns a formula in $\folFormulaSet{\dynamicSignature}$ such that:\\
$
\models_{\dynamicTheory} \programRendition{\programOne}\mleft(\variableVecOne,\variableVecOneNext\mright)
\iff
\modDia{\programOne}{\variableVecOne \doteq \variableVecOneNext}
$
}
\end{localeAssmBlock}
\end{definition}
For clarity, we will denote $\doteq\left(\variableOne,\variableTwo\right)$ as $\variableOne\doteq\variableTwo$.
For example, there exist results for rendition of first-order dynamic logic programs~\cite{harel_first-order_1979} or hybrid programs~\cite{Platzer2012} (via relative completeness w.r.t. the discrete fragment of hybrid programs).
We can then derive the following additional axiom on the behavior of equality:
\begin{lemmaE}[Identity Equality][all end]
\label{lem:rel-completeness:id_eq}
For any dynamic theory $\dynamicTheory$ that is FOL Expressive or Communicating, the following axiom holds for any $\variableOne\in\signatureVariables$:
\[
\variableOne \doteq \variableOne.
\]
\end{lemmaE}
\begin{proofE}
Note, that $\variableOne$ is a twin variable of itself.
Hence, by definition we know that for all states $\mu\in\sem{\doteq\left(\variableOne,\variableOne\right)}$ iff $\variableEval\left(\stateOne,\variableOne\right)=\variableEval\left(\stateOne,\variableOne\right)$ which holds in all states.
\end{proofE}

\paragraph{Gödelization.}
Another characteristic of relative completeness proofs in dynamic logic is that they typically rely on the ability of first-order theories to encode both program behavior (see above) \emph{and} finite sequences of domain values into a single value.
The first example of this mechanism is due to Gödel~\cite{godel1931formal} who showed that any finite sequence of integer values can be injectively mapped to a single integer value.
The purpose of these encodings lies in the reasoning about loops:
To prove relative completeness, we require that the first-order fragment of our dynamic theory is capable of expressing the strongest loop invariant.

A more general description of this sequence encoding mechanism was established by Parikh (and first published by Harel~\cite[p. 29]{harel_first-order_1979}):
Intuitively, without precisely formalizing the logic below, we desire a total predicate $R\mleft(x,i,y\mright)$ such that for all $n\in\mathbb{N}$ it holds that (taken from \cite{harel_first-order_1979}):
\[
\left(\forall x_1\dots x_n\right)\left(\exists y\right)\left(\forall w\right)\left(\forall i\right) \left(\left(i\in\mathbb{N} \land 1 \leq i \leq n\right) \Longrightarrow \left(R\mleft(w,i,y\mright) \equiv \left(w=x_i\right)\right)\right)
\]
This definition is w.r.t. a common domain for $x$ and $y$; however, this is not a strict requirement for our dynamic theories where it would also suffice to encode sequences of values for variable $\variableOne$ in the value set of another variable $\variableTwo$ with $\valueSet\mleft(\variableOne\mright)\neq\valueSet\mleft(\variableTwo\mright)$.
Hence, we define Gödelization of a variable $\variableOne$ w.r.t. a second variable $\variableTwo$, which is a representative for the value range into which sequences of values of $\variableOne$ are embedded.
Given twin variables $\variableOneVariant,\variableTwoVariant$ for $\variableOne,\variableTwo$, we then demand that there exists a first-order formula $\goedel{\variableOne}\mleft(\variableTwoVariant,n,j,\variableOneVariant\mright)$ that decodes an $n$-dimensional sequence of values stored in $\variableTwoVariant$, extracts the $j$-th value, and asserts this component is ``stored'' in $\variableOneVariant$.
Formally, this leads to the following definition, which describes that the value range of $\variableOne\in\signatureVariables$ supports Gödelization:
\begin{definition}[Gödelization Support]
\label{def:rel-complete:goedel}
Let $\dynamicTheory$ be some dynamic theory that is inductively expressive w.r.t. $\integerVariables$.
We say $\variableOne\in\signatureVariables$ has \emph{Gödelization Support} w.r.t. $\integerVariables$ iff there exists a variable $\variableTwo\in\signatureVariables$ and some formula $\goedel{\variableOne}$, called \emph{Gödel Formula}, such that:\\
For any $N\in\mathbb{N}$ and any vector of values $\valueVecOne\in\valueSet\mleft(\variableOne\mright)^N$ there exists a value $\valueTwo\in\valueSet\mleft(\variableTwo\mright)$ such that:\\
For any integer variables $i,n\in\integerVariables$ and any twin variables $\variableOneVariant$ (resp. $\variableTwoVariant$) of $\variableOne$ (resp. $\variableTwo$) and any state $\stateOne$ with
$\variableEval\mleft(\stateOne,\variableTwoVariant\mright)=\valueTwo$ and $ \universeToNat\mleft(\variableEval\mleft(\stateOne,n\mright)\mright)=N$ it holds that $\stateOne \models_{\dynamicTheory} \goedel{\variableOne}\mleft(\variableTwoVariant,n,j,\variableOneVariant\mright)$ iff $1\leq k =\universeToNat\mleft(\variableEval\mleft(\stateOne,j\mright)\mright) \leq N$ and $\variableEval\mleft(\stateOne,\variableOneVariant\mright) = \valueVecOne_{k}$.
\end{definition}
Some concrete examples of Gödelization would be:
\begin{itemize}
    \item In interpreted (integer) first-order dynamic logic, the variables support Gödelization using standard Gödel Numbers~\cite{godel1931formal,harel_dynamic_2000}
    \item In differential dynamic logic, real values have a Gödelization 
    for a differentially expressive first-order structure
    (formalizable via our atoms)~\cite{Platzer2012}
    \item Under certain constraints, non-rigid functions (i.e., functions whose evaluation depends on the state) equally admit Gödelization~\cite{beckert_dynamic_2006}
\end{itemize}
For the relative completeness proofs below, we consequently demand that our underlying inductively expressive dynamic theory is also \emph{Gödel Expressive}:
\begin{definition}[Gödel Expressive]
\label{def:rel-complete:goedel-expressive}
Consider a dynamic theory $\dynamicTheory$ that is inductively expressive w.r.t. variables $\integerVariables$.
We say $\dynamicTheory$ is \emph{Gödel Expresive} iff it satisfies the following properties:
\begin{localeAssmBlock}
\localeAssm{Gödel Formulas}{
All variables $\variableOne\in\signatureVariables$ support Gödelization w.r.t. $\integerVariables$ for a provided Gödel Function $\goedel{\variableOne}$ (see \Cref{def:rel-complete:goedel}).
}
\localeAssm{Less Predicate}{
There exists a function $\natLess: \integerVariables^2 \to \folFormulaSet{\dynamicTheory}$ such that for variables $\variableOne,\variableTwo\in\integerVariables$  with $\variableOne \neq \variableTwo$ and any state $\stateOne\in\stateSpace$:\\
$\stateOne\in\sem{\natLess\mleft(\variableOne,\variableTwo\mright)}$ iff $
\universeToNat\mleft(\variableEval\mleft(\stateOne,\variableOne\mright)\mright) < 
\universeToNat\mleft(\variableEval\mleft(\stateOne,\variableTwo\mright)\mright)
$
}
\end{localeAssmBlock}
\end{definition}

\paragraph{Communication in Heterogeneous Dynamic Theories.}
The properties described above concern a single dynamic theory.
We now discuss assumptions under which relative completeness results can also be obtained for heterogeneous dynamic theories as described in \Cref{sec:hdl}.
To this end, one might hope that Gödelization and/or FOL expressiveness might be sufficient to prove relative completeness of fully heterogeneous dynamic theories.
Unfortunately, it turns out that this is not quite sufficient yet:
While subtle, Gödelization requires that the Gödel formulas over all variables use the same set of integer expressive variables.
However, by default, variables $\zero{\variableOne}\in\zero{\signatureVariables}$ from theory $\zero{\dynamicTheory}$ and their Gödel Formulas $\goedel{\zero{\variableOne}}$ will be using a different set of integer expressive variables for Gödelization ($\zero{\integerVariables}$) than variables $\one{\variableOne}\in\one{\signatureVariables}$ from theory $\one{\dynamicTheory}$ which rely on $\one{\integerVariables}$.
The requirement that all Gödelizations use the same set of integer expressive variables is no definitional accident, but indeed imperative for the derivation of loop renditions in regular programs:
Without a common ``integer language'' we cannot encode that both worlds evolve the loop to an equal number of iterations and in lock-step. Hence, we must guarantee a minimal amount of information flow between the two worlds, which must be implemented using suitable (equality) atoms in $\common{\signatureAtoms}$ and corresponding functions mapping variables to $\preHero{\dynamicSignature}$-formulas:
\begin{definition}[Communicating Theories]
\label{def:hdl:communicating}
Let $\zero{\dynamicTheory},\one{\dynamicTheory}$ be two inductively expressive, FOL expressive dynamic theories with integer expressive sets $\zero{\integerVariables},\one{\integerVariables}$ which together form $\preHero{\dynamicTheory}$.
We say $\zero{\dynamicTheory},\one{\dynamicTheory}$ are \emph{communicating in $\preHero{\dynamicTheory}$} iff the following properties are satisfied:
\begin{localeAssmBlock}
    \localeAssm{$\mathbb{Z}$ Communcation}{
    There exists a function $\preHero{\natEq}:\zero{\integerVariables}\times\one{\integerVariables} \to \folFormulaSet{\preHero{\dynamicSignature}}$ such that for variables $\variableOne\in\zero{\integerVariables}, \variableTwo\in\one{\integerVariables}$ and any state $\stateOne\in\stateSpace$:\\
    $\stateOne\in\sem[\preHero{\dynamicTheory}]{\preHero{\natEq}\mleft(\variableOne,\variableTwo\mright)}$ iff
    $\zero{\universeToNat}\mleft(\preHero{\variableEval}\mleft(\stateOne,\variableOne\mright)\mright) = \one{\universeToNat}\mleft(\preHero{\variableEval}\mleft(\stateOne,\variableTwo\mright)\mright)$
    }
    \localeAssm{Eq Predicate}{
    $\preHero{\dynamicTheory}$ satisfies the \textsc{\bf Eq Predicate} property from \Cref{def:rel-complete:fol-expressive} on $\preHero{\dynamicTheory}$.
    }
\end{localeAssmBlock}
\end{definition}

\begin{figure}[t]
    \centering
\fbox{\begin{minipage}[t]{\linewidth}\vspace{0pt}
\begin{minipage}[t]{0.35\linewidth}\vspace{0pt}
\begin{align*}
\intertext{\textbf{By Definition:}}
    \proofProp{FolExp}{\Delta} \Longrightarrow& \proofProp{Fin}{\Delta}
\intertext{\textbf{\Cref{lem:rel-complete}}:}
    \proofProp{FolExp}{\Delta} \land
    \proofProp{FolRelCom}{\Delta}
    \Longrightarrow& \proofProp{RelCom}{\Delta}
\end{align*}
\end{minipage}\begin{minipage}[t]{0.552\linewidth}\vspace{0pt}
\begin{align*}
\intertext{\textbf{Results for Havoc Lifting:}}
    \proofProp{Fin}{\Delta} \Longrightarrow& \proofProp{Fin}{\havoced{\Delta}} & \text{\Cref{lem:rel-completeness:trivial}}\\
    \proofProp{Göd}{\Delta} \Longrightarrow& \proofProp{Göd}{\havoced{\Delta}} & \text{\Cref{lem:rel-completeness:trivial}}\\
    \proofProp{FolExp}{\Delta} \Longrightarrow& \proofProp{FolExp}{\havoced{\Delta}} & \text{\Cref{lem:rel-completeness:fol-expressiveness-havoc}}\\
    \proofProp{RelCom}{\Delta} \Longrightarrow& \proofProp{FolRelCom}{\havoced{\Delta}} & \text{\Cref{lem:rel-completeness:box_diamond_complete}}
\end{align*}
\end{minipage}
\begin{align*}
\intertext{\textbf{Results for Regular Program Lifting:}}
    \proofProp{Fin}{\Delta} \Longrightarrow& \proofProp{Fin}{\regular{\Delta}}
    &\text{\Cref{lem:rel-completeness:trivial-regular}}\\
    \proofProp{Göd}{\Delta} \Longrightarrow& \proofProp{Göd}{\regular{\Delta}}&\text{\Cref{lem:rel-completeness:trivial-regular}}\\
    \proofProp{Göd}{\Delta} \land \proofProp{FolExp}{\Delta} \Longrightarrow& \proofProp{FolExp}{\regular{\Delta}}&\text{\Cref{lem:rel-completeness:fol-expressiveness-regular}}\\
    \proofProp{Göd}{\Delta} \land \proofProp{FolExp}{\Delta} \land \proofProp{RelCom}{\Delta} \Longrightarrow & \proofProp{FolRelCom}{\regular{\Delta}} &
    \text{\Cref{lem:rel-completeness:regular-box}, \ref{lem:rel-completeness:regular-diamond}}
\intertext{\textbf{Results for Heterogeneous Dynamic Theories:}}
    \proofProp{Fin}{\zero{\Delta}} \land \proofProp{Fin}{\one{\Delta}} \Longrightarrow& \proofProp{Fin}{\preHero{\Delta}}
    &\text{\Cref{lem:rel-complte:prehero-BV}}\\
    \proofProp{Göd}{\zero{\Delta}} \land \proofProp{Göd}{\one{\Delta}} \land \proofProp{Co}{\preHero{\Delta}} \Longrightarrow& \proofProp{Göd}{\preHero{\Delta}} & \text{\Cref{lem:rel-completeness:hero-goedel}}\\
    \proofProp{FolExp}{\zero{\Delta}} \land \proofProp{FolExp}{\one{\Delta}} \land \proofProp{Co}{\preHero{\Delta}} \Longrightarrow& \proofProp{FolExp}{\preHero{\Delta}} & \text{\Cref{lem:rel-completeness:hero-fol-expressive}}\\
    \left(\begin{array}{l}
         \proofProp{FolExp}{\zero{\Delta}} \land \proofProp{RelCom}{\zero{\Delta}} \land\\
         \proofProp{FolExp}{\one{\Delta}} \land \proofProp{RelCom}{\one{\Delta}} \land
         \proofProp{Co}{\preHero{\Delta}}
    \end{array}\right)
    \Longrightarrow&
    \proofProp{FolRelCom}{\preHero{\Delta}} & \text{\Cref{lem:rel-completeness:hero-box}, \ref{lem:rel-completeness:hero-diamond}}
\end{align*}
\end{minipage}
}
    \caption{Results on Relative Completeness}
    \label{fig:rel-completeness:relative_completeness_calculus}
\end{figure}

\subsection{A Calculus of Relative Completeness}
\label{subsec:rel-complete:calculus}
To succinctly represent the property preservation and relative completeness results obtained from our theory, we summarize and denote possible properties a dynamic theory (and its axiomatization) may have as follows:
\begin{description}[leftmargin=!, labelwidth=3cm, align=right]
    \item[\proofProp{Fin}{\Delta}] Finite Support (\Cref{def:rel-complete:finite-support})
    \item[\proofProp{FolExp}{\Delta}] FOL Expressiveness (\Cref{def:rel-complete:fol-expressive})
    \item[\proofProp{Göd}{\Delta}] Gödel Expressiveness (\Cref{def:rel-complete:goedel-expressive})
    \item[\proofProp{FolRelCom}{\Delta}] Calculus of Dynamic Logic is relatively complete w.r.t. box and diamond first-order formulas (i.e. $\formulaOneFOL\implies\modBox{\programOne}{\formulaTwoFOL}$ and $\formulaOneFOL\implies\modDia{\programOne}{\formulaTwoFOL}$)%
    \item[\proofProp{RelCom}{\Delta}] Relative Completeness (\Cref{def:rel-complete:rel-complete})
    \item[\proofProp{Co}{\Delta}] Communicating heterogeneous dynamic theory (\Cref{def:hdl:communicating})
\end{description}

Based on this formalization, we have derived the rules presented in \Cref{fig:rel-completeness:relative_completeness_calculus}, describing under which circumstances properties carry over from one dynamic theory to the next during lifting and theory combination.
While we defer details of these results to the Appendix,
we can derive \Cref{lem:rel-completeness:rel-complete-regular,thm:rel-complete:hero-dl} from the results in \Cref{fig:rel-completeness:relative_completeness_calculus}, which concern relative completeness under regular program closure and for theory combination.

\begin{theorem}[Relative Completeness for Regular Programs]
\label{lem:rel-completeness:rel-complete-regular}
Let $\dynamicTheory$ be some FOL Expressive, Gödel expressive, relatively complete dynamic theory with calculus $\calculus[\dynamicTheory]$.
Then $\regular{\dynamicTheory}$ is relatively complete with the calculus $\calculus[\regular{\dynamicTheory}]$ from \Cref{sec:liftedDL:regular}.
\end{theorem}
\begin{proof}
Using the notation from \Cref{fig:rel-completeness:relative_completeness_calculus}:
We have \proofProp{FolExp}{\dynamicTheory}, \proofProp{Göd}{\dynamicTheory},\proofProp{Fin}{\dynamicTheory},\proofProp{RelCom}{\dynamicTheory}.
Then via \Cref{lem:rel-completeness:trivial-regular} get \proofProp{Fin}{\regular{\dynamicTheory}}.
Via \Cref{lem:rel-completeness:fol-expressiveness-regular} get \proofProp{FolExpr}{\regular{\dynamicTheory}} and via \Cref{lem:rel-completeness:regular-box,lem:rel-completeness:regular-diamond} get \proofProp{FolRelCom}{\regular{\dynamicTheory}}.
Then via \Cref{lem:rel-complete} get \proofProp{RelCom}{\regular{\dynamicTheory}}.
\end{proof}

\begin{theorem}[Relative Completeness of Fully Heterogeneous Dynamic Theories] %
\label{end-thm:rel-complete:hero-dl}
Let $\zero{\dynamicTheory},\one{\dynamicTheory}$ be two relatively complete, FOL expressive, Gödel Expressive dynamic theories with finite support
that are communicating in $\preHero{\dynamicTheory}$.
Then the fully heterogeneous dynamic theory $\hero{\dynamicTheory}$ with its proof calculus $\calculus[\hero{\dynamicTheory}]$ is relatively complete.
\end{theorem}
\begin{proof}
\label{proof:rel-complete:hero-dl}
We use the notation and rules from \Cref{fig:rel-completeness:relative_completeness_calculus}:\\
First, let us reconsider our assumptions:
We have \proofProp{FolExp}{\zero{\dynamicTheory}}, \proofProp{Göd}{\zero{\dynamicTheory}},\proofProp{Fin}{\zero{\dynamicTheory}},\proofProp{RelCom}{\zero{\dynamicTheory}}.
We also have \proofProp{FolExp}{\one{\dynamicTheory}}, \proofProp{Göd}{\one{\dynamicTheory}},\proofProp{Fin}{\one{\dynamicTheory}},\proofProp{RelCom}{\one{\dynamicTheory}}.
We further know \proofProp{Co}{\preHero{\dynamicTheory}}.

Using \Cref{lem:rel-completeness:hero-box,lem:rel-completeness:hero-diamond} we get \proofProp{FolRelCom}{\preHero{\dynamicTheory}}.
Moreover have \proofProp{Göd}{\preHero{\dynamicTheory}} (via \Cref{lem:rel-completeness:hero-goedel}) and \proofProp{FolExp}{\preHero{\dynamicTheory}} (via \Cref{lem:rel-completeness:hero-fol-expressive}).
Together this yields \proofProp{RelCom}{\preHero{\dynamicTheory}} via \Cref{lem:rel-complete}.

Via \Cref{lem:rel-completeness:trivial} we moreover get \proofProp{Fin}{\havoced{(\preHero{\dynamicTheory})}} and \proofProp{Göd}{\havoced{(\preHero{\dynamicTheory})}} and we can further derive \proofProp{FolExpr}{\havoced{(\preHero{\dynamicTheory})}} (\Cref{lem:rel-completeness:fol-expressiveness-havoc}) and \proofProp{FolRelCom}{\havoced{(\preHero{\dynamicTheory})}} (\Cref{lem:rel-completeness:box_diamond_complete}).
Again, this yields \proofProp{RelCom}{\havoced{(\preHero{\dynamicTheory})}} via \Cref{lem:rel-complete}.

Based on this, \Cref{lem:rel-completeness:trivial-regular} yields \proofProp{Fin}{\regular{(\havoced{(\preHero{\dynamicTheory})})}} and \proofProp{Göd}{\regular{(\havoced{(\preHero{\dynamicTheory})})}}.
\Cref{lem:rel-completeness:fol-expressiveness-regular} then delivers \proofProp{FolExpr}{\regular{(\havoced{(\preHero{\dynamicTheory})})}} and together with our relative completeness result for \havoced{(\preHero{\dynamicTheory})} (\proofProp{RelCom}{\havoced{(\preHero{\dynamicTheory})}}) this delivers \proofProp{FolRelCom}{\regular{(\havoced{(\preHero{\dynamicTheory})})}} (\Cref{lem:rel-completeness:regular-box,lem:rel-completeness:regular-diamond}).

Finally, this yields $\proofProp{RelCom}{\regular{(\havoced{(\preHero{\dynamicTheory})})}} \equiv \proofProp{RelCom}{\hero{\dynamicTheory}}$ (via \Cref{lem:rel-complete}).
\end{proof}

\section{Proofs and Additional Definitions}
\subsection{Auxilliary Definitions and Lemmas}
\printProofs

\subsection{Definitions and Proofs for Relative Completeness Results}
In order to reason more explicitly about the behavior of renditions, we require a vehicle to reason more explicitly about renditions.
To this end, note that due to the \newNameForInterpolation{} property (see \Cref{def:DL:universe_state_space_var_eval}), a variable $\variableOne$ can be updated to any value from $\valueSet\left(\variableOne\right)$ in any state.
This gives rise to the notion of twin states that allows us to selectively pull values of a twin variable from another state:
\begin{definition}[Twin States]
Let $\dynamicTheory$ be some dynamic theory and $\variableVecOne,\variableVecOneNext$ be twin variable vectors from $\signatureVariables$.
For any two states $\stateOne,\stateTwo\in\stateSpace$ we define the \emph{twin state} $\twinState{\stateOne}{\stateTwo}{\variableVecOne}{\variableVecOneNext}$
as the state that maps variables from $\variableVecOneNext$ to the value of their twin $\variableVecOne$ in $\stateTwo$ and otherwise to $\stateOne$, formally we require that:
\[
\variableEval\left(\twinState{\stateOne}{\stateTwo}{\variableVecOne}{\variableVecOneNext},\variableOne\right) = \begin{cases}
    \variableEval\left(\stateTwo,\variableVecOne_i\right) & \variableOne=\variableVecOneNext_i \text{ for some }i\\
    \variableEval\left(\stateOne,\variableOne\right) & \text{else}
\end{cases}
\]
\end{definition}
Note that due to the \newNameForInterpolation{} property of the state space, we have for any twin variable vectors $\variableVecOne,\variableVecOneNext$ and any states $\stateOne,\stateTwo\in\stateSpace$ that $\twinState{\stateOne}{\stateTwo}{\variableVecOne}{\variableVecOneNext} \in \stateSpace$.
We can then derive the following, equivalent rendition property for a finite support dynamic theory:
\begin{definition}[Equivalent Rendition Definition]
\label{def:rel-complete:rendition-equiv}
For a dynamic theory $\dynamicTheory$ with finite support we say it supports \emph{rendition} iff there exists a function $\programRenditionSym$ such that for all programs $\programOne\in\signaturePrograms$ and all twin variable vectors $\variableVecOne,\variableVecOneNext$ with $\variableVecOne \supseteq \variablesOf{\programOne}$ and $\variableVecOne\cap\variableVecOneNext=\emptyset$ the function $\programRendition{\programOne}\left(\variableVecOne,\variableVecOneNext\right)$ returns a formula in $\folFormulaSet{\dynamicSignature}$ such that:\\
$
\left(\stateOne,\stateTwo\right) \in \programEval\mleft(\programOne\mright)$ iff
$
\twinState{\stateOne}{\stateTwo}{\variableVecOne}{\variableVecOneNext} \models_{\dynamicTheory} \programRendition{\programOne}\left(\variableVecOne,\variableVecOneNext\right)
$ and $
\equalOn{\stateOne}{\stateTwo}{\variableVecOne^\complement}
$
\end{definition}
\begin{proposition}
The definition in \Cref{def:rel-complete:rendition-equiv} is equivalent to the \textsc{Rendition} property as defined in \Cref{def:rel-complete:fol-expressive}.
\end{proposition}
\begin{lemma}[Expressibility of $\dynamicSignature$-formulas as first-order formulas]
\label{lem:expressibility_fol}
If a dynamic theory $\dynamicTheory$ is FOL Expressive with finite support, then for every formula 
$\formulaOne\in\formulaSet{\dynamicSignature}$ there exists a $\formulaOneFOL\in\folFormulaSet{\dynamicSignature}$ such that
$\vDash_{\dynamicTheory} \formulaOne \iff \formulaOneFOL$.
\end{lemma}
\begin{proof}
Proof by structural induction:
\begin{itemize}[label={ABCDEFG},align=left, leftmargin=*]
\item[$\formulaOne\in\folFormulaSet{\dynamicSignature}$]
In this case $\formulaOneFOL \equiv \formulaOne$.

\item[$\formulaOne \equiv \neg \formulaTwo$]
By structural induction we know there is a formula $\formulaTwoFOL$ such that $\vDash_{\dynamicTheory} \formulaTwo \iff \formulaTwoFOL$.
Then, it also holds that
$\vDash_{\dynamicTheory} \formulaOne \iff \neg\formulaTwoFOL$.

\item[$\formulaOne \equiv \modBox{\programOne} \formulaTwo$]
By structural induction we know there is a formula $\formulaTwoFOL$ such that $\vDash_{\dynamicTheory} \formulaTwo \iff \formulaTwoFOL$.
Through our assumption of FOL Expressiveness, we also know that there exists a formula $\programRendition{\programOne}$ such that 
$\left(\stateOne,\stateTwo\right)\in\programEval\mleft(\programOne\mright)$ iff
we have $\twinState{\stateOne}{\stateTwo}{\variableVecOne}{\variableVecOneNext} \models_{\dynamicTheory} \programRendition{\programOne}\left(\variableVecOne,\variableVecOneNext\right)$ and $\equalOn{\stateOne}{\stateTwo}{\variableVecOne^\complement}$.

Hence, choose twin variable vectors $\variableVecOne \supseteq \variablesOf{\programOne}\cup\freeVarsSyn\left(\formulaTwoFOL\right)$ and $\variableVecOneNext$.
We then construct the following formula
$\formulaThree$:
\[
\fa{\variableVecOneNext}{\left(
\programRendition{\programOne}\left(\variableVecOne,\variableVecOneNext\right)
\implies
\left(
\fa{\variableVecOne}{
\left(
\variableVecOne \doteq \variableVecOneNext
\implies
\formulaTwoFOL
\right)
}
\right)
\right)}.
\]
Let $\stateOne\in\stateSpace$ be some state such that $\stateOne\models_{\dynamicTheory} \formulaThree$.
Then consider a state $\stateOneVariant$ such that $\equalOn{\stateOne}{\stateOneVariant}{\left(\variableVecOneNext\right)^\complement}$.
Then we can interpret $\stateOneVariant$ as a twin state for some $\stateTwo$ such that $\stateOneVariant \equiv \twinState{\stateOne}{\stateTwo}{\variableVecOne}{\variableVecOneNext}$ and $\equalOn{\stateOne}{\stateTwo}{\variableVecOne^\complement}$.
Then $\stateOneVariant\models_{\dynamicTheory} \programRendition{\programOne}\left(\variableVecOne,\variableVecOneNext\right)$ iff $\left(\stateOne,\stateTwo\right)\in\programEval\mleft(\programOne\mright)$.
We consider the case where $\programRendition{\programOne}$ is satisfied by $\stateOneVariant$.
Then consider any state $\stateTwoVariant$ with $\equalOn{\stateOneVariant}{\stateTwoVariant}{\variableVecOne^\complement}$ and $\stateTwoVariant \models_{\dynamicTheory} \variableVecOne \doteq \variableVecOneNext$.
Then by assumption we know that $\stateTwoVariant \models_{\dynamicTheory} \formulaTwoFOL$.
However, we also know that $\equalOn{\stateTwo}{\stateTwoVariant}{\variableVecOne}$,
because $\stateOneVariant$ encodes the values of $\stateTwo$ in $\variableVecOneNext$, $\stateTwoVariant$ has the same values as $\stateOneVariant$ for $\variableVecOneNext$ and $\stateTwoVariant$ has equal values for $\variableVecOne$ and $\variableVecOneNext$.
Consequently, by coincidence we get that $\stateTwo \models_{\dynamicTheory}\formulaTwoFOL$.
Hence $\stateOne \models_{\dynamicTheory} \modBox{\programOne}{\formulaTwoFOL}$.

Conversely, assume $\stateOne \models_{\dynamicTheory} \modBox{\programOne}{\formulaTwoFOL}$.
In this case, we know that for any $\left(\stateOne,\stateTwo\right)\in\programEval\mleft(\programOne\mright)$ we get that $\stateTwo \models_{\dynamicTheory} \formulaTwoFOL$.
Importantly, for any $\variableSet\supseteq \variablesOf{\programOne}$ we get $\equalOn{\stateOne}{\stateTwo}{\variableSet}$.
Then for any $\stateOneVariant$ with $\equalOn{\stateOne}{\stateOneVariant}{\variableVecOneNext^\complement}$ with $\stateOneVariant\models_{\dynamicTheory}\programRendition{\programOne}\left(\variableVecOne,\variableVecOneNext\right)$ (otherwise $\stateOneVariant$ trivially satisfies the implication), we know that we can interpret $\stateOneVariant$ as a twin state $\twinState{\stateOne}{\stateTwoVariant}{\variableVecOne}{\variableVecOneNext}$ with $\equalOn{\stateOne}{\stateTwoVariant}{\variableVecOne^\complement}$.
We then know that $\left(\stateOne,\stateTwoVariant\right)\in\programEval\mleft(\programOne\mright)$ and hence $\stateTwoVariant\models_{\dynamicTheory} \formulaTwoFOL$.
However, just as before, we can then construct a state $\stateThree$ such that $\equalOn{\stateOneVariant}{\stateThree}{\variableVecOne^\complement}$ but $\stateThree\models_{\dynamicTheory} \variableVecOne \doteq \variableVecOneNext$.
Hence, $\equalOn{\stateTwoVariant}{\stateThree}{\variableVecOne}$ and hence by coincidence lemma $\stateThree\models_{\dynamicTheory}\formulaTwoFOL$.
Hence, since $\stateThree$ merely represents an arbitrary state after instantiation of quantifiers for $\stateOne$, we get that $\stateOne\models_{\dynamicTheory} \formulaThree$.

\item[$\formulaOne \equiv \forall \variableOne \formulaTwo$]
By structural induction we know there is a formula $\formulaTwoFOL$ such that $\models_{\dynamicTheory} \formulaTwo \iff \formulaTwoFOL$.
Then it also holds that $\models_{\dynamicTheory} \forall \variableOne\formulaTwo \iff \forall \variableOne\formulaTwoFOL$

\item[$\formulaOne \equiv \formulaTwo \land \formulaThree$]
By structural induction we know there are formulas $\formulaTwoFOL$ and $\formulaThreeFOL$ such that $\vDash_{\dynamicTheory} \formulaTwo \iff \formulaTwoFOL$ and resp. $\vDash_{\dynamicTheory} \formulaThree \iff \formulaThreeFOL$.
Then, it also holds that
$\vDash_{\dynamicTheory} \formulaOne \iff \formulaTwoFOL \land \formulaThreeFOL$.
\end{itemize}
\end{proof}

\begin{theorem}[Relative Completeness]
    \label{lem:rel-complete}
    Let $\dynamicTheory$ be an FOL Expressive dynamic theory with finite support. %
    If its proof calculus $\calculus[\dynamicTheory]$ admits proving all valid formulas of the form $\formulaOneFOL \implies \modBox{\programOne}{\formulaTwoFOL}$ and all valid formulas of the form $\formulaOneFOL \implies \modDia{\programOne}{\formulaTwoFOL}$, then its proof calculus is relatively complete w.r.t. $\folFormulaSet{\dynamicTheory}$.
\end{theorem}
\begin{proof}
\label{proof:hdl_relative_completeness}
While admitting more modularity through the properties outlined in \Cref{sec:rel-complete}, the basic gist of the proof is structurally identical to similar proofs in prior literature~\cite{harel_first-order_1979,Platzer2012}.
We assume we are given a valid formula in conjunctive normal form (if not, perform appropriate propositional recombination), where negations are pushed inside modalities and modalities are not pulled out of conjunctions/disjunctions.
Like \cite{harel_first-order_1979}, we perform induction over the number of modalities and the number $N$ of \emph{non-trivial} quantifiers, i.e.\ such quantifiers which have a modality inside.
First, note that the case for $N=0$ modalities and non-trivial quantifiers is covered by our assumption of relative completeness:
If the formula under consideration has no modalities, we can prove its validity relative to $\folFormulaSet{\dynamicTheory}$.
Next, consider the case $N+1$:
We first consider the conjunctions of formulas:
Since the formula is valid, we can prove each conjunct separately.
If a conjunction has $N$ modalities/non-trivial quantifiers, this case is already handled by the inductive hypothesis.
For a conjunct with $N+1$ modalities/non-trivial quantifiers, we must consider a disjunction.
There exist the following cases:
\begin{description}[leftmargin=!, labelwidth=2cm, align=right]
    \item[\ensuremath{\formulaOne\lor\modBox{\programOne}{\formulaTwo}}]
We then know that $\formulaOne,\formulaTwo$ have fewer modalities than our considered formula (covered by case $N$).
Thus, our calculus can already prove that $\formulaOne \iff \formulaOneFOL$ and $\formulaTwo \iff \formulaTwoFOL$ where $\formulaOneFOL,\formulaTwoFOL$ are constructed according to \Cref{lem:expressibility_fol} (note that $\formulaOneFOL,\formulaTwoFOL$ have $N=0$ modalities/non-trivial quantifiers).
We can furthermore prove $\neg\formulaOneFOL \implies \modBox{\programOne}{\formulaTwoFOL}$ by assumption.
Via \ref{axiom:MR} and propositional reasoning, we then derive $\neg\formulaOne \implies \modBox{\programOne} \formulaTwo$.
Propositional recombination then yields $\formulaOne \lor \modBox{\programOne} \formulaTwo$.
    \item[\ensuremath{\formulaOne\lor\modDia{\programOne}{\formulaTwo}}]
    Just as before, $\formulaOne,\formulaTwo$ have fewer modalities and our proof calculus can hence prove $\formulaOne \iff \formulaOneFOL$ and $\formulaTwo \iff \formulaTwoFOL$ which yields $\neg\formulaOne \implies \neg\formulaOneFOL$.
    Using $\formulaTwoFOL\implies\formulaTwo$ with \ref{axiom:G} and \ref{axiom:KDiamond} gives $\modDia{\programOne}{\formulaTwoFOL} \implies \modDia{\programOne}{\formulaTwo}$.
    By assumption, we can then prove $\neg\formulaOneFOL \implies \modDia{\programOne}{\formulaTwoFOL}$.
    Using propositional recombination we can then derive $\neg\formulaOne \implies \modDia{\programOne}{\formulaTwo}$ which then yields $\formulaOne \lor \modDia{\programOne}{\formulaTwo}$.
    \item[\ensuremath{\formulaOne\lor\fa{\variableOne}{\formulaTwo}}]
    \item[\ensuremath{\formulaOne\lor\ex{\variableOne}{\formulaTwo}}]
    Just as before, we know there exist first-order formulas $\formulaOneFOL,\formulaTwoFOL$ equivalent to $\formulaOne,\formulaTwo$.
    We leave $\neg\formulaOneFOL \implies \fa{\variableVecOne}{\formulaTwoFOL}$ (resp. $\neg\formulaOneFOL \implies \ex{\variableVecOne}{\formulaTwoFOL}$) as an exercise for our first-order oracle.
    It remains to prove that $\neg \formulaOne \iff \neg\formulaOneFOL$ and $\formulaTwoFOL\iff\formulaTwo$ (which we know to be true by construction).
    Importantly, these formulas have one \emph{non-trivial} quantifier less and are hence covered by our inductive hypothesis.
    First-order recombination of these three provable formulas then yields $\formulaOne\lor\fa{\variableOne}{\formulaTwo}$ (resp. $\formulaOne\lor\ex{\variableOne}{\formulaTwo}$).
\end{description}
Overall, the induction demonstrates that any valid formula can be derived via the given calculus.
\end{proof}

\begin{isabelleLemma}{Bound Variables over Havoc Programs}{HavocDynamicLogic.BV\_overapprox}
\label{lem:rel-complte:havoc_BV}
Let $\dynamicTheory$ be a dynamic theory with finite support,
then $\havoced{\dynamicTheory}$ also has finite support with the following $\havoced{\boundVarsProgram}$:
\begin{itemize}
    \item $\havoced{\boundVarsProgram}\mleft(\programThree\mright) = \boundVarsProgram\mleft(\programThree\mright)$ (for $\programThree \in \signaturePrograms$)
    \item $\havoced{\boundVarsProgram}\mleft(\variableOne \coloneqq *\mright) = \left\{\variableOne\right\}$
\end{itemize}
\end{isabelleLemma}
\begin{proof}
First note that by construction if $\boundVarsProgram$ is finite (assumed) then so is $\havoced{\boundVarsProgram}$.
Next, we consider the coverage property.
We have proven this property via a structural induction proof in Isabelle.
\end{proof}
\begin{isabelleLemma}{Bound Variables over Regular Programs}{KATDynamicLogic.BV\_overapprox}
\label{lem:rel-complte:regular_BV}
Let $\dynamicTheory$ be a dynamic theory with finite support,
then $\regular{\dynamicTheory}$ also has finite support with the following $\regular{\boundVarsProgram}$:
\begin{itemize}
    \item $\regular{\boundVarsProgram}\mleft(\programThree\mright) = \boundVarsProgram\mleft(\programThree\mright)$ (for $\programThree \in \signaturePrograms$)
    \item $\regular{\boundVarsProgram}\mleft(\programOne;\programTwo\mright) = \regular{\boundVarsProgram}\mleft(\programOne\cup\programTwo\mright) = \regular{\boundVarsProgram}\mleft(\programOne\mright) \cup \regular{\boundVarsProgram}\mleft(\programTwo\mright)$
    \item $\regular{\boundVarsProgram}\mleft(?\left(\formulaOneFOL\right)\mright) = \boundVarsProgram\mleft(\formulaOneFOL\mright)$
    \item $\regular{\boundVarsProgram}\left(\mleft(\programOne\mright)^*\right) = \regular{\boundVarsProgram}\mleft(\programOne\mright)$
\end{itemize}
\end{isabelleLemma}
\begin{proof}
First note that by construction if $\boundVarsProgram$ is finite (assumed) then so is $\regular{\boundVarsProgram}$.
Next, we consider the coverage property.
We have proven this property via a structural induction proof in Isabelle.
\end{proof}
\begin{isabelleLemma}{Bound Variables for Simple Heterogeneous Dynamic Theories}{}
\label{lem:rel-complte:prehero-BV}
Let $\preHero{\dynamicTheory}$ be the simple heterogeneous dynamic theory over $\zero{\dynamicTheory},\one{\dynamicTheory}$. If both homogeneous dynamic theories have finite support, then so has $\preHero{\dynamicTheory}$ with the following $\preHero{\boundVarsProgram}$:
\begin{itemize}
    \item $\preHero{\boundVarsProgram}\left(\programThree\right) = \zero{\boundVarsProgram}\mleft(\programThree\mright)$ (for $\programThree \in \zero{\signaturePrograms}$)
    \item $\preHero{\boundVarsProgram}\left(\programThree\right) = \one{\boundVarsProgram}\mleft(\programThree\mright)$ (for $\programThree \in \one{\signaturePrograms}$)
\end{itemize}
\end{isabelleLemma}
\begin{proof}
First note that by construction if $\zero{\boundVarsProgram}$ and $\one{\boundVarsProgram}$ are finite (assumed) then so is $\preHero{\boundVarsProgram}$.
Next, we consider the coverage property.
We have proven this property via a structural induction proof in Isabelle.
\end{proof}

\subsubsection{Havoc and Relative Completeness}

To illustrate on a simple example how relative completeness results can be extended to lifted logics and heterogeneous dynamic theories, we begin by demonstrating how relative completeness and similar properties of interest carry over when moving from a dynamic theory $\dynamicTheory$ to its lifted version with havoc support $\havoced{\dynamicTheory}$.
To this end, we recall the axioms that we derived earlier for $\havoced{\dynamicTheory}$:
\ref{axiom:havoc} allows us to replace a havoc operation by a first-order quantifier, and \ref{axiom:HR} allows us to reuse proof results for $\dynamicTheory$.
Intuitively, it is hence sensible that this lifting would preserve relative completeness.
Formally, we first observe the following:
\begin{lemma}[Trivial Liftings for Havoc]
\label{lem:rel-completeness:trivial}
For an inductively expressive, Gödel expressive dynamic theory $\dynamicTheory$ with finite support it holds that $\havoced{\dynamicTheory}$ is inductively expressive, Gödel expressive and has finite support
\end{lemma}
\begin{proof}
Finite support is shown in \Cref{lem:rel-complte:havoc_BV}, inductive expressiveness and Gödel expressiveness are first-order properties, i.e. they depend on the state space, the universe, the variables, and the atoms. However, all these components of our dynamic theory are not modified by the lifting. Consequently, $\havoced{\dynamicTheory}$ is inductively expressive by the same argument as $\dynamicTheory$.
\end{proof}
We can next consider the property of FOL Expressiveness.
We observe that FOL Expressiveness also carries over to $\havoced{\dynamicTheory}$:
\begin{lemma}[FOL Expressiveness for Havoc Lifting]
\label{lem:rel-completeness:fol-expressiveness-havoc}
For an FOL expressive dynamic theory $\dynamicTheory$ it holds that $\havoced{\dynamicTheory}$ is also FOL expressive.
\end{lemma}
\begin{proof}
Towards \textsc{\bf Infinite Variables} and \textsc{\bf Eq Predicate} note that $\havoced{\left(\cdot\right)}$ does not modify the state space nor the variable set. Consequently, this property automatically transfers.

Towards \textsc{\bf Rendition} we note that for all programs $\programOne\in\signaturePrograms$ the property is already satisfied.
In addition, we propose the following encoding for the havoc operation:
\[
\programRendition{\variableTwo\coloneqq *}
\left(\variableVecOne,\variableVecOneNext\right)
\equiv
\bigwedge_{
\substack{
    \variableOne\in\variableVecOne\\
    \variableOne \neq \variableTwo
}
} \variableOne\doteq\variableOneNext
\]
Since it holds that $\left(\stateOne,\stateTwo\right)\in\programEval\mleft(\variableTwo\coloneqq*\mright)$ iff
$\equalOn{\stateOne}{\stateTwo}{\left\{\variableTwo\right\}^\complement}$, it naturally follows that
$\left(\stateOne,\stateTwo\right)\in\programEval\mleft(\variableTwo\coloneqq*\mright)$ iff
$\equalOn{\variableOne}{\variableTwo}{\variableVecOne^\complement}$ and 
$\twinState{\stateOne}{\stateTwo}{\variableVecOne}{\variableVecOneNext} \models_{\havoced{\dynamicTheory}} \bigwedge_{
\substack{
    \variableOne\in\variableVecOne\\
    \variableOne \neq \variableTwo
}
} \variableOne\doteq\variableOneNext$.
Consequently, $\havoced{\dynamicTheory}$ is FOL Expressive.
\end{proof}
By extending the calculus $\calculus[\dynamicTheory]$ of $\dynamicTheory$ by the axiom \ref{axiom:havoc} and the rule \ref{axiom:HR}, we can show that the extended calculus $\calculus[\havoced{\dynamicTheory}]$ can derive all valid first-order safety and diamond properties:
\begin{lemma}[Box and Diamond Properties for Havoc Lifting]
\label{lem:rel-completeness:box_diamond_complete}
Consider a dynamic theory $\dynamicTheory$ with relatively complete proof calculus $\calculus[\dynamicTheory]$.
The extended proof calculus $\calculus[\havoced{\dynamicTheory}]$ is relatively complete for formulas of the following two forms where $\formulaOneFOL,\formulaTwoFOL\in\folFormulaSet{\havoced{\dynamicSignature}}$ and $\programOne\in\havoced{\signaturePrograms}$:
\begin{align*}
    \formulaOneFOL & \implies \modBox{\programOne}{\formulaTwoFOL}
    &\formulaOneFOL & \implies \modDia{\programOne}{\formulaTwoFOL}
\end{align*}
\end{lemma}
\begin{proof}
We first consider box properties and perform a case distinction over $\programOne$:
\begin{description}[leftmargin=!, labelwidth=2cm, align=right]
    \item[$\programOne\in\signaturePrograms$]
    In this case, the formula is part of $\formulaSet{\dynamicTheory}$. Since $\calculus[\dynamicTheory]$ is relatively complete, there exists a first-order formula $\models_{\dynamicTheory} \formulaThreeFOL\in\folFormulaSet{\dynamicTheory}$ such that $\formulaThreeFOL \calculus[\dynamicTheory] \formulaOneFOL \implies\modBox{\programOne}{\formulaTwoFOL}$. Using \ref{axiom:HR} we can hence derive $\formulaThreeFOL \calculus[\havoced{\dynamicTheory}] \formulaOneFOL \implies\modBox{\programOne}{\formulaTwoFOL}$ and furthermore $\models_{\dynamicTheory} \formulaThreeFOL$.
    \item[$\programOne\equiv \variableOne \coloneqq *$]
    In this case, we can apply \ref{axiom:havoc} and via propositional reasoning we arrive at the equivalent formulation $\formulaOneFOL \implies \fa{\variableOne}{\formulaTwoFOL}$.
    Since this first-order formula is equivalent and we assumed validity, we have that $\models_{\havoced{\dynamicTheory}} \formulaOneFOL \implies \fa{\variableOne}{\formulaTwoFOL}$ and $\formulaOneFOL \implies \fa{\variableOne}{\formulaTwoFOL} \calculus[\havoced{\dynamicTheory}] \formulaOneFOL \implies \modBox{\programOne}{\formulaTwoFOL}$
\end{description}
Next, we consider diamond properties.
To this end, note that for $p\in\signaturePrograms$ we can make the structurally same argument.
For $v\coloneqq *$ we swap $\forall$ for $\exists$ and apply the same argument as above.
\end{proof}
Using these Lemmas, we can then prove the following result:
\begin{lemma}[Relative Completeness for Havoc Lift]
\label{lem:rel-completeness:rel-complete-havoc}
Let $\dynamicTheory$ be some FOL Expressive relatively complete dynamic theory with calculus $\calculus[\dynamicTheory]$.
Then $\havoced{\dynamicTheory}$ is relatively complete with the calculus $\calculus[\havoced{\dynamicTheory}]$ from \Cref{sec:liftedDL:havoc}.
\end{lemma}
\begin{proof}
This is a corollary of \Cref{lem:rel-complete} using \Cref{lem:rel-completeness:trivial,lem:rel-completeness:box_diamond_complete,lem:rel-completeness:fol-expressiveness-havoc}
\end{proof}
The case of the havoc lifting is quite straight forward, however, we can apply similar principles to the case of regular programs which will be discussed next.

\subsubsection{Regular Programs and Relative Completeness}

\begin{lemma}[Trivial Liftings for Regular Programs]
\label{lem:rel-completeness:trivial-regular}
For an inductively expressive, Gödel expressive dynamic theory $\dynamicTheory$ with finite support it holds that $\regular{\dynamicTheory}$ is inductively expressive, Gödel expressive and has finite support
\end{lemma}
\begin{proof}
Finite support is shown in \Cref{lem:rel-complte:regular_BV}, inductive expressiveness and Gödel expressiveness are first-order properties, i.e. they depend on the state space, the universe, the variables, and the atoms. However, all these components of our dynamic theory are not modified by the lifting. Consequently, $\regular{\dynamicTheory}$ is inductively expressive by the same argument as $\dynamicTheory$.
\end{proof}

\begin{lemma}[FOL Expressiveness for Regular Lifting]
\label{lem:rel-completeness:fol-expressiveness-regular}
For an FOL expressive, Gödel Expressive dynamic theory $\dynamicTheory$ with finite support, it holds that $\regular{\dynamicTheory}$ is also FOL expressive.
\end{lemma}
\begin{proof}
Towards \textsc{\bf Infinite Variables} and \textsc{\bf Eq Predicate} note that $\regular{\left(\cdot\right)}$ does not modify the state space nor the variable set. Consequently, this property automatically transfers.

We will now focus on \textsc{\bf Rendition}:
To prove this property, we must construct a program rendition $\programRendition{\programOne}$ for all programs $\programOne \in \regular{\signaturePrograms}$.
To this end, we first compute $\variableSet= \freeVarsProgram[\regular{\dynamicTheory}]\mleft(\programOne\mright)\cup\regular{\boundVarsProgram}\mleft(\programOne\mright)$ and then proceed as follows where
$\variableVecOne,\variableVecOneNext,\variableVecTwo,\variableVecOneVariant,\variableVecTwoVariant$ are twin variable vectors
and all contain twin variables of all variables in $\variableSet$.
Note that $V$ is finite and $\left|V\right|$ is syntactically fixed.
Hence, quantifiers over variable vectors are to be understood as a sequence of quantifiers instantiating all $\left|\variableSet\right|$ variables.
Since we assume $\dynamicTheory$ to be FOL Expressive, we can skip the case for $\programOne\in\signaturePrograms$ (given by assumption) and focus on the case $\programOne\in\regular{\signaturePrograms}\setminus\signaturePrograms$:
\begin{align*}
\programRendition{\programOne;\programTwo}
\left(\variableVecOne,\variableVecOneNext\right)
&\equiv
\exists \variableVecTwo
\left(
    \programRendition{\programOne}\left(\variableVecOne,\variableVecTwo\right)
    \land
    \programRendition{\programTwo}\left(\variableVecTwo,\variableVecOneNext\right)
\right)
&&\text{for }
\programOne,\programTwo \in \regular{\signaturePrograms}\\
\programRendition{\programOne\cup\programTwo}
\left(\variableVecOne,\variableVecOneNext\right)
&\equiv
\left(
    \programRendition{\programOne}\left(\variableVecOne,\variableVecOneNext\right)
    \lor
    \programRendition{\programTwo}\left(\variableVecOne,\variableVecOneNext\right)
\right)
&&\text{for }
\programOne,\programTwo \in \regular{\signaturePrograms}\\
\programRendition{?\left(\formulaOne\right)}
\left(\variableVecOneVariant,\variableVecOneNext\right)
&\equiv
\bigwedge_{\variableOne\in V} \variableOneVariant\doteq\variableOneNext
\land
\fa{\variableVecOne}{
\left(\variableVecOne\doteq\variableVecOneVariant \implies \formulaOne\right)
}
&&
\text{for }\formulaOne \in \formulaSet{\regular{\signaturePrograms}}\text{ iff } 
\variableVecOneVariant\neq\variableVecOne\supseteq\freeVarsSyn\left(\formulaOne\right)\\
\programRendition{?\left(\formulaOne\right)}
\left(\variableVecOne,\variableVecOneNext\right)
&\equiv
\bigwedge_{\variableOne\in V} \variableOne \doteq\variableOneNext
\land
\formulaOne
&&
\text{for }\formulaOne \in \formulaSet{\regular{\signaturePrograms}}\text{ else}
\end{align*}
For the case of loops, let $k,n\in\integerVariables$ be integer expressive variables of $\dynamicTheory$.
Let $\variableVecOne$ be $M$ dimensional, then
$Z_1,\dots,Z_M$ is a sequence of fresh variables such that $Z_i$ is a twin variable of the variable $\variableTwo$ that admits Gödelization for $\variableVecOne_i$ (remember that $\variableVecOne_i$ may only admit Goedelization w.r.t. variables with a different value set represented by $\variableTwo$).
We can then construct the following first-order formula for $\programRendition{\left(\programOne\right)^*}\left(\variableVecOne,\variableVecOneNext\right)$:
\begin{align*}
\exists n&
\exists Z_1
\dots
\exists Z_{\zero{M}}
\Big(\\
&
\bigwedge_{j=1}^{M}
\goedel{\variableVecOne_j}\left(Z_j, n, 1, \variableVecOne_j\right) \land
\goedel{\variableVecOne_j}\left(Z_j, n, n, \variableVecOneNext_j\right) \land \\
&
\forall k \exists \variableVecTwo \exists \variableVecTwoVariant \\
    &\left(0 < k < n
    \implies 
    \programRendition{\programOne}\left(\variableVecTwo,\variableVecTwoVariant\right)
    \land
    \bigwedge_{j=1}^{M}
    \goedel{\variableVecOne_j}\left(Z_j, n, k, \variableVecTwo_j\right) \land
    \goedel{\variableVecOne_j}\left(Z_j, n, k+1, \variableVecTwoVariant_j\right)
\right)
\\
\Big)
\end{align*}
Note that $0 < k$ and $k < n$ can be encoded via $\natLess$ and $k + 1$ can be encoded as follows:
\[
\ex{\tilde{k}}{\natPlusOne\left(k,\tilde{k}\right) \land \goedel{\variableVecOne_j}\left(Z_j, n, \tilde{k}, \variableVecTwoVariant_j\right)}
\]
Further, for $\tilde{o},\tilde{z}\in\integerVariables$ the literals $0$ and $1$ inside a formula $\formulaOne$ can resp. be encoded as $\tilde{z},\tilde{o}$ via $\ex{\tilde{o}}{\ex{\tilde{z}}{\left(\neg\natPositive\left(\tilde{z}\right) \land \natPlusOne\left(\tilde{z},\tilde{o}\right) \land \formulaOne\right)}}$.

We now need to prove that $\programRendition{\programOne}$ the way defined above indeed satisfies our definition of program renditions.
To this end, we proceed via structural induction over the programs where $\programOne\in\signaturePrograms$ is given by assumption.\\
Remaining Base Case:
\begin{description}[leftmargin=!, labelwidth=1cm, align=right]
\item[$?\left(\formulaOne\right)$]
    By definition we know that $\left(\stateOne,\stateTwo\right) \in \regular{\programEval}\mleft(?\left(\formulaOne\right)\mright)$ iff $\equalOn{\stateOne}{\stateTwo}{\signatureVariables}$
     and $\stateOne \models \formulaOne$.
     This is the case iff we have
    $
    \twinState{\stateOne}{\stateTwo}{\variableVecOneVariant}{\variableVecOneNext}
    \models_{\regular{\dynamicTheory}}
    \bigwedge_{\variableOne\in V} \variableOneVariant\doteq\variableOneNext
    \land
    \fa{\variableVecOne}{
    \left(\variableVecOne\doteq\variableVecOneVariant \implies \formulaOne\right)
    }
    $ and $\equalOn{\stateOne}{\stateTwo}{\variableVecOne^\complement}$.
    Note that we require the $\forall$ around $\formulaOne$ since we have no infrastructure for rewriting $\formulaOne$ with a different set of variables.
    This is due to our minimal set of assumptions about first-order atoms.
    For the case that the first argument of $\programRendition{?\left(\formulaOne\right)}$ already coincides with the variables of $F$, we provide the case without quantification.
\end{description}
Inductive Cases:\\
We assume that for $\programOne,\programTwo \in \regular{\signaturePrograms}$ there exist program renditions $\programRendition{\programOne},\programRendition{\programTwo}$.
\begin{itemize}[label={ABCDE},align=left, leftmargin=*]
\item[$\programOne ; \programTwo$]
By definition of sequential composition for any $\left(\stateOne,\stateThree\right) \in \regular{\programEval}\mleft(\programOne ; \programTwo\mright)$ there exists a state $\stateTwo \in \regular{\stateSpace}$ such that $\left(\stateOne,\stateTwo\right) \in \regular{\programEval}\mleft(\programOne\mright)$ and $\left(\stateTwo,\stateThree\right) \regular{\programEval}\mleft(\programTwo\mright)$.
By our inductive assumption, we then know
$\left(\stateOne,\stateTwo\right) \in \regular{\programEval}\mleft(\programOne\mright)$
iff
we have $\equalOn{\stateOne}{\stateTwo}{\variableVecOne^\complement}$ and
$
\twinState{\stateOne}{\stateTwo}{\variableVecOne}{\variableVecOneNext}
\models_{\regular{\dynamicTheory}}\programRendition{\programOne}\left(\variableOne,\variableOneNext\right)$.
Further, we know that 
$\left(\stateTwo,\stateThree\right) \in \regular{\programEval}\mleft(\programTwo\mright)$
iff
$\equalOn{\stateTwo}{\stateThree}{\variableVecOne^\complement}$ and
$
\twinState{\stateTwo}{\stateThree}{\variableVecOne}{\variableVecOneNext}
\models_{\regular{\dynamicTheory}}\programRendition{\programTwo}\left(\variableOne,\variableOneNext\right)$.
It is then straight forward that $\left(\stateOne,\stateThree\right) \in \regular{\programEval}\mleft(\programOne ; \programTwo\mright)$
iff
$\equalOn{\stateOne}{\stateThree}{\variableVecOne^\complement}$ and 
$\twinState{\stateTwo}{\stateThree}{\variableVecOne}{\variableVecOneNext}
\models_{\regular{\dynamicTheory}}\programRendition{\programOne;\programTwo}\left(\variableOne,\variableOneNext\right)$ (note we get $\equalOn{\stateOne}{\stateThree}{\variableVecOne^\complement}$ from the two inductive assumptions).

\item[$\programOne \cup \programTwo$]
By definition of nondeterministic choice we know that $\left(\stateOne,\stateTwo\right) \in \regular{\programEval}\mleft(\programOne \cup \programTwo\mright)$ iff 
$\left(\stateOne,\stateTwo\right) \in \regular{\programEval}\mleft(\programOne\mright)$ or $\left(\stateOne,\stateTwo\right) \in \regular{\programEval}\mleft(\programTwo\mright)$.
Combining the inductive hypotheses, this yields that 
$\left(\stateOne,\stateTwo\right) \in \regular{\programEval}\mleft(\programOne \cup \programTwo\mright)$ iff 
$\equalOn{\stateOne}{\stateTwo}{\variableVecOne^\complement}$ and furthermore
$
\twinState{\stateTwo}{\stateThree}{\variableVecOne}{\variableVecOneNext}
\models_{\regular{\dynamicTheory}}\programRendition{\programOne\cup\programTwo}\left(\variableOne,\variableOneNext\right)$.
\item[$\left(\programOne\right)^*$]
We know by definition that  $\left(\stateOne,\stateTwo\right) \in \regular{\programEval}\mleft(\left(\programOne\right)^*\mright)$ iff for some $N \in \mathbb{N}$ it holds that $\left(\stateOne,\stateTwo\right) \in \regular{\programEval}\mleft(\left(\programOne\right)^N\mright)$.
Again, by definition this is the case iff there exist states $\stateOneVariant_1,\dots,\stateOneVariant_{N}$ such that $\stateOne=\stateOneVariant_{1}$, $\stateOneVariant_{N}=\stateTwo$ and for all $1 \leq k < N$ it holds that $\left(\stateOneVariant_k,\stateOneVariant_{k+1}\right) \in \regular{\programEval}\mleft(\programOne\mright)$.
We now need to show that this is the case iff 
$\twinState{\stateOne}{\stateTwo}{\variableVecOne}{\variableVecOneNext} \models_{\dynamicTheory} \programRendition{\left(\programOne\right)^*}\left(\variableVecOne,\variableVecOneNext\right)$ and $\equalOn{\stateOne}{\stateTwo}{\variableVecOne^\complement}$.

We show this by initializing the existential variables of the formula.
To this end (via \newNameForInterpolation{}), pick a state $\stateThree$
that is equal to $\stateOne$ everywhere except for:
\begin{itemize}
    \item $\universeToNat\left(\regular{\variableEval}\left(\stateThree,n\right)\right)=N$
    \item For all variables $\variableVecOne_j$ the variable $Z_j$ are assigned to the corresponding value for the sequence
    \[
    \regular{\variableEval}\left(\stateOneVariant_1,\variableVecOne_j\right),\dots,\regular{\variableEval}\left(\stateOneVariant_N,\variableVecOne_j\right)
    \]
    w.r.t. the G\"odel-Predicate $\goedel{\variableVecOne_j}$.
    That is, we assign $Z_j$ a G\"odel-Encoding of the sequence of values that $\variableOne_j$ took on across all $N$ loop iterations.
\end{itemize}
The latter assignments of $Z_j$ ensure that
$\stateThree \models 
\bigwedge_{j=1}^{M}
\goedel{\variableVecOne_j}\left(Z_j, n, 1, \variableVecOne_j\right) \land
\goedel{\variableVecOne_j}\left(Z_j, n, n, \variableVecOneNext_j\right)$ as we previously asserted that $\stateOne=\stateOneVariant_{1}$ and $\stateOneVariant_{\regular{n}}=\stateTwo$.
And the first/last element of $\ith{Z}_j$'s series evaluates to the values of $\variableOne_j$ in $\stateOne$/$\stateTwo$.
We also know that for all $1 \leq K < N$ it holds that $\left(\stateOneVariant_{K},\stateOneVariant_{K+1}\right) \in \regular{\programEval}\mleft(\programOne\mright)$
iff
$\twinState{\stateOneVariant_{K}}{\stateOneVariant_{K+1}}{\variableVecOne}{\variableVecOneNext} \models_{\dynamicTheory}
\programRendition{\programOne}\left(\variableVecOne,\variableVecOneNext\right)$ and $\equalOn{\stateOneVariant_{K}}{\stateOneVariant_{K+1}}{\variableVecOne^\complement}$.
We still need to show that $\stateThree$ satisfies the remaining formula.
To this end, consider any version of $\stateThree$ after the the quantifier $\forall k$ (denoted $\stateThreeVariant$).
To this end, consider the case where $0 < K=\universeToNat\left(\regular{\variableEval}\left(\stateThreeVariant,k\right)\right) < \universeToNat\left(\regular{\variableEval}\left(\stateThreeVariant,n\right)\right)=N$ (the other case trivially satisfies the implication).
For any such $K$, $\stateThreeVariant$ can be updated for $\variableVecTwo$ and $\variableVecTwoVariant$ such that they are (component wise) assigned to the values of $\regular{\variableEval}\left(\stateOneVariant_{K},\variableVecOne\right)$ and
$\regular{\variableEval}\left(\stateOneVariant_{K+1},\variableVecOne\right)$.
Since we know that $\left(\stateOneVariant_{K},\stateOneVariant_{K+1}\right)\in\regular{\programEval}\mleft(\programOne\mright)$ and recalling the condition on $\programRendition{\programOne}$ from above, this yields that:
\begin{align*}
    \stateThreeVariant \models_{\regular{\dynamicTheory}} \exists \variableVecTwo \exists \variableVecTwoVariant \Big(\programRendition{\programOne}\left(\variableVecTwo,\variableVecTwoVariant\right)
    \land \bigwedge_{j=1}^{M}
    \goedel{\variableVecOne_j}\left(Z_j, n, k, \variableVecTwo_j\right) \land
    \goedel{\variableVecOne_j}\left(Z_j, n, k+1, \variableVecTwoVariant_j\right)
\Big)
\end{align*}

In conclusion, $\stateThree$ satisfies all parts of the formula of $\programRendition{\left(\programOne\right)^*}\left(\variableVecOne,\variableVecOneNext\right)$.
Remember that $\stateThreeVariant$ only varies from $\stateThree$ in existentially quantified variables.
Thus, by coincidence lemma we get that
$\twinState{\stateOne}{\stateTwo}{\variableVecOne}{\variableVecOneNext} \models_{\regular{\dynamicTheory}} \programRendition{\left(\programOne\right)^*}$ and via the coincidence lemma and bounded effect, we also know that $\equalOn{\stateOne}{\stateTwo}{\variableVecOne^\complement}$.
Conversely, given a state $\stateThree$ such that $\stateThree = \twinState{\stateOne}{\stateTwo}{\variableVecOne}{\variableVecOneNext} \models_{\regular{\dynamicTheory}} \programRendition{\left(\programOne\right)^*}$ with $\equalOn{\stateOne}{\stateTwo}{\variableVecOne^\complement}$, we can reconstruct a program run of the loop:
The existential instantiation of $n$ tells us the number of loop iterations and the sequences encoded in $Z_j$ (which by construction must all be sequences of length $n$) provide us with values for the intermediate states, where it is ensured that all loop steps correspond to the transition relation of $\programOne$.
Overall, this yields that 
$\left(\stateOne,\stateTwo\right)\in\regular{\programEval}\mleft(\programOne^*\mright)$
iff
$\equalOn{\stateOne}{\stateTwo}{\variableVecOne^\complement}$ and
$\twinState{\stateOne}{\stateTwo}{\variableVecOne}{\variableVecOneNext} \models_{\regular{\dynamicTheory}} \programRendition{\left(\programOne\right)^*}$
\end{itemize}
This concludes our proof which showed that for all programs $\programOne \in \regular{\signaturePrograms}$ it holds that 
$\left(\stateOne,\stateTwo\right) \in \regular{\programEval}\mleft(\programOne\mright)$
iff
$\twinState{\stateOne}{\stateTwo}{\variableVecOne}{\variableVecOneNext} \models_{\regular{\dynamicTheory}} \programRendition{\programOne}$
and $\equalOn{\stateOne}{\stateTwo}{\variableVecOne^\complement}$
\end{proof}

\begin{lemma}[Box Properties for Regular Programs]
\label{lem:rel-completeness:regular-box}
Consider an  FOL expressive, Gödel expressive dynamic theory $\dynamicTheory$ with finite support and a relatively complete proof calculus $\calculus[\dynamicTheory]$.
The extended proof calculus $\calculus[\regular{\dynamicTheory}]$ is relatively complete for formulas of the form $\formulaOneFOL\implies\modBox{\programOne}{\formulaTwoFOL}$ with $\formulaOneFOL,\formulaTwoFOL\in\folFormulaSet{\regular{\dynamicSignature}}$ and $\programOne\in\regular{\signaturePrograms}$
\end{lemma}
\begin{proof}
Note that due to our assumptions, we know that $\regular{\dynamicTheory}$ is FOL expressive.
We perform the proof by structural induction over programs in a similar fashion to Harel \emph{et al.}~\cite{harel_first-order_1979} and Platzer~\cite{Platzer2012} while obtaining the base case from assumptions:
\begin{itemize}[label={ABCDEFG},align=left, leftmargin=*]
\item[$\programOne \in \signaturePrograms$]
In this case, the formula is in $\formulaSet{\dynamicSignature}$.
Since the formula is valid and $\calculus[\dynamicTheory]$ is relatively complete, we can use \ref{axiom:RR} do relatively completely derive the formula in $\regular{\dynamicTheory}$ using $\calculus[\dynamicTheory]$.

\item[$\programOne \equiv ?\left(\formulaThree\right)$]
By applying the axiom \ref{axiom:?} we get a formula which is equivalent and modality free, namely $\formulaOneFOL \implies \left(\formulaThree \implies \formulaTwoFOL\right)$.

\item[$\programOne \equiv \programTwo \cup \programThree$]
By applying the axiom \ref{axiom:cup} we get a formula which is equivalent, namely (after propositional restructuring) $\left(\formulaOneFOL \implies \modBox{\programTwo}{ \formulaTwoFOL}\right) \land \left(\formulaOneFOL \implies \modBox{\programThree}{\formulaTwoFOL}\right)$.
We can then prove these two parts of the conjunction separately.
By inductive assumption, our calculus is capable of proving the two statements, as they must both be valid.

\item[$\programOne \equiv \programTwo;\programThree$]
By applying the axiom \ref{axiom:seq} we get an equivalent formula, namely $\formulaOneFOL \implies \modBox{\programTwo}{\modBox{\programThree}{\formulaTwoFOL}}$.
Using \Cref{lem:expressibility_fol}, we know there is a formula $\formulaThreeFOL \in \formulaSet{\dynamicSignature}$ such that $\models_{\dynamicTheory} \formulaThreeFOL \iff \modBox{\programThree} \formulaTwoFOL$.
We then apply \ref{axiom:MR} w.r.t. $\formulaThreeFOL$, which yields two proof obligations.
First, we have to prove 
$
\formulaOneFOL \implies \modBox{\programOne}{\formulaThreeFOL}
$.
By construction of $\formulaThreeFOL$, this is equivalent to the original formulation.
Therefore, this formula is valid and can be proven according to our inductive assumption.
Secondly, we have to prove 
$
\formulaThreeFOL \implies \modBox{\programOne}{\formulaTwoFOL}.
$.
As the two sides are, by construction, equivalent, this is equally a valid formula which can be proven according to our inductive assumption.

\item[$\programOne \equiv \left(\programTwo\right)^*$]
First, note that $\modBox{\left(\programTwo\right)^*}{\formulaTwoFOL}$ is a provable invariant:
Via \ref{axiom:star} we get 
$\modBox{\left(\programTwo\right)^*} \formulaTwoFOL \implies \modBox{\programTwo}\modBox{\left(\programTwo\right)^*}\formulaTwoFOL$.
Using \ref{axiom:ind} this implies that the formula is an invariant.
Using \Cref{lem:expressibility_fol}, we construct the formula $\formulaThreeFOL$ which is equivalent to $\modBox{\left(\programTwo\right)^*} \formulaTwoFOL$.
Then the formula $\formulaOneFOL \implies \formulaThreeFOL$ is valid and first-order.
Moreover, the formula $\modBox{\left(\programTwo\right)^*}{\formulaTwoFOL} \implies\formulaTwoFOL$ is valid (as $\formulaTwoFOL$ must also hold after 0 iterations) which yields the validity of $\formulaThreeFOL \implies \formulaTwoFOL$ which is again a first-order formula.
Thus, we first adjust the left hand-side via  $\formulaOneFOL \implies \formulaThreeFOL$
and then apply \ref{axiom:MR} with $\formulaThreeFOL$.
The formula $\formulaThreeFOL \implies \modBox{\left(\programTwo\right)^*}\formulaThreeFOL$ is then valid and provable using \ref{axiom:ind} via the induction hypothesis and the right side $\formulaThreeFOL \implies \formulaTwoFOL$ is valid and a first-order provable.
\end{itemize}
\end{proof}

\begin{lemma}[Diamond Properties for Regular Programs]
\label{lem:rel-completeness:regular-diamond}
Consider an FOL expressive, Gödel expressive dynamic theory $\dynamicTheory$ with finite support and a relatively complete proof calculus $\calculus[\dynamicTheory]$.
The extended proof calculus $\calculus[\regular{\dynamicTheory}]$ is relatively complete for formulas of the form $\formulaOneFOL\implies\modDia{\programOne}{\formulaTwoFOL}$ with $\formulaOneFOL,\formulaTwoFOL\in\folFormulaSet{\regular{\dynamicSignature}}$ and $\programOne\in\regular{\signaturePrograms}$.
\end{lemma}
\begin{proof}
Note that due to our assumptions we know that $\regular{\dynamicTheory}$ is FOL expressive.
We perform the proof by structural induction over programs in a similar fashion to Harel \emph{et al.}~\cite{harel_first-order_1979} and Platzer~\cite{Platzer2012} while obtaining the base case from assumptions:
\begin{itemize}[label={ABCDEFG},align=left, leftmargin=*]
\item[\ensuremath{\programOne \equiv \programTwo;\programThree}]
By applying the axiom \ref{axiom:seq} we get an equivalent formula namely $\formulaOneFOL \implies \modDia{\programTwo}{\modDia{\programThree}{\formulaTwoFOL}}$.
Then using \Cref{lem:expressibility_fol} get a formula $\formulaThreeFOL$ such that $\models_{\regular{\dynamicTheory}} \formulaThreeFOL \iff \modDia{\programThree}{\formulaTwoFOL}$.
Due to the assumed validity in combination with the inductive hypothesis we can then use \calculus[\regular{\dynamicTheory}] to (relatively) prove $\formulaOneFOL \implies \modDia{\programTwo}{\formulaThreeFOL}$ and $\formulaThreeFOL \implies \modDia{\programThree}{\formulaTwoFOL}$.
From the latter via \ref{axiom:G} (with $\programTwo$) and \ref{axiom:KDiamond} get
$\modDia{\programTwo}{\formulaThreeFOL} \implies \modDia{\programTwo}{\modDia{\programThree}{\formulaTwoFOL}}$.
Via propositional reasoning this yields $\formulaOneFOL \implies \modDia{\programTwo}{\modDia{\programThree}{\formulaTwoFOL}}$.
\item[...]
For checks and choice the proof follows directly from the dual version of the proof for \Cref{lem:rel-completeness:regular-box}.
For example, for check we apply the dual interpretation of \ref{axiom:?} $\modDia{?\left(\schemaVarOne\right)}\schemaVarTwo \equiv \neg \modBox{?\left(\schemaVarOne\right)} \neg \schemaVarTwo \equiv \neg \left(\schemaVarOne \implies \neg \schemaVarTwo\right) \equiv \schemaVarOne\land\schemaVarTwo$.
Similar arguments follow for the other composition primitives besides loops, which we will discuss now.
\item[\ensuremath{\programOne \equiv \left(\programTwo\right)^*}]%
For loops we instantiate the first-order formula for $\modDia{\left(\programTwo\right)^*} \formulaTwoFOL$ (see \Cref{lem:expressibility_fol}) via the loop construction from \Cref{lem:rel-completeness:fol-expressiveness-regular} while dropping the outer existential quantification of $\exists n$.
Instead, we choose $n,m\in\integerVariables$ in such a way that it is a fresh variable .
We refer to this formula as $\formulaThreeFOL\left(n\right)$.
The formula
\[
\fa{n}{
\natPositive\left(n\right) \implies \modDia{\programTwo}{
\fa{m}{\left(
\natPlusOne\left(m,n\right)
\implies
\fa{n}{\left(
\natEq\left(n,m\right)
\implies
\formulaThreeFOL\left(n\right)
\right)
}\right)
}
}
}
\]
then is valid, because the reachability of a state satisfying $\formulaTwo$ after $n$ states implies there exists a state transition via $\programTwo$ such that after one iteration we can reach $\formulaTwo$ in $n-1$ transitions.
Thus, we can prove the formula above via inductive assumption.
Via \ref{axiom:G} we can then derive the precondition for \ref{axiom:C}
which yields
\[
\fa{n}{
\formulaThreeFOL\left(n\right)
\implies
\modDia{\programTwo^*}{
\ex{n}{
\left(
\neg\natPositive\left(n\right)
\land
\formulaThreeFOL\left(n\right)
\right)
}
}
}
\]
We then know that $\formulaOneFOL \implies \ex{n}{ \formulaThreeFOL\left(n\right)}$ is valid (due to the validity of $\formulaOneFOL\implies\modDia{\mleft(\programOne\mright)^*}\formulaTwoFOL$) which implies that we can derive it via first-order reasoning.
Moreover, it is valid (and thus derivable via first-order reasoning) that
$
\ex{n}{\left(
\neg\natPositive\left(n\right)
\land \formulaThreeFOL\left(n\right)\right)
\implies
\formulaTwoFOL}
$.
Using \ref{axiom:G} and further reasoning via \ref{axiom:KDiamond}, this yields that: 
\[
\modDia{\mleft(\programOne\mright)^*}
\exists
\variableOne \left(
\neg\natPositive\left(n\right)
\land \formulaThreeFOL\left(n\right)\right)
\implies
\modDia{\mleft(\programOne\mright)^*}
\formulaTwoFOL.
\]
Via Modus Ponens, we can therefore derive that $\modDia{\mleft(\programOne\mright)^*}
\formulaTwoFOL$
\end{itemize}
\end{proof}%

\subsubsection{Heterogeneous Dynamic Theories and Relative Completeness}

\begin{lemma}[Gödel Encodings for Simple Heterogeneous Dynamic Theories]
\label{lem:rel-completeness:hero-goedel}
Let $\zero{\dynamicTheory},\one{\dynamicTheory}$ be two Gödel expressive dynamic theories
that are communicating and together form $\preHero{\dynamicTheory}$.
Then $\preHero{\dynamicTheory}$ is also Gödel expressive.
\end{lemma}
\begin{proof}
For any variable in $\zero{\signatureVariables}$ we leave the Gödel formula as is.
For any variable in $\variableOne\in\one{\signatureVariables}$ (encoded in $\variableTwo\in\one{\signatureVariables}$) and $\zero{n},\zero{j}\in\zero{\integerVariables},\one{n},\one{j}\in\one{\integerVariables}$ have:
\[
\goedel{\variableOne}\left(\variableTwo,\zero{n},\zero{j},\variableOne\right) \equiv 
\left(\ex{\one{n}}{\ex{\one{j}}{\left(
\preHero{\natEq}\left(\zero{n},\one{n}\right) \land
\preHero{\natEq}\left(\zero{j},\one{j}\right) \land
\goedel{\variableOne}\left(\variableTwo,\one{n},\one{j},\variableOne\right)
\right)}}\right)
\]
This translates the integer indices $\zero{n},\zero{j}$ into integer expressive variables native to $\one{\dynamicTheory}$.
\end{proof}

\begin{lemma}[FOL Expressiveness]
\label{lem:rel-completeness:hero-fol-expressive}
Let $\zero{\dynamicTheory},\one{\dynamicTheory}$ be two FOL expressive dynamic theories
that are communicating in $\preHero{\dynamicTheory}$.
Then $\preHero{\dynamicTheory}$ is also FOL expressive.  
\end{lemma}
\begin{proof}
For \textsc{\bf Infinite Variables} note that the value ranges of variables for any $\variableOne\in\ith{\signatureVariables}$ has not changed.
Hence, since $\zero{\dynamicTheory},\one{\dynamicTheory}$ have the infinite variable property, so has $\preHero{\dynamicTheory}$.
For \textsc{\bf Eq Predicate} this is enforced via the assumption that $\zero{\dynamicTheory},\one{\dynamicTheory}$ are communicating in $\preHero{\dynamicTheory}$ (and thus via new atoms in $\common{\signatureAtoms}$).

The \textsc{\bf Rendition} is constructed as follows:
\begin{align*}
\programRendition{\programOne}
\left(\variableVecOne,\variableVecOneNext\right)
&\equiv
\ith{\programRendition{\programOne}}
\left(\variableVecOne\mid_{\ith{\signatureVariables}},\variableVecOneNext\mid_{\ith{\signatureVariables}}\right)
\land
\bigwedge_{\variableOne \in \ith[1-i]{\signatureVariables} \cap \variableVecOne
} \variableOne\doteq\variableOneNext
&&
\text{for }
\programOne \in \ith{\signaturePrograms}\
\end{align*}
Where $\variableVecOne\mid_{\ith{\signatureVariables}}=\left(\ith{\variableOne}_1,\dots,\ith{\variableOne}_{\ith{m}}\right)$ with $\ith{\variableOne}_1,\dots,\ith{\variableOne}_{\ith{m}}$ corresponding to the (syntactically created) subvector of $\variableVecOne$ whose variables are in $\ith{\signatureVariables}$ and $\variableVecOneNext\mid_{\ith{\signatureVariables}}$ contains the corresponding twin variables.
With our assumption about the correctness of $\programRendition{\programOne}$ in the FOL Expressive dynamic theory $\ith{\dynamicTheory}$, this immediately implies that $\preHero{\dynamicTheory}$ equally has the rendition property.
Note that we may not assume finite support in this proof, as we solely rely on the correctness of the FOL expressivity of the underlying logics.
\end{proof}

\begin{lemma}[Box Properties for Simple Heterogeneous Theories]
\label{lem:rel-completeness:hero-box}
Let $\zero{\dynamicTheory},\one{\dynamicTheory}$ be two relatively complete, FOL expressive dynamic theories with finite support
that are communicating in $\preHero{\dynamicTheory}$.
Consider now the simple heterogeneous dynamic theory (i.e. $\left(\preHero{\dynamicTheory}\right)$), then its proof calculus $\calculus[\preHero{\dynamicTheory}]$ is relatively complete for formulas of the form $\formulaOneFOL\implies\modBox{\programOne}{\formulaTwoFOL}$ with $\formulaOneFOL,\formulaTwoFOL\in\folFormulaSet{\preHero{\dynamicTheory}},\programOne\in\folFormulaSet{\preHero{\signaturePrograms}}$.
\end{lemma}
\begin{proof}
In contrast to similar proofs (e.g. Harel \emph{et al.}~\cite{harel_first-order_1979} and Platzer~\cite{Platzer2012}), the major challenge in deriving relative completeness is the case where $\programOne \in \ith{\signaturePrograms}$ but $\formulaOneFOL,\formulaTwoFOL \in \formulaSet{\preHero{\dynamicSignature}}$,
i.e. the formulas may contain predicates, functions and variables from both theories.
We tackle this case, which is in fact the only one for programs $\programOne\in\preHero{\signaturePrograms}$, as follows:
\begin{description}[leftmargin=!, labelwidth=1.5cm, align=right]
\item[$\programOne \in \ith{\signaturePrograms}$]
In this case we must discharge the reasoning about $\programOne$ to the homogeneous calculus $\calculus[\ith{\dynamicTheory}]$.
However, $\formulaOneFOL$ and $\formulaTwoFOL$ are first-order formulas over $\preHero{\dynamicSignature}$ and can therefore contain variables or atoms not native to $\ith{\dynamicSignature}$.
Thus, we cannot simply hand the formula to $\calculus[\ith{\dynamicTheory}]$.
Instead, for each $\variableOne \in \variablesOf{\programOne}$ we apply the rule \ref{axiom:ghost} (keeping in mind \Cref{lem:rel-completeness:id_eq}) with $\variableOnePrev=\variableOne$ for a fresh variable $\variableOnePrev$.
Decomposition of the introduced modalities using
\ref{axiom:havoc}, \ref{axiom:seq} and \ref{axiom:?}
coupled with further simplification using first-order reasoning yields
$
    \left(\formulaOneFOL \land 
    \variableVecOnePrev \doteq \variableVecOne
    \right)
    \implies
    \modBox{\programOne}
    \formulaTwoFOL
$

We can then derive the validity of the following formula (note that all $\variableOnePrev$ are fresh):
\[
\left(
    \left(\formulaOneFOL \land 
    \variableVecOnePrev \doteq \variableVecOne
    \right)
    \iff
    \left(
    \fa{\variableVecOne}{\left(
    \variableVecOne\doteq\variableVecOnePrev \implies \formulaOneFOL
    \right)} \land 
    \variableVecOnePrev \doteq \variableVecOne
    \right)
\right)
\]
Using this fact, we can obtain an equivalent formula where we substitute the left hand-side by the right hand-side using first-order reasoning.
This yields 
\[
    \left(
    \fa{\variableVecOne}{\left(
    \variableVecOne\doteq\variableVecOnePrev \implies \formulaOneFOL
    \right)} \land 
    \variableVecOnePrev \doteq \variableVecOne
    \right)
    \implies
    \modBox{\programOne}
    \formulaTwoFOL
\]
We then apply \ref{axiom:MR} w.r.t. $\fa{\variableVecOne}{\left(
    \variableVecOne\doteq\variableVecOnePrev \implies \formulaOneFOL
    \right)} \land \programRendition{\programOne}\left(\variableVecOnePrev,\variableVecOne\right)$ which yields two proof obligations.
The first proof obligation reads:
\[
\left(
    \fa{\variableVecOne}{\left(
    \variableVecOne\doteq\variableVecOnePrev \implies \formulaOneFOL
    \right)} \land 
    \variableVecOnePrev \doteq \variableVecOne
    \right)
    \implies
    \modBox{\programOne}
    \left(\fa{\variableVecOne}{\left(
    \variableVecOne\doteq\variableVecOnePrev \implies \formulaOneFOL
    \right)} \land \programRendition{\programOne}\left(\variableVecOnePrev,\variableVecOne\right)\right)  
\]
Using \ref{axiom:boxAnd} we decompose this formula into a proof that $\fa{\variableVecOne}{\left(
    \variableVecOne\doteq\variableVecOnePrev \implies \formulaOneFOL
    \right)}$ is preserved (this immediately follows from \ref{axiom:V}) and a proof that $\programRendition{\programOne}\left(\variableVecOnePrev,\variableVecOne\right)$ holds after execution of $\programOne$.
To this end, we weaken the right side yielding the formula
$\variableVecOnePrev\doteq\variableVecOne
\implies
\modBox{\programOne}
\programRendition{\programOne}\left(\variableVecOnePrev,\variableVecOne\right)
$.
Note, that this formula is in $\ith{\dynamicSignature}$ and is furthermore valid, as $\programRendition{\programOne}$ precisely describes the variable values of all reachable states.
Thus, we can discharge this proof obligation to $\calculus[\ith{\dynamicTheory}]$ which (by assumption) must relatively prove its validity w.r.t. first-order formulas.

Secondly, we must prove that:
\[
\fa{\variableVecOne}{\left(
    \variableVecOne\doteq\variableVecOnePrev \implies \formulaOneFOL
    \right)}  \land \programRendition{\programOne}\left(\variableVecOnePrev,\variableVecOne\right)
\implies
\formulaTwoFOL
\]
First, note that this formula is a first-order formula over $\preHero{\dynamicSignature}$.
Consider any state $\stateTwo$ such that $\stateTwo \models_{\dynamicTheory} \fa{\variableVecOne}{\left(
    \variableVecOne\doteq\variableVecOnePrev \implies \formulaOneFOL
    \right)}  \land \programRendition{\programOne}\left(\variableVecOnePrev,\variableVecOne\right)$ (otherwise the formula is trivially satisfied).
Then, we construct $\stateOne \in \stateSpace$ where
$\preHero{\variableEval}\left(\stateOne,\variableOne\right)=\preHero{\variableEval}\left(\stateTwo,\variableOnePrev\right)$ for all $\variableOne \in \variablesOf{\programOne}$ and
$\equalOn{\stateOne}{\stateTwo}{\variablesOf{\programOne}^\complement}$.
Then, by definition of $\programRendition{\programOne}$ it holds that $\left(\stateOne,\stateTwo\right) \in \preHero{\programEval}\mleft(\programOne\mright)$ and $\stateOne \models_{\dynamicTheory} \formulaOneFOL$.
Since we assumed the validity of the original formula, this implies that $\stateTwo \vDash \formulaTwoFOL$.
Thus, the second proof obligation is valid and we therefore reduced the validity of the original formula to a first-order formulas.
\end{description}
\end{proof}
\begin{lemma}[Diamond Properties for Simple Heterogeneous Theories with Havoc]
\label{lem:rel-completeness:hero-diamond}
Let $\zero{\dynamicTheory},\one{\dynamicTheory}$ be two relatively complete, FOL expressive dynamic theories with finite support
that are communicating in $\preHero{\dynamicTheory}$.
Consider now the simple heterogeneous dynamic theory (i.e. $\left(\preHero{\dynamicTheory}\right)$), then its proof calculus $\calculus[\preHero{\dynamicTheory}]$ is relatively complete for formulas of the form $\formulaOneFOL\implies\modDia{\programOne}{\formulaTwoFOL}$ with $\formulaOneFOL,\formulaTwoFOL\in\folFormulaSet{\preHero{\dynamicTheory}},\programOne\in\folFormulaSet{\preHero{\signaturePrograms}}$.
\end{lemma}
\begin{proof}
In contrast to similar proofs (e.g. Harel \emph{et al.}~\cite{harel_first-order_1979} and Platzer~\cite{Platzer2012}), the major challenge in deriving relative completeness is the case where $\programOne \in \ith{\signaturePrograms}$ but $\formulaOneFOL,\formulaTwoFOL \in \formulaSet{\preHero{\dynamicSignature}}$,
i.e. the formulas may contain predicates, functions and variables from both theories.
We tackle this case, which is in fact the only one for programs $\programOne\in\preHero{\signaturePrograms}$, as follows:
\begin{description}[leftmargin=!, labelwidth=1.5cm, align=right]
\item[$\programOne \in \ith{\signaturePrograms}$]%
We apply the \ref{axiom:PB} rule with the two fresh formulas $\formulaThree_1 \equiv \programRendition{\programOne}\left(\variableVecOne,\variableVecOneNext\right)$ and $\formulaThree_2 \equiv {\variableVecOne \doteq \variableVecOneNext}$.
This results in the formula
\[
\formulaOneFOL \implies
\left(\forall \variableVecOneNext
\left(\formulaThree_1 \implies \modDia{\programOne} \formulaThree_2\right)
\land
\exists \variableVecOneNext
\left(
\formulaThree_1 \land
\forall \variableVecOne
\left(\formulaThree_2 \implies \formulaTwoFOL\right)
\right)\right).
\]
We can split this up into two separate proofs.
The first proof obligation (after omitting $\formulaOneFOL$) reads
$\forall \variableVecOneNext
\left(\formulaThree_1 \implies \modDia{\programOne} \formulaThree_2\right)$.
Notably, this formula is entirely in $\ith{\dynamicSignature}$ and can therefore be discharged via the calculus $\calculus[\ith{\dynamicTheory}]$.
Since we assumed the relative completeness of $\calculus[\ith{\dynamicTheory}]$ this must succeed as $\formulaThree_1$ is our program rendition.
The second proof obligation reads 
$
\formulaOneFOL \implies
\exists \variableVecOneNext
\left(
\formulaThree_1 \land
\forall \variableVecOne
\left(\formulaThree_2 \implies \formulaTwoFOL\right)
\right)
$.
We assumed that $\formulaOneFOL \implies \modDia{\programOne} \formulaTwoFOL$ is valid and $\formulaThree_1\equiv\programRendition{\programOne}\left(\variableVecOne,\variableVecOneNext\right)$ exactly describes the possible state transitions starting from $\variableVecOne$ (in this instance in particular for a state satisfying $\formulaOneFOL$).
Thus, there must exists a $\variableVecOneNext$ such that $\programRendition{\programOne}\left(\variableVecOne,\variableVecOneNext\right)$
and that further for all $\variableVecOne$ with $\variableVecOne\doteq\variableVecOneNext$ it holds that $\formulaTwoFOL$ is satisfied.
This is exactly what the formula checks and the formula is thus valid.
Since it is moreover by definition a first-order formula, we can prove it using our oracle.
Thus, we can prove the validity of diamond formulas for $\programOne \in \ith{\signaturePrograms}$.
\end{description}
\end{proof}
\begin{lemma}[Relative Completeness: Simple Heterogeneous Dynamic Theories]
\label{lem:rel-completeness:rel-complete-simple-hero}
Let $\zero{\dynamicTheory},\one{\dynamicTheory}$ be two relatively complete, FOL expressive dynamic theories with finite support
that are communicating in $\preHero{\dynamicTheory}$.
Then the simple heterogeneous dynamic theory $\preHero{\dynamicTheory}$ with its proof calculus $\calculus[\preHero{\dynamicTheory}]$ is relatively complete.
\end{lemma}
\begin{proof}
This is a corollary from \Cref{lem:rel-complete}:
We know that $\preHero{\dynamicTheory}$ is FOL expressive (\Cref{lem:rel-completeness:hero-fol-expressive}) and that the diamond and box fragments are relatively complete (\Cref{lem:rel-completeness:hero-box,lem:rel-completeness:hero-diamond}).
Consequently $\calculus[\preHero{\dynamicTheory}]$ for the full theory.
\end{proof}

\subsection{Proof of Case Study Guarantee}
\label{apx:proof_example}
Here we prove the example described in \Cref{case_study_guarantee}.
We begin by giving a full account of the (in the main text abbreviated) formula $\mathrm{coupledPre}$:
\begin{align*}
    \mathrm{coupledPre} ~\equiv~&
    \texttt{ctrlPre} ~\land~\\
& \textit{A}>0 ~\land~ \textit{B}>0 ~\land~ \textit{T}>0 ~\land~ x + v^2/(2*B) \leq s ~\land~\\
&
    \intToReal\left(\texttt{A}^-, 100\textit{A}\right)  ~\land~
     \intToReal\left(\texttt{B}^-, -100\textit{B}\right)~\land~
     \intToReal\left(\texttt{T}^-, 100*\textit{T}\right)\\
&
    \mathrm{round}\left(\texttt{p}, 100\left(\textit{x}-\textit{s}\right)\right) \land
    \mathrm{round}\left(\texttt{v},100\textit{v}\right)
\end{align*}
We then apply \ref{axiom:MR}, introducing $\schemaVarOne \equiv \mathrm{coupledPre}$ which leaves us with the cases \textbf{(A)} and \textbf{(B)}.

\noindent
\textbf{(A)}
Note that $\mathrm{coupledPre}$ is a large conjunction which in particular contains $x + v^2/(2B) \leq s$
We can thus show that $\mathrm{coupledPre} \implies x \leq s$.

\noindent
\textbf{(B)}
We now need to prove that $\mathrm{coupledPre} \implies \modBox{\left(\texttt{ctrl};\hero{\alpha}\right)^*} \mathrm{coupledPre}$.
This allows us to apply the \ref{axiom:ind} rule leaving the proof obligation 
\[
\mathrm{coupledPre} \implies \modBox{\texttt{ctrl};\hero{\alpha}} \mathrm{coupledPre}.
\]
As a reminder we once again print $\hero{\alpha}$:
\[
\left(
\textit{a} \coloneqq *;~
?\left(%
\intToReal\left(\texttt{this.acc},a*100\right)\right);~
\textit{env};~
\texttt{p}\coloneqq *;~
\texttt{v} \coloneqq *;~
?\left(
\mathrm{coupling}
\right)
\right)
\]
We decompose $\texttt{ctrl};\hero{\alpha}$ via \ref{axiom:seq}.
Next we again apply \ref{axiom:MR} w.r.t $\schemaVarOne\equiv\formulaOne_1$. (see below) which produces \textbf{(B.A)} and \textbf{(B.B)}.
\begin{align*}
   \formulaOne_1 \equiv &\texttt{ctrlPost} \land\\
   & \left(\textit{A}>0 ~\land~ \textit{B}>0 ~\land~ \textit{T}>0 ~\land~ x + v^2/(2*B) \leq s\right) ~\land~\\
    &\big(
    \intToReal\left(\texttt{A}^-, 100\textit{A}\right)  ~\land~
     \intToReal\left(\texttt{B}^-, -100\textit{B}\right)~\land~
     \intToReal\left(\texttt{T}^-, 100*\textit{T}\right)~\land~\\
&     \mathrm{round}\left(\texttt{p}, 100\left(\textit{x}-\textit{s}\right)\right) \land
    \mathrm{round}\left(\texttt{v},100\textit{v}\right)
    \big)
\end{align*}

\noindent
\textbf{(B.A)}
We now need to prove that $\mathrm{coupledPre} \implies \modBox{\texttt{ctrl}} \formulaOne_1$.
We split this into two separate proof obligations via application of \ref{axiom:boxAnd}:
The preservation of the $\intToReal$/$\mathrm{round}$ atoms (last part) can be shown via \ref{axiom:V} as these variables are not part of $\boundVarsSem\mleft(\texttt{ctrl}\mright)$ of JavaDL and follow from $\mathrm{coupledPre}$.
The preservation of $\left(\textit{A}>0 ~\land~ \dots ~\land~ x + v^2/(2*B) \leq s\right)$ can be skipped via \ref{axiom:Fi} (for $i=\dL{}$) as this is already part of $\mathrm{coupledPre}$ and soley concerns \dL{} variables.
Finally, it remains to show that $\mathrm{coupledPre}\implies\modBox{\texttt{ctrl}}\texttt{ctrlPost}$.
Via \ref{axiom:HR0} we reduce this to \texttt{JavaDL}.
In fact, we have already proven this property in \Cref{lem:java_ctrl_correctness}.

\noindent
\textbf{(B.B)}
This proof branch now reads:
\[
\formulaOne_1
\implies
\modBox{\textit{a} \coloneqq *;~
?\left(
\intToReal\left(\texttt{this.acc}, \textit{a}*100\right)
\right);~
\textit{env};~
\texttt{p}\coloneqq *;~
\texttt{v} \coloneqq *;~
?\left(
\mathrm{coupling}
\right)}
\mathrm{coupledPre}
\]
We decompose further via
\ref{axiom:seq}
\ref{axiom:havoc}
\ref{axiom:seq}
\ref{axiom:?}.
After further simplification, the formula then reads:
\[
\left(
\begin{array}{l}
    \texttt{ctrlPost} \land \\
    \intToReal\left(\texttt{A}^-,100\textit{A}\right)~\land~\\
     \intToReal\left(\texttt{B}^-,-100\textit{B}\right)~\land~\\
     \intToReal\left(\texttt{T}^-,100\textit{T}\right) ~\land~\\
     \mathrm{round}\left(\texttt{p}, 100\left(\textit{x}-\textit{s}\right)\right) ~\land~\\
    \mathrm{round}\left(\texttt{v}, 100\textit{v}\right) \land\\
    \intToReal\left(\texttt{this.acc}, \textit{a}*100\right)~\land~\\
    \textit{A} > 0 ~\land~ -\textit{B} > 0 ~\land~\\
    T>0 ~\land~ x + v^2/(2B) \leq s
\end{array}
\right)
\implies
\modBox{\textit{env};~
\texttt{p}\coloneqq *;~
\texttt{v} \coloneqq *;~
?\left(
\mathrm{coupling}
\right)}
\mathrm{coupledPre}
\]
Note, that in comparison to \textit{envPre} we lack the formula \textit{acc\_assumptions}.
To remedy this, we perform a cut w.r.t. \textit{acc\_assumptions}.
We can prove this formula because it follows from the equation on the left hand-side (in particular via \texttt{ctrlPost}).
We perform this proof via mechanized first-order reasoning in \ac{keymaerax}.
After simplification and application of \ref{axiom:seq} this yields the following proof situation:
\[
\left(
\begin{array}{l}
    \texttt{ctrlPost} \land \textit{envPre} \\
    \intToReal\left(\texttt{A}^-,100\textit{A}\right)~\land~\\
     \intToReal\left(\texttt{B}^-,-100\textit{B}\right)~\land~\\
     \intToReal\left(\texttt{T}^-,100\textit{T}\right)
\end{array}
\right)
\implies
\modBox{\textit{env}}
\modBox{
\texttt{p}\coloneqq *;~
\texttt{v} \coloneqq *;~
?\left(
\mathrm{coupling}
\right)}
\mathrm{coupledPre}
\]
Once again we apply \ref{axiom:MR} this time w.r.t. $\schemaVarOne\equiv\formulaOne_2$ as follows:
\begin{align*}
\formulaOne_2 \equiv &
\textit{envPost} \land \texttt{ctrlPost}~\land\\
&\left(\intToReal\left(\texttt{A}^-,100\textit{A}\right)~\land~
     \intToReal\left(\texttt{B}^-,-100\textit{B}\right)~\land~
     \intToReal\left(\texttt{T}^-,100\textit{T}\right)\right)~\land\\
 & \left(\textit{A}>0 ~\land~ \textit{B}>0 ~\land~ \textit{T}>0\right)
\end{align*}
This leaves proof obligations \textbf{(B.B.A)} and \textbf{(B.B.B)}.

\noindent
\textbf{(B.B.A)}
We now need to prove that the following formula is valid:
\[
\left(
\begin{array}{l}
    \texttt{ctrlPost} \land \textit{envPre} \\
    \intToReal\left(\texttt{A}^-,100\textit{A}\right)~\land~\\
     \intToReal\left(\texttt{B}^-,-100\textit{B}\right)~\land~\\
     \intToReal\left(\texttt{T}^-,100\textit{T}\right)
\end{array}
\right)
\implies
\modBox{\textit{env}}
\formulaOne_2
\]
To this end, we apply \ref{axiom:boxAnd} to factor out the equalities and inequalities of the second and third line of $\formulaOne_2$, which can be shown via \ref{axiom:V} as the contained variables are not part of $\boundVarsSem\mleft(\textit{env}\mright)$.
For the remaining formula, we apply \ref{axiom:Fi} (for i=\texttt{JavaDL}), which shows the preservation of \texttt{ctrlPost}, meaning it only remains to show \textit{envPost}.
After weakening the precondition, we thus get $\textit{envPre} \implies \modBox{\textit{env}} \textit{envPost}$.
Fortunately, we have previously shown this with \ac{keymaerax} in \Cref{lem:dl_env_correctness}.
Hence, we can reuse this result via \ref{axiom:HR1}

\noindent
\textbf{(B.B.B)}
It remains to show that the following formula is valid:
\[
\formulaOne_2
\implies
\modBox{
\texttt{p}\coloneqq *;~
\texttt{v} \coloneqq *;~
?\left(
\mathrm{coupling}
\right)
}
\mathrm{coupledPre}
\]
After further reduction via 
\ref{axiom:seq}
\ref{axiom:havoc}
\ref{axiom:seq}
\ref{axiom:havoc}
\ref{axiom:seq}
\ref{axiom:?}
and further simplification this yields:
\begin{align*}
\left(
\begin{array}{l}
\texttt{heap\_assumptions}~\land~\\
\texttt{A}^-\doteq\texttt{A} \land \texttt{B}^-\doteq\texttt{B}~\land~\\
\texttt{T}^-\doteq\texttt{T} ~\land~\textit{envPost}~\land~\\
\intToReal\left(\texttt{A}^-,100\textit{A}\right)~\land~\\
 \intToReal\left(\texttt{B}^-,-100\textit{B}\right)~\land~\\
 \intToReal\left(\texttt{T}^-,100\textit{T}\right)\\
\mathrm{round}\left(\texttt{p},100\left(\textit{x}-\textit{s}\right)\right) ~\land~\\
\mathrm{round}\left(\texttt{v},100\textit{v}\right)~\land\\
\textit{A}>0 ~\land~ \textit{B}>0 ~\land~ T>0
\end{array}
\right)
\implies
\left(
\begin{array}{l}
    \texttt{ctrlPre} ~\land~\\
    \textit{A}>0 ~\land~ \textit{B}>0 ~\land~ \textit{T}>0 ~\land~\\
    x + v^2/(2*B) \leq s~\land\\
    \intToReal\left(\texttt{A}^-,100\textit{A}\right)~\land~\\
     \intToReal\left(\texttt{B}^-,-100\textit{B}\right)~\land~\\
     \intToReal\left(\texttt{T}^-,100\textit{T}\right)\\
    \mathrm{round}\left(\texttt{p},100\left(\textit{x}-\textit{s}\right)\right) ~\land~\\
\mathrm{round}\left(\texttt{v},100\textit{v}\right)~\land\\
\end{array}
\right)
\end{align*}
Almost all elements of the conjunction on the right (after expansion of \texttt{ctrlPre} and \textit{envPost}) can be shown via $\phi \implies \phi$.
The exception are the formulas $\texttt{A}>0$, $\texttt{B}<0$ and $\texttt{T}>0$ which nonetheless easily follow from $\textit{A}>0,\textit{B}>0,\textit{T}>0$ and the $\intToReal$ constraints.
This concludes our proof that $\mathrm{coupledPre}$ is inductive w.r.t. $\hero{\alpha}$ and simulaneously concludes our proof on the validity of \Cref{eq:valid_formula_example}.

\clearpage
\section{Isabelle Formalization: Overview}
\label{apx:isabelle}
To obtain the same modularity of components presented in the paper in Isabelle, the formalization is spread across several locales, which impose additional assumptions on the considered dynamic theory in exchange for providing additional proof rules.

\subsection{Basic Constructions}

\paragraph{FirstOrder.}
The underlying first-order structure is implemented in the locale \texttt{folCore} in the file \texttt{FirstOrder.thy}.
This locale demands the assumptions from \Cref{def:DL:universe_state_space_var_eval,def:DL:atom_eval} and constructs first-order formulas over the provided atoms (given as type parameter \texttt{'atom}).

\paragraph{Dynamic Logic\del{S14}{ Core}.}
Given an instantiation of the locale \texttt{folCore} (see above) and programs satisfying the assumptions from \Cref{def:DL:prog_eval}, we can instantiate the locale \chg{S14}{\texttt{dynamic\_theory}}{\texttt{dynLogCore}} in \texttt{DynamicLogicCore.thy}, which implements the core results of a dynamic theory. Programs are provided as the type parameter \texttt{'prog}.
This includes the data structures for formulas, their evaluation, all coincidence and bound effect lemmas, Hilbert calculus axioms (sanity check), and the axioms \ref{axiom:G} and \ref{axiom:K}.\del{S14}{
Given an instantiation of the locale \texttt{dynLogCore} we can instantiate (without further assumptions) the locale \texttt{dynLog} in}
\chg{S14}{The file \texttt{DynamicLogic.thy} then proves}{
the file \texttt{DynamicLogic.thy}, which implements} the additional axioms \ref{axiom:V} and \ref{axiom:B}.
When this paper states a particular structure is a \emph{dynamic theory} (see \Cref{def:DL:theory}) this corresponds to it being an instantiation of this locale (\chg{S14}{\texttt{dynamic\_theory}}{\texttt{dynLog}}).

\subsection{Liftings}
Liftings are implemented as locales whose only assumption is the instantiation of a\chg{S14}{n existing locale instantiation}{ locale} \chg{S14}{\texttt{dynamic\_theory}}{\texttt{dynLog}}. We have implemented Havoc lifting (see \Cref{sec:liftedDL:havoc}) and a lifting to generate the closure over regular programs (see \Cref{sec:liftedDL:regular}).

\paragraph{Havoc.}
The havoc lifting is implemented in the locale \texttt{\chg{S14}{havoc\_lifting}{HavocDynLog}} in \texttt{HavocDynamicLogic.thy}.
This locale takes a given instantiation of \chg{S14}{\texttt{dynamic\_theory}}{\texttt{dynLog}}.
It then constructs an extended programming language which includes the havoc operation, defines appropriate primitives for program evaluation, and free variable computation.
It then proves that the extended constructs are once again an instantiation of \chg{S14}{\texttt{dynamic\_theory}}{\texttt{dynLog}} via Isabelle's \texttt{sublocale} mechanism.
Further, all axioms presented in the paper are proven.

\paragraph{Regular Programs.}
Similarly to the havoc lifting, the regular program closure lifting is implemented as a locale \texttt{\chg{S14}{regular\_lifting}{KATDynLog}} in the file \texttt{KATDynamicLogic.thy}.
In the same approach as above, this locale takes as only assumptions the instantiation of the locale \chg{S14}{\texttt{dynamic\_theory}}{\texttt{dynLog}} and proves that the regular closure over programs (with an appropriately defined program evaluation) is once again an instantiation of the locale \chg{S14}{\texttt{dynamic\_theory}}{\texttt{dynLog}} via Isabelle's \texttt{sublocale} mechanism.
Further, all axioms presented in the paper are proven.

\subsection{Combination}
As outlined in \Cref{sec:hdl}, the construction of a heterogeneous dynamic theory has three requirements:
Two existing dynamic theories and a set of new first-order atoms evaluating over the composed state (i.e., $\zero{\stateSpace}\times\one{\stateSpace}$).
In Isabelle, this is implemented as the locale \texttt{HeterogeneousDynLog} inside \texttt{HeterogeneousDynamicLogic.thy}:
This locale takes as input two instantiations of \chg{S14}{\texttt{dynamic\_theory}}{\texttt{dynLog}} and one instantiation of \texttt{folCore} whose state space is fixed to the combined state space of the individual theories.
The locale constructs the simple heterogeneous dynamic theory as described in \Cref{subsec:hdl:simple} and proves that it is once again an instantiation of \chg{S14}{\texttt{dynamic\_theory}}{\texttt{dynLog}}.
Using \texttt{\chg{S14}{havoc\_lifting}{HavocDynLog}} and \texttt{\chg{S14}{regular\_lifting}{KATDynLog}}, this heterogeneous instantiation is then lifted to the fully heterogeneous dynamic theory.
Additionally, we prove the reduction rules \ref{axiom:HR0}, \ref{axiom:HR1}.

\subsection{Inductive Expressiveness}
To additionally derive axiom \ref{axiom:C}, we require inductive expressiveness.
This is formalized in the locale \texttt{\chg{S14}{nat\_regular\_lifting}{NatKATDynLog}}, which, in addition to an instantiation of \texttt{\chg{S14}{regular\_lifting}{KATDynLog}} (required for the availability of loops) makes the assumptions formalized in \Cref{def:liftedDL:inductive_expressive}.
It then proves the soundness of axiom \ref{axiom:C} for an arbitrary given regular program lifted dynamic theory.
In \texttt{HeterogeneousDynLog} we also prove that if a given homogeneous dynamic theory $\ith{\dynamicTheory}$ is inductively expressive (i.e. instantiates \texttt{\chg{S14}{nat\_regular\_lifting}{NatKATDynLog}}), then the fully heterogeneous dynamic theory is also inductively expressive w.r.t $\ith{\integerVariables}$ (i.e. instantiates \texttt{\chg{S14}{nat\_regular\_lifting}{NatKATDynLog}}).

\subsection{Equational reasoning}
We introduce the fundamental notions of program equivalence and refinement in the theory \texttt{DLRewriting} which extends \chg{S14}{\texttt{dynamic\_theory}}{\texttt{dynLog}}.
For dynamic logics with regular program closure the theory \texttt{KATRewriting} introduces the locale \texttt{KleeneDynLog} which proves the Kleene Algebra with Tests axioms and auxilliary lemmas.
\texttt{HeterogeneousRewriting} formalizes the KAT rewriting results for heterogeneous theory combination.

\paragraph{Note on assumptions of \texttt{KATRewriting}.}
For clarity this paper introduces dynamic logics as requiring the extensionality axiom required by \texttt{KATRewriting}.
In reality, \chg{S14}{\texttt{dynamic\_theory}}{\texttt{dynLog}} itself requires a slightly weaker notion of extensionality which we prove to be weaker in \texttt{ExtensionalDynamicLogic}.
Hence, all results of this paper are in particular derivable for \new{S14}{the }definition of dynamic theory presented here while also holding for a slightly weaker notion of extensionality formalized in \chg{S14}{\texttt{dynamic\_theory}}{\texttt{dynLog}}.

\subsection{Instantiations}
So far we have proven the following instantiations of provided locales:
\begin{itemize}
    \item Our formalization of Propositional Dynamic Logic is an instantiation of \chg{S14}{\texttt{dynamic\_theory}}{\texttt{dynLog}} (this instantiation implemented its own regular programs)
    \item Our formalization of semi-ring first-order logic and DL (solely with assignment) instantiates:
    \begin{itemize}
        \item \texttt{folCore}
        \item \chg{S14}{\texttt{dynamic\_theory}}{\texttt{dynLog}}
        \item \texttt{\chg{S14}{havoc\_lifting}{HavocDynLog}}
        \item \texttt{\chg{S14}{regular\_lifting}{KATDynLog}}
        \item \texttt{\chg{S14}{nat\_regular\_lifting}{NatKATDynLog}} for \texttt{nat} or \texttt{int}
    \end{itemize}
    \item The AFP's formalization of differential dynamic logic~\cite{Differential_Dynamic_Logic-AFP} instantiates \chg{S14}{\texttt{dynamic\_theory}}{\texttt{dynLog}}
    \item Two instantiations of semi-ring DL and a semi-ring first-order logic over corresponding state tuples instantiates \texttt{HeterogeneousDynLog} (sanity check for assumptions of HDL).
\end{itemize}

\subsection{Artifact}
We provide Isabelle outlines for our main project as well as our instantiation with Differential Dynamic Logic, along with the Isabelle files in the supplementary materials.

\clearpage
\subsection{Reference to Isabelle Proofs}
\label{apx:isabelle_reference}
\begin{tabular}{l|l|l}
    \textbf{Result} & \textbf{Isabelle Reference} & \textbf{File} \\\hline\hline
    \Cref{isaLem:fv_soundness} & \makecell[l]{\texttt{\chg{S14}{dynamic\_theory}{dynLogCore}.dyn\_FV\_fml\_sound},\\ \texttt{\chg{S14}{dynamic\_theory}{dynLogCore}.dyn\_FV\_prog\_sound}} & \texttt{DynamicLogicCore.thy} \\\hline
    \Cref{isaLem:pdl_dyn_theory} & \texttt{pdlDLFull} & \texttt{PDL.thy}\\\hline
    \Cref{lem:DL:semiring_dl} & \texttt{semiring\_dl\_full} & \texttt{SemiRingDL.thy}\\\hline
    \Cref{lem:DL:diffDL_theory} & \texttt{dl\_dyn\_core} & \texttt{DifferentialDL.thy}\\\hline
    \Cref{isaLem:coincidence_formulas} & \makecell[l]{\texttt{dyn\_fv\_sem\_coincidence},\\\texttt{dyn\_fv\_sem\_coincidence\_smallest}}  & \texttt{DynamicLogicCore.thy}\\\hline
    \Cref{isaLem:coincidence_programs} & \makecell[l]{
    \texttt{dyn\_prog\_fv\_sem\_coincidence},\\
    \texttt{dyn\_prog\_fv\_sem\_coincidence\_smallest}
    } & \texttt{DynamicLogicCore.thy}\\\hline
    \Cref{isaLem:bounded_effect_progams} & 
    \makecell[l]{\texttt{dyn\_prog\_bv\_sem\_bounded},\\
    \texttt{dyn\_prog\_bv\_sem\_bounded\_smallest1},\\
    \texttt{dyn\_prog\_bv\_sem\_bounded\_smallest2}}
    & \texttt{DynamicLogicCore.thy}\\\hline
    \Cref{isaLem:fv_is_overapprox} & 
        \texttt{\chg{S14}{dynamic\_theory}{dynLogCore}.dyn\_FV\_fml\_sound}
    & \texttt{DynamicLogicCore.thy}\\\hline
    \Cref{thm:elementary:soundness} & 
    \makecell[l]{
    \texttt{\chg{S14}{dynamic\_theory}{dynLogCore}.dyn\_axiom\_G},\\
    \texttt{\chg{S14}{dynamic\_theory}{dynLogCore}.dyn\_axiom\_K},\\
    \texttt{\chg{S14}{dynamic\_theory}{dynLog}.dyn\_axiom\_V},\\
    \texttt{\chg{S14}{dynamic\_theory}{dynLog}.dyn\_axiom\_B}\\
    \texttt{\chg{S14}{dynamic\_theory}{dynLog}.rewriteAxiomReplace}}
    & \makecell[l]{
    \texttt{DynamicLogicCore.thy}\\~\\
    \texttt{DynamicLogic.thy}\\~\\
    \texttt{DLRewriting.thy}}\\\hline
    \Cref{thm:liftedDL:havoced} &
    \makecell[l]{
    \texttt{\chg{S14}{havoc\_lifting}{HavocDynLog}.havoc\_dl}
    }
    & \texttt{HavocDynamicLogic.thy}\\\hline
    \Cref{isaLem:liftedDL:havoc:reduction} &
    \makecell[l]{
    \texttt{\chg{S14}{havoc\_lifting}{HavocDynLog}.dyn\_axiom\_HR}
    } &
    \texttt{HavocDynamicLogic.thy}\\\hline
    \Cref{isalem:liftedDL:havoc:axiom} &
    \makecell[l]{
    \texttt{\chg{S14}{havoc\_lifting}{HavocDynLog}.dyn\_axiom\_havoc}
    } &
    \texttt{HavocDynamicLogic.thy}\\\hline
    \Cref{thm:liftedDL:regular} &
    \texttt{\chg{S14}{regular\_lifting}{KATDynLog}.kat\_dl} & \texttt{KATDynamicLogic.thy}\\\hline
    \Cref{isalem:liftedDL:rr_axiom} &
    \texttt{\chg{S14}{regular\_lifting}{KATDynLog}.dyn\_axiom\_RR} &
    \texttt{KATDynamicLogic.thy}\\\hline
    \Cref{isathm:soundess_regular_axioms} &
    \makecell[l]{
    \texttt{\chg{S14}{regular\_lifting}{KATDynLog}.dyn\_axiom\_test},\\
    \texttt{\chg{S14}{regular\_lifting}{KATDynLog}.dyn\_axiom\_seq},\\
    \texttt{\chg{S14}{regular\_lifting}{KATDynLog}.dyn\_axiom\_star},\\
    \texttt{\chg{S14}{regular\_lifting}{KATDynLog}.dyn\_axiom\_choice},\\
    \texttt{\chg{S14}{regular\_lifting}{KATDynLog}.dyn\_axiom\_I}\\
    } & \texttt{KATDynamicLogic.thy}\\\hline
    \Cref{isathm:reg_closure_kat} &
    \makecell[l]{\texttt{KleeneDynLog.kat\_rule\_}*\\
    \texttt{KleeneDynLog.bool\_rule}*\texttt{\_general}
    }&
    \texttt{KATRewriting.thy}
    \\\hline
    \Cref{isalem:correctness_check_to_fml} &
    \texttt{KleeneDynLog.eq\_if\_pure} &
    \texttt{KATRewriting}
    \\\hline
    \Cref{isalem:loop_convergence} & 
    \makecell[l]{
    \texttt{\chg{S14}{nat\_regular\_lifting}{NatKATDynLog}.dyn\_axiom\_C}
    } & \texttt{NatKATDynamicLogic.thy}\\\hline
    \Cref{lem:semiring_dl_inductive} &
    \makecell[l]{\texttt{semiring\_dl\_full\_nat},\\
    \texttt{semiring\_dl\_full\_nat\_int}}
    & \texttt{SemiRingDL.thy}
\end{tabular}

\begin{tabular}{l|l|l}
    \textbf{Result} & \textbf{Isabelle Reference} & \textbf{File} \\\hline\hline
    \Cref{isathm:simple_hdl} &
    \texttt{HeterogeneousDynLog.het\_dl} &
    \texttt{HeterogeneousDynamicLogic.thy}\\\hline
    \Cref{isathm:full_hdl} & 
    \makecell[l]{\texttt{HeterogeneousDynamicLogic.het\_dl\_havoc}\\
    \texttt{HeterogeneousDynamicLogic.het\_dl\_kat}} & 
    \texttt{HeterogeneousDynamicLogic.thy}\\\hline
    \Cref{isalem:hdl:reduction} &
    \makecell[l]{\texttt{HeterogeneousDynLog.dyn\_axiom\_R0},\\
    \texttt{HeterogeneousDynLog.dyn\_axiom\_R1}} &
    \texttt{HeterogeneousDynamicLogic.thy}\\\hline
    \Cref{lem:hdl:inductive} & 
    \makecell[l]{\texttt{inductively\_expressive\_1}\\
    \texttt{inductively\_expressive\_2}} &
    \texttt{HeterogeneousDynamicLogic.thy}\\\hline
    \Cref{isathm:heterogeneous_rewriting} &
    \makecell[l]{%
    \texttt{HeterogeneousRewriting.choice\_rewrite}\\
    \texttt{HeterogeneousRewriting.seq\_rewrite}\\
    \texttt{HeterogeneousRewriting.check\_rewrite}\\
    \texttt{HeterogeneousRewriting.loop\_rewrite}\\
    }&
    \texttt{HeterogeneousRewriting.thy}\\\hline
    \Cref{lem:rel-complte:havoc_BV} &
    \texttt{HavocDynamicLogic.BV\_overapprox}
    &\texttt{HavocDynamicLogic.thy}\\\hline
    \Cref{lem:rel-complte:regular_BV} &
    \texttt{KATDynamicLogic.BV\_overapprox} &
    \texttt{KATDynamicLogic.thy} \\\hline
    \Cref{lem:rel-complte:prehero-BV} &
    \texttt{HeterogeneousDynLog.BV\_overapprox} &
    \texttt{HeterogeneousDynamicLogic.thy}
\end{tabular}

\clearpage
\subsection{Reference to Isabelle Definitions}
\label{apx:isa_defs}
\begin{tabular}{l|l|l}
    \textbf{Result} & \textbf{Isabelle Reference} & \textbf{File} \\\hline\hline
    \Cref{def:DL:signature} &
    \makecell[l]{
    types of \texttt{folCore}\\
    types of \texttt{dynamic\_theory}
    } &
    \makecell[l]{
    \texttt{FirstOrder.thy}\\
    \texttt{DynamicLogicCore.thy}
    }\\\hline
    \Cref{def:DL:formulas} &
    \makecell[l]{
    \texttt{fol\_formula}\\
    \texttt{dyn\_formula}
    } &
    \makecell[l]{
    \texttt{FirstOrder.thy}\\
    \texttt{DynamicLogicCore.thy}
    }\\\hline
    \Cref{def:DL:universe_state_space_var_eval} &
    \texttt{folCore} &
    \texttt{FirstOrder.thy}\\\hline
    \Cref{def:DL:atom_eval} &
    \texttt{folCore} &
    \texttt{FirstOrder.thy}\\\hline
    \Cref{def:DL:prog_eval} &
    \texttt{dynamic\_theory} &
    \texttt{DynamicLogicCore.thy}\\\hline
    \Cref{def:DL:domain} &
    \texttt{dynamic\_theory} &
    \texttt{DynamicLogicCore.thy}\\\hline
    \Cref{def:DL:theory} &
    \texttt{dynamic\_theory} &
    \texttt{DynamicLogicCore.thy}\\\hline
    \Cref{def:DL:semantics} &
    \makecell[l]{
    \texttt{folCore.fol\_eval}\\
    \texttt{dynamic\_theory.dyn\_eval}
    }&
    \makecell[l]{
    \texttt{FirstOrder.thy}\\
    \texttt{DynamicLogicCore.thy}
    }\\\hline
    \Cref{def:DL:free_bound_vars} &
    \makecell[l]{
    \texttt{folCore.fol\_fv\_sem}\\
    \texttt{dynamic\_theory.dyn\_fv\_sem}\\
    \texttt{dynamic\_theory.dyn\_prog\_fv\_sem}\\
    \texttt{dynamic\_theory.dyn\_prog\_bv\_sem}
    } &
    \makecell[l]{
    \texttt{FirstOrder.thy}\\
    \texttt{DynamicLogicCore.thy}\\
    \texttt{DynamicLogicCore.thy}\\
    \texttt{DynamicLogicCore.thy}
    }\\\hline
    \Cref{def:DL:refinement_equivalence} &
    \makecell[l]{
    \texttt{dynamic\_theory.progRefine}\\
    \texttt{dynamic\_theory.progEquiv}
    } &
    \texttt{DLRewriting.thy}\\\hline
    \Cref{def:liftedDL:havoced} &
    \makecell[l]{
    \texttt{havoc\_prog}\\
    \texttt{havoc\_lifting.full\_prog\_eval}\\
    \texttt{havoc\_lifting.full\_prog\_FV}
    } &
    \texttt{HavocDynamicLogic.thy}\\\hline
    \Cref{def:liftedDL:regular_closure}&
    \makecell[l]{
    \texttt{kat\_prog}\\
    \texttt{regular\_lifting.full\_prog\_eval\_int}\\
    \texttt{regular\_lifting.full\_prog\_eval}
    } &
    \texttt{KATDynamicLogic.thy}\\\hline
    \Cref{def:liftedDL:freeVarsFOL} &
    \texttt{folCore.fol\_fv\_syn} &
    \texttt{FirstOrder.thy}\\\hline
    \Cref{def:liftedDL:programFV} &
    \texttt{regular\_lifting.full\_prog\_fv\_syn} &
    \texttt{KATDynamicLogic.thy}\\\hline
    \Cref{def:liftedDL:regular_lift} &
    \texttt{regular\_lifting} &
    \texttt{KATDynamicLogic.thy}\\\hline
    \Cref{def:liftedDL:inductive_expressive} &
    \texttt{nat\_regular\_lifting} &
    \texttt{NatKATDynamicLogic.thy}\\\hline
    \Cref{def:hdl:hero_state_space} &
    \makecell[l]{
    \textit{see \texttt{HeterogeneousDynLog}}\\
    \texttt{('state1 × 'state2)}
    } &
    \texttt{HeterogeneousDynamicLogic.thy}\\\hline
    \Cref{def:hero_atoms} &
    \makecell[l]{
    \texttt{HeterogeneousDynLog.atom\_eval\_comm}\\
    \texttt{HeterogeneousDynLog.atom\_fv\_syn\_comm}\\
    } &
    \texttt{HeterogeneousDynamicLogic.thy}\\\hline
    \Cref{def:hdl:simple_signature} &
    \makecell[l]{
    \texttt{HeterogeneousDynLog}
    } &
    \texttt{HeterogeneousDynamicLogic.thy}\\\hline
    \Cref{def:hdl:simple_dom_comp} &
    \makecell[l]{
    \texttt{HeterogeneousDynLog.het\_atom\_eval}\\
    \texttt{HeterogeneousDynLog.het\_prog\_eval}\\
    } &
    \texttt{HeterogeneousDynamicLogic.thy}\\\hline
    \Cref{def:free_vars_simple_hero} &
    \makecell[l]{
    \texttt{HeterogeneousDynLog.het\_atom\_fv\_syn}\\
    \texttt{HeterogeneousDynLog.het\_dyn\_prog\_FV}
    } &
    \texttt{HeterogeneousDynamicLogic.thy}\\\hline
    \Cref{def:hdl:fully_hero} &
    \texttt{HeterogeneousDynLog.het\_dl\_kat} &
    \texttt{HeterogeneousDynamicLogic.thy}
    
\end{tabular}
}{}

\end{document}